\documentclass[runningheads]{llncs}

\newcommand{\ponders}[3]{\ignorespaces}
\newcommand{\TODO}[1]{}

\usepackage{fancyvrb}
\usepackage{mathpartir}
\usepackage{bold-extra}
\usepackage{amsmath}
\usepackage{cite}

\DefineVerbatimEnvironment%
  {core}{Verbatim}
  {xleftmargin=\mathindent}

\usepackage{subcaption}
\usepackage{caption}
\usepackage{tikz}
\usepackage{tikz-cd}
\usepackage{sansmath}
\usepackage{mathtools}
\usepackage{quiver}
\usepackage{mathpartir}
\usepackage{hyperref}
\usepackage[capitalize,nameinlink,noabbrev]{cleveref} 

\usepackage{enumitem}
\usepackage{float}
\usepackage{bbding}
\usepackage{orcidlink}

\makeatletter
\@ifundefined{lhs2tex.lhs2tex.sty.read}%
  {\@namedef{lhs2tex.lhs2tex.sty.read}{}%
   \newcommand\SkipToFmtEnd{}%
   \newcommand\EndFmtInput{}%
   \long\def\SkipToFmtEnd#1\EndFmtInput{}%
  }\SkipToFmtEnd

\newcommand\ReadOnlyOnce[1]{\@ifundefined{#1}{\@namedef{#1}{}}\SkipToFmtEnd}
\usepackage{amstext}
\usepackage{amssymb}
\usepackage{stmaryrd}
\DeclareFontFamily{OT1}{cmtex}{}
\DeclareFontShape{OT1}{cmtex}{m}{n}
  {<5><6><7><8>cmtex8
   <9>cmtex9
   <10><10.95><12><14.4><17.28><20.74><24.88>cmtex10}{}
\DeclareFontShape{OT1}{cmtex}{m}{it}
  {<-> ssub * cmtt/m/it}{}

\DeclareFontShape{OT1}{cmtt}{bx}{n}
  {<5><6><7><8>cmtt8
   <9>cmbtt9
   <10><10.95><12><14.4><17.28><20.74><24.88>cmbtt10}{}
\DeclareFontShape{OT1}{cmtex}{bx}{n}
  {<-> ssub * cmtt/bx/n}{}

\newcommand{\Conid}[1]{\mathit{#1}}
\newcommand{\Varid}[1]{\mathit{#1}}
\newcommand{\anonymous}{\kern0.06em \vbox{\hrule\@width.5em}}

\newcommand{\bind}{\mathbin{>\!\!\!>\mkern-6.7mu=}}

\usepackage{polytable}

\@ifundefined{mathindent}%
  {\newdimen\mathindent\mathindent\leftmargini}%
  {}%

\def\resethooks{%
  \global\let\SaveRestoreHook\empty
  \global\let\ColumnHook\empty}
\newcommand*{\savecolumns}[1][default]%
  {\g@addto@macro\SaveRestoreHook{\savecolumns[#1]}}
\newcommand*{\restorecolumns}[1][default]%
  {\g@addto@macro\SaveRestoreHook{\restorecolumns[#1]}}
\newcommand*{\aligncolumn}[2]%
  {\g@addto@macro\ColumnHook{\column{#1}{#2}}}

\resethooks

\newcommand{\onelinecommentchars}{\quad-{}- }
\newcommand{\commentbeginchars}{\enskip\{-}
\newcommand{\commentendchars}{-\}\enskip}

\newcommand{\visiblecomments}{%
  \let\onelinecomment=\onelinecommentchars
  \let\commentbegin=\commentbeginchars
  \let\commentend=\commentendchars}

\newcommand{\invisiblecomments}{%
  \let\onelinecomment=\empty
  \let\commentbegin=\empty
  \let\commentend=\empty}

\visiblecomments

\newlength{\blanklineskip}
\newcommand{\hsindent}[1]{\quad}%
\let\hspre\empty
\let\hspost\empty

\EndFmtInput
\makeatother
\ReadOnlyOnce{polycode.fmt}%
\makeatletter

\newcommand{\hsnewpar}[1]%
  {{\parskip=0pt\parindent=0pt\par\vskip #1\noindent}}

\newcommand{\hscodestyle}{}

\newcommand{\sethscode}[1]%
  {\expandafter\let\expandafter\hscode\csname #1\endcsname
   \expandafter\let\expandafter\endhscode\csname end#1\endcsname}

  {\par\noindent
   \advance\leftskip\mathindent
   \hscodestyle
   \let\\=\@normalcr
   \let\hspre\(\let\hspost\)%
   \pboxed}%
  {\endpboxed\)%
   \par\noindent
   \ignorespacesafterend}

  {\hsnewpar\abovedisplayskip
   \advance\leftskip\mathindent
   \hscodestyle
   \let\hspre\(\let\hspost\)%
   \pboxed}%
  {\endpboxed%
   \hsnewpar\belowdisplayskip
   \ignorespacesafterend}

  {\hsnewpar\abovedisplayskip
   \advance\leftskip\mathindent
   \hscodestyle
   \let\\=\@normalcr
   \(\pboxed}%
  {\endpboxed\)%
   \hsnewpar\belowdisplayskip
   \ignorespacesafterend}

\newcommand{\plainhs}{\sethscode{plainhscode}}

\plainhs

  {\hsnewpar\abovedisplayskip
   \advance\leftskip\mathindent
   \hscodestyle
   \let\\=\@normalcr
   \(\parray}%
  {\endparray\)%
   \hsnewpar\belowdisplayskip
   \ignorespacesafterend}

  {\parray}{\endparray}

  {\(\parray}{\endparray\)}

\def\codeframewidth{\arrayrulewidth}
\RequirePackage{calc}

  {\parskip=\abovedisplayskip\par\noindent
   \hscodestyle
   \arrayrulewidth=\codeframewidth
   \tabular{@{}|p{\linewidth-2\arraycolsep-2\arrayrulewidth-2pt}|@{}}%
   \hline\framedhslinecorrect\\{-1.5ex}%
   \let\endoflinesave=\\
   \let\\=\@normalcr
   \(\pboxed}%
  {\endpboxed\)%
   \framedhslinecorrect\endoflinesave{.5ex}\hline
   \endtabular
   \parskip=\belowdisplayskip\par\noindent
   \ignorespacesafterend}

\newcommand{\framedhslinecorrect}[2]%
  {#1[#2]}

  {\(\def\column##1##2{}%
   \let\>\undefined\let\<\undefined\let\\\undefined
   \newcommand\>[1][]{}\newcommand\<[1][]{}\newcommand\\[1][]{}%
   \def\fromto##1##2##3{##3}%
   }{\) }%

  {\let\orighscode=\hscode
   \let\origendhscode=\endhscode
   \def\endhscode{\def\hscode{\endgroup\def\@currenvir{hscode}\\}\begingroup}
   \orighscode\def\hscode{\endgroup\def\@currenvir{hscode}}}%
  {\origendhscode
   \global\let\hscode=\orighscode
   \global\let\endhscode=\origendhscode}%

\makeatother
\EndFmtInput
\ReadOnlyOnce{forall.fmt}%
\makeatletter

\let\HaskellResetHook\empty
\newcommand*{\AtHaskellReset}[1]{%
  \g@addto@macro\HaskellResetHook{#1}}
\newcommand*{\HaskellReset}{\HaskellResetHook}

\newcommand\hsforall{\global\let\hsdot=\hsperiodonce}
\newcommand*\hsperiodonce[2]{#2\global\let\hsdot=\hscompose}
\newcommand*\hscompose[2]{#1}

\AtHaskellReset{\global\let\hsdot=\hscompose}

\HaskellReset

\makeatother
\EndFmtInput

\renewcommand{\commentbegin}{\enskip\enskip\{\enskip}
\renewcommand{\commentend}{\}}

\newcommand{\Set}{\texttt{Set}}

\newcommand{\CatC}{\mathbb{C}}
\newcommand{\CatD}{\mathbb{D}}

\newcommand{\CatEndo}{\CatC^\CatC}
\newcommand{\CatEndof}{\texttt{\itshape Endo}_f(\CatC)}

\newcommand{\CatEC}{\CatEndof\times\CatC}
\newcommand{\abracket}[1]{\langle #1 \rangle}
\newcommand{\tuple}[1]{\langle #1 \rangle}
\newcommand{\CatRes}[1][E]{\texttt{\itshape Res}(#1)}

\newcommand{\Id}{\texttt{\itshape Id}}

\newcommand{\previous}{\triangleleft}
\newcommand{\later}{\triangleright}
\newcommand{\Nat}{\mathbb{N}}

\newcommand{\ladjunct}[1]{\lceil #1 \rceil}

\newcommand{\radjunct}[1]{\lfloor #1 \rfloor}

\newcommand{\identity}{\mathit{id}}

\newcommand{\counit}{\epsilon}

\newcommand{\Ran}[1]{\texttt{Ran}_{#1}}

\newcommand{\Lan}[1]{\texttt{Lan}_{#1}}
\newcommand{\Free}{\texttt{\itshape Free}}
\newcommand{\Ul}{\texttt{\itshape U}}
\newcommand{\CatCN}{\CatC^{\vert\Nat\vert}}

\newcommand{\CatN}{\vert\Nat\vert}

\newcommand{\Alg}{\texttt{\itshape-Alg}}

\newcommand{\Ix}{\texttt{\itshape Ix}} 
\newcommand{\Fn}{\texttt{\itshape Fn}}
\newcommand{\FFn}{\widehat{\texttt{\itshape Free}}_\Fn}
\newcommand{\KL}{K^{{\scriptscriptstyle{\texttt{EM}}}}_{\scriptscriptstyle{\texttt{Fn}}}}
\newcommand{\KK}{K^{{\scriptscriptstyle{\texttt{Fn}}}}_{\scriptscriptstyle{\texttt{Ix}}}}
\newcommand{\Kb}[1][Ix]{K^{\scriptscriptstyle{\texttt{#1}}}_{\scriptscriptstyle{\texttt{EM}}}}

\newcommand{\lfix}[2]{\mu #1.\;#2}
\newcommand{\iso}{\cong}
\newcommand{\iniso}{\ensuremath{\mathit{in}}}
\newcommand{\inOp}{\ensuremath{\mathit{in}^{\circ}}}
\newcommand{\vcomp}{\mathop{\cdot}}
\newcommand{\hcomp}{\mathop{\circ}}
\newcommand{\upC}{\mathop{\uparrow}}
\newcommand{\downC}{\mathop{\downarrow}}
\newcommand{\interp}[2]{\textit{handle}_{#1}\;#2}
\newcommand{\interpName}{\textit{handle}}

\newcommand{\op}{\mathit{op}}

\newcommand{\cata}[1]{\llparenthesis #1 \rrparenthesis}

\newcommand{\varr}{\mathit{var}}

\let\llncsproof\proof
\renewcommand{\proof}[1][]{%
  \ifx!#1!\else\renewcommand{\proofname}{#1}\fi
  \llncsproof
}

\spnewtheorem*{corollary*}{Corollary}{\normalfont\bfseries}{\itshape}
\spnewtheorem*{theorem*}{Theorem}{\normalfont\bfseries}{\itshape}
\spnewtheorem*{lemma*}{Lemma}{\normalfont\bfseries}{\itshape}

\newcommand{\GADTs}{\textsc{gadt}s}
\newcommand{\aset}[1]{\{\,#1\,\}}

\newcommand{\refapp}[1]{Appendix~\ref{#1}}

\newenvironment{codepage}[1][.9\linewidth]
  { 
    \begin{minipage}{#1}
    \vspace{-\abovedisplayskip} }
  { \vspace{-\belowdisplayskip}\vspace{-0.7\baselineskip}
    \end{minipage} }

\newenvironment{nscenter}
 {\parskip=0pt\par\centering\vspace{-0.2\abovedisplayskip}}
 {\par\noindent\ignorespacesafterend\vspace{-0.2\belowdisplayskip}}

\begin{document}

\title{Structured Handling of Scoped Effects: Extended Version}

\authorrunning{Yang et al.}

\author{
 Zhixuan Yang \inst{1} {\href{mailto:s.yang20@imperial.ac.uk}{\Envelope}}\ \orcidlink{0000-0001-5573-3357} \and
 Marco Paviotti \inst{1}\orcidlink{0000-0002-1513-0807} \and
 Nicolas Wu \inst{1}\orcidlink{0000-0002-4161-985X} \and
 Birthe {van den Berg}\inst{2}\orcidlink{0000-0002-0088-9546} \and
 Tom Schrijvers \inst{2}\orcidlink{0000-0001-8771-5559}
}
\institute{
Imperial College London, London, United Kingdom \\ \email{\{s.yang20,m.paviotti,n.wu\}@imperial.ac.uk} \and
KU Leuven, Leuven, Belgium \\\email{\{birthe.vandenberg,tom.schrijvers\}@kuleuven.be} 
}

\maketitle

\begin{abstract}

Algebraic effects offer a versatile framework that covers a wide
variety of effects.  However, the family of operations that
delimit scopes are not algebraic and are usually modelled as handlers, 
thus preventing them from being used freely in conjunction with algebraic
operations. Although proposals for scoped operations exist,
they are either ad-hoc and unprincipled, or too inconvenient for practical
programming. This paper provides the best of both worlds: a
theoretically-founded model of scoped effects that is convenient for
implementation and reasoning. 
Our new model is based on an adjunction between a locally finitely presentable
category and a category of \emph{functorial algebras}.
Using comparison functors between adjunctions, we show that our new model, an
existing indexed model, and a third approach that simulates scoped operations in
terms of algebraic ones have equal expressivity for handling scoped operations. 
We consider our new model to be the sweet spot between ease of implementation
and structuredness.
Additionally, our approach automatically induces fusion laws of handlers of
scoped effects, which are useful for reasoning and optimisation.


\end{abstract}

\section{Introduction}

For a long time, monads~\cite{Moggi95,Spivey1990,Wadler95} have been the go-to
approach for purely functional modelling of and programming with side effects.
However, in recent years an alternative approach, \emph{algebraic effects}
\cite{fossacs/PlotkinP02}, is gaining more traction.  A big breakthrough has been the
introduction of \emph{handlers} \cite{PlotPret13Hand}, which has made
algebraic effects suitable for programming and has led to numerous dedicated
languages and libraries implementing algebraic effects and handlers.
In comparison to monads, algebraic effects provide a more modular approach to
computations with effects, in which the syntax and semantics of effects are
separated---computations invoking algebraic operations can be defined
syntactically, and the semantics of operations are given by handlers separately
in possibly many ways.

A disadvantage of algebraic effects is that they are less expressive than
monads; not all effects can be easily expressed or composed within their
confines.
For instance, operations like \ensuremath{\Varid{catch}} for exception handling, \ensuremath{\Varid{spawn}} for
parallel composition of processes, or \ensuremath{\Varid{once}} for restricting nondeterminism are
not conventional algebraic operations; instead they delimit a computation
within their scope.
Such operations are usually modelled as handlers, but the problem is that they
cannot be freely used amongst other algebraic operations:
when a handler implementing a scoped operation is applied to a
computation, the computation is transformed from a syntactic tree of algebraic
operations into some semantic model implementing the scoped operation. Consequently, all 
subsequent operations on the computation can only be given in the particular 
semantic model rather than as mere syntactic operations,
thus nullifying the crucial advantage of modularity when separating syntax
and semantics of effects.

To remedy the situation, Wu et al.\ \cite{WuSH14} proposed a practical, but ad-hoc,
generalization of algebraic effects in Haskell that encompasses scoped effects,
that has been adopted by several algebraic effects libraries~\cite{fused-effects,polysemy,eff}.
More recently, Pir\'{o}g et al.\ \cite{PirogSWJ18} sought to put this ad-hoc approach for scoped
effects on the same formal footing as algebraic effects.
Their solution resulted in a construction based on a level-indexed category,
called \emph{indexed algebras}, as the way to give semantics to scoped effects.
However, this formalisation introduces a disparity between syntax and semantics
that makes indexed algebras not as {structured} as the programs they interpret,
where they use an ad-hoc hybrid fold that requires indexing for the handlers,
but not for the program syntax.
Moreover, indexed algebras are not ideal for widespread implementation as they
require dependent typing, in at least a limited form like
\GADTs~\cite{JohannG08}.

This paper presents a more structured way of handling scoped effects, which we
call \emph{functorial algebras}.
They are principled and formally grounded on category theory, and at
the same time more structured than the indexed algebras of
Pir\'{o}g et al.\ \cite{PirogSWJ18}, in the sense that the structure of functorial algebras
directly follows the abstract syntax of programs with scoped effects.
Additionally, our approach can be practically implemented without the need for
dependent types or \GADTs, making it available for a wider range of programming
languages.  In particular, we make the following contributions:

\begin{itemize}[topsep=0pt]
\item We highlight the loss of modularity when modelling scoped operations as
handlers and sketch how the problem is solved using functorial algebras in
Haskell, along with a number of programming examples
(\autoref{sec:working});

\item We develop a category-theoretic foundation of {functorial algebras}
as a notion of handlers of scoped effects.  Specifically, we show that there is
an adjunction between functorial algebras and a base category, inducing the
monad modelling the syntax of scoped effects (\autoref{sec:foundation});

\item We show that the expressivity of functorial algebras, Pir\'{o}g et al.\ \cite{PirogSWJ18}'s
indexed algebras, and simulating scoped effects with algebraic operations and
recursion are equal, by constructing interpretation-preserving functors
between the three categories of algebras (\autoref{sec:comparing});

\item We present the fusion law of functorial algebras, which is useful for
reasoning and optimisation. 
The fusion law directly follows from the naturality of the
adjunction underlying functorial algebras (\autoref{sec:hybrid-adjoint}).
\end{itemize}

\noindent
Finally, we discuss related work (\autoref{sec:related}) and conclude (\autoref{sec:conclusion}).

\section{Scoped Effects for the Working Programmer}
\label{sec:working}
We start with a recap of \emph{handlers of algebraic effects} 
(\autoref{sec:alg}), and then we highlight the loss of modularity when
modelling non-algebraic effectful operations as handlers
(\autoref{sec:problem}).
We then show how the problem is solved by modelling them as \emph{scoped
operations} and handling them with \emph{functorial algebras} in Haskell
(\autoref{sec:overview}), whose categorical foundation will be developed later.

\subsection{Handlers of Algebraic Effects} \label{sec:alg}

For the purpose of demonstration, in this section we base our discussion on a
simplistic implementation of effect handlers in Haskell using \emph{free
monads}, although the problem with effect handlers highlighted in this section
applies to other more practical implementations of effect handlers, either as
libraries (e.g.\ \cite{Oleg15Freer,Kammar13}) or standalone languages (e.g.\
\cite{BP14Eff,Koka17,Lin17DoBe}).

Following Plotkin and Pretnar \cite{PlotPret13Hand}, computational effects, such as
exceptions, mutable state, and nondeterminism, are described by
\emph{signatures} of primitive effectful operations.
Signatures can be abstractly represented by Haskell functors:
\[\ensuremath{\mathbf{class}\;\Conid{Functor}\;\Varid{f}\;\mathbf{where}\;\Varid{fmap}\mathbin{::}(\Varid{a}\to \Varid{b})\to \Varid{f}\;\Varid{a}\to \Varid{f}\;\Varid{b}}\]
The following functor \ensuremath{\Conid{ES}} (with the evident \ensuremath{\Conid{Functor}} instance) 
is the signature of three operations: throwing an exception, writing and reading
an \ensuremath{\Conid{Int}}-state:
\begin{equation}\label{eq:sig:es}
\ensuremath{\mathbf{data}\;\Varid{ES}\;\Varid{x}\mathrel{=}\Conid{Throw}\mid \Conid{Put}\;\Conid{Int}\;\Varid{x}\mid \Conid{Get}\;(\Conid{Int}\to \Varid{x})}
\end{equation}
Typically, a constructor of a signature functor \ensuremath{\Sigma} has a type
isomorphic to \ensuremath{\Conid{P}\to (\Conid{R}\to \Varid{x})\to \Sigma\;\Varid{x}} for some types \ensuremath{\Conid{P}} and \ensuremath{\Conid{R}}.
As in (\ref{eq:sig:es}), the types of the three constructors are isomorphic to 
\ensuremath{\Conid{Throw}\mathbin{::}()\to (\Conid{Void}\to \Varid{x})\to \Varid{ES}\;\Varid{x}}, \ensuremath{\Conid{Put}\mathbin{::}\Conid{Int}\to (()\to \Varid{x})\to \Varid{ES}\;\Varid{x}} 
and \ensuremath{\Conid{Get}\mathbin{::}()\to (\Conid{Int}\to \Varid{x})\to \Varid{ES}\;\Varid{x}} respectively where \ensuremath{\Conid{Void}} is the empty type.
Each constructor of a signature functor \ensuremath{\Sigma} is thought of as an \emph{operation}
that takes a parameter of type \ensuremath{\Conid{P}} and produces a result of type \ensuremath{\Conid{R}}, or
equivalently, has \ensuremath{\Conid{R}}-many possible ways to continue the computation after the
operation.
Given any (signature) functor \ensuremath{\Sigma}, computations invoking operations from
\ensuremath{\Sigma} are modelled by the following datatype, called the \emph{free monad}
of \ensuremath{\Sigma},
\begin{nscenter}\begin{hscode}\SaveRestoreHook
\column{B}{@{}>{\hspre}l<{\hspost}@{}}%
\column{E}{@{}>{\hspre}l<{\hspost}@{}}%
\>[B]{}\mathbf{data}\;\Conid{Free}\;\Sigma\;\Varid{a}\mathrel{=}\Conid{Return}\;\Varid{a}\mid \Conid{Call}\;(\Sigma\;(\Conid{Free}\;\Sigma\;\Varid{a})){}\<[E]%
\ColumnHook
 \end{hscode}\resethooks
\end{nscenter}
whose first case represents a computation that just \emph{return}s a value,
and the second case represents a computation \emph{call}ing an operation from 
\ensuremath{\Sigma} with more \ensuremath{\Conid{Free}\;\Sigma\;\Varid{a}} subterms as arguments, which are
understood as the continuation of the computation after this call, depending on
the outcome of this operation.

The inductive datatype \ensuremath{\Conid{Free}\;\Sigma\;\Varid{a}} comes with a \emph{recursion principle}:
\begin{nscenter}\begin{hscode}\SaveRestoreHook
\column{B}{@{}>{\hspre}l<{\hspost}@{}}%
\column{26}{@{}>{\hspre}l<{\hspost}@{}}%
\column{E}{@{}>{\hspre}l<{\hspost}@{}}%
\>[B]{}\Varid{handle}\mathbin{::}(\Sigma\;\Varid{b}\to \Varid{b})\to (\Varid{a}\to \Varid{b})\to \Conid{Free}\;\Sigma\;\Varid{a}\to \Varid{b}{}\<[E]%
\\
\>[B]{}\Varid{handle}\;\Varid{alg}\;\Varid{g}\;(\Conid{Return}\;\Varid{x}){}\<[26]%
\>[26]{}\mathrel{=}\Varid{g}\;\Varid{x}{}\<[E]%
\\
\>[B]{}\Varid{handle}\;\Varid{alg}\;\Varid{g}\;(\Conid{Call}\;\Varid{op}){}\<[26]%
\>[26]{}\mathrel{=}\Varid{alg}\;(\Varid{fmap}\;(\Varid{handle}\;\Varid{alg}\;\Varid{g})\;\Varid{op}){}\<[E]%
\ColumnHook
 \end{hscode}\resethooks
\end{nscenter}
which folds a tree of operations \ensuremath{\Conid{Free}\;\Sigma\;\Varid{a}} into a type \ensuremath{\Varid{b}}, providing
a way \ensuremath{\Sigma\;\Varid{b}\to \Varid{b}}, usually called a \emph{\ensuremath{\Sigma}-algebra}, to perform
operations from \ensuremath{\Sigma} on \ensuremath{\Varid{b}} and a way \ensuremath{\Varid{a}\to \Varid{b}} to transform the returned
type \ensuremath{\Varid{a}} of computations to \ensuremath{\Varid{b}}.
The function \ensuremath{\Varid{handle}} can be used to give \ensuremath{\Conid{Free}\;\Sigma} a monad instance:
\begin{equation*}
\begin{codepage}[0.30\linewidth]
\begin{nscenter}\begin{hscode}\SaveRestoreHook
\column{B}{@{}>{\hspre}l<{\hspost}@{}}%
\column{E}{@{}>{\hspre}l<{\hspost}@{}}%
\>[B]{}\Varid{return}\mathbin{::}\Varid{a}\to \Conid{Free}\;\Sigma\;\Varid{a}{}\<[E]%
\\
\>[B]{}\Varid{return}\mathrel{=}\Conid{Return}{}\<[E]%
\ColumnHook
 \end{hscode}\resethooks
\end{nscenter}
\end{codepage}
\ 
\begin{codepage}[0.70\linewidth]
\begin{nscenter}\begin{hscode}\SaveRestoreHook
\column{B}{@{}>{\hspre}l<{\hspost}@{}}%
\column{E}{@{}>{\hspre}l<{\hspost}@{}}%
\>[B]{}(\bind )\mathbin{::}\Conid{Free}\;\Sigma\;\Varid{a}\to (\Varid{a}\to \Conid{Free}\;\Sigma\;\Varid{b})\to \Conid{Free}\;\Sigma\;\Varid{b}{}\<[E]%
\\
\>[B]{}\Varid{m}\bind \Varid{k}\mathrel{=}\Varid{handle}\;\Conid{Call}\;\Varid{k}\;\Varid{m}{}\<[E]%
\ColumnHook
 \end{hscode}\resethooks
\end{nscenter}
\end{codepage}
\label{eq:freeMonadInst}
\end{equation*}
The monadic instance allows the programmer to build effectful computations using
the \ensuremath{\mathbf{do}}-notation in a clean way.
For example, the following program updates the state \ensuremath{\Varid{s}} to \ensuremath{\Varid{n}\mathbin{/}\Varid{s}} for some \ensuremath{\Varid{n}\mathbin{::}\Conid{Int}}, and throws an exception when \ensuremath{\Varid{s}} is \ensuremath{\mathrm{0}}:
\begin{equation*}
\setlength\mathindent{0pt}
\begin{codepage}[\linewidth]
\begin{nscenter}\begin{hscode}\SaveRestoreHook
\column{B}{@{}>{\hspre}l<{\hspost}@{}}%
\column{20}{@{}>{\hspre}l<{\hspost}@{}}%
\column{41}{@{}>{\hspre}l<{\hspost}@{}}%
\column{E}{@{}>{\hspre}l<{\hspost}@{}}%
\>[B]{}\Varid{safeDiv}\mathbin{::}\Conid{Int}\to \Conid{Free}\;\Varid{ES}\;\Conid{Int}{}\<[E]%
\\
\>[B]{}\Varid{safeDiv}\;\Varid{n}\mathrel{=}\mathbf{do}\;{}\<[20]%
\>[20]{}\Varid{s}\leftarrow \Varid{get};\mathbf{if}\;\Varid{s}\equiv \mathrm{0}\;{}\<[41]%
\>[41]{}\mathbf{then}\;\Conid{Call}\;\Conid{Throw}{}\<[E]%
\\
\>[41]{}\mathbf{else}\;\mathbf{do}\;\{\mskip1.5mu \Varid{put}\;(\Varid{n}\mathbin{/}\Varid{s});\Varid{return}\;(\Varid{n}\mathbin{/}\Varid{s})\mskip1.5mu\}{}\<[E]%
\ColumnHook
 \end{hscode}\resethooks
\end{nscenter}
\end{codepage}
\end{equation*}
where the auxiliary wrapper functions (the so-called \emph{smart constructors}
in the Haskell community) that invoke \ensuremath{\Conid{Call}} appropriately are
\[
\ensuremath{\Varid{get}} = \ensuremath{\Conid{Call}\;(\Conid{Get}\;\Conid{Return})} \hspace{1cm} \ensuremath{\Varid{put}\;\Varid{n}} = \ensuremath{\Conid{Call}\;(\Conid{Put}\;\Varid{n}\;(\Conid{Return}\;()))}
\]

The free monad merely models effectful computations \emph{syntactically}
without specifying how these operations are actually implemented.
Indeed, the program \ensuremath{\Varid{safeDiv}} above is defined without saying how mutable
state and exceptions are implemented at all.
To actually give useful semantics to programs built with free monads, the
programmer uses the \ensuremath{\Varid{handle}} function above to interpret programs with
\ensuremath{\Sigma}-algebras, which are called \emph{handlers} in this context.

For example, given a program \ensuremath{\Varid{r}\mathbin{::}\Conid{Free}\;\Varid{ES}\;\Varid{a}} for some \ensuremath{\Varid{a}}, 
a handler \ensuremath{\Varid{catchHdl}\;\Varid{r}\mathbin{::}\Conid{ES}\;(\Conid{Free}\;\Varid{ES})\to \Conid{Free}\;\Varid{ES}} that gives
the usual semantics to \ensuremath{\Varid{throw}} is
\begin{equation}\label{eq:catchHdl}
\begin{codepage}
\begin{nscenter}\begin{hscode}\SaveRestoreHook
\column{B}{@{}>{\hspre}l<{\hspost}@{}}%
\column{58}{@{}>{\hspre}l<{\hspost}@{}}%
\column{E}{@{}>{\hspre}l<{\hspost}@{}}%
\>[B]{}\Varid{catchHdl}\mathbin{::}\Conid{Free}\;\Varid{ES}\;\Varid{a}\to \Varid{ES}\;(\Conid{Free}\;\Varid{ES}\;\Varid{a})\to \Conid{Free}\;\Varid{ES}\;\Varid{a}{}\<[E]%
\\
\>[B]{}\Varid{catchHdl}\;\Varid{r}\;\Conid{Throw}\mathrel{=}\Varid{r}; \hspace{2em}\Varid{catchHdl}\;\Varid{r}\;\Varid{op}{}\<[58]%
\>[58]{}\mathrel{=}\Conid{Call}\;\Varid{op}{}\<[E]%
\ColumnHook
 \end{hscode}\resethooks
\end{nscenter}
\end{codepage}
\end{equation}
which evaluates \ensuremath{\Varid{r}} for \ensuremath{\Varid{r}}ecovery in case of throwing an exception, and leaves
other operations untouched in the free monad.
An important advantage of the approach of effect handlers is that different
semantics of a computational effect can be given by different handlers.
For example, suppose that in some scenario one would like to interpret exceptions
as unrecoverable errors and stop the execution of the program when an exception
is raised.
Then the following handler can be defined for this behaviour:
\begin{equation}\label{eq:catchHdl2}
\setlength\mathindent{0pt}
\begin{codepage}
\begin{nscenter}\begin{hscode}\SaveRestoreHook
\column{B}{@{}>{\hspre}l<{\hspost}@{}}%
\column{12}{@{}>{\hspre}l<{\hspost}@{}}%
\column{20}{@{}>{\hspre}l<{\hspost}@{}}%
\column{74}{@{}>{\hspre}l<{\hspost}@{}}%
\column{E}{@{}>{\hspre}l<{\hspost}@{}}%
\>[B]{}\Varid{catchHdl'}{}\<[12]%
\>[12]{}\mathbin{::}\Conid{Free}\;\Varid{ES}\;\Varid{a}\to \Varid{ES}\;(\Conid{Free}\;\Varid{ES}\;(\Conid{Maybe}\;\Varid{a}))\to \Conid{Free}\;\Varid{ES}\;(\Conid{Maybe}\;\Varid{a}){}\<[E]%
\\
\>[B]{}\Varid{catchHdl'}\;\Varid{r}\;\Conid{Throw}{}\<[20]%
\>[20]{}\mathrel{=}\Varid{return}\;\Conid{Nothing}; \hspace{1em}\Varid{catchHdl'}\;\Varid{r}\;\Varid{op}{}\<[74]%
\>[74]{}\mathrel{=}\Conid{Call}\;\Varid{op}{}\<[E]%
\ColumnHook
 \end{hscode}\resethooks
\end{nscenter}
\end{codepage}
\end{equation}
As expected, applying these two handlers to the program \ensuremath{\Varid{safeDiv}\;\mathrm{5}}
produces different results (of types \ensuremath{\Conid{Free}\;\Varid{ES}\;\Conid{Int}} and \ensuremath{\Conid{Free}\;\Varid{ES}\;(\Conid{Maybe}\;\Conid{Int})} respectively):
\begin{equation*}
\setlength\mathindent{0pt}
\begin{codepage}[\linewidth]
\begin{nscenter}\begin{hscode}\SaveRestoreHook
\column{B}{@{}>{\hspre}c<{\hspost}@{}}%
\column{BE}{@{}l@{}}%
\column{4}{@{}>{\hspre}l<{\hspost}@{}}%
\column{9}{@{}>{\hspre}l<{\hspost}@{}}%
\column{31}{@{}>{\hspre}l<{\hspost}@{}}%
\column{52}{@{}>{\hspre}l<{\hspost}@{}}%
\column{E}{@{}>{\hspre}l<{\hspost}@{}}%
\>[4]{}\Varid{handle}\;(\Varid{catchHdl}\;(\Varid{return}\;\mathrm{42}))\;\Varid{return}\;(\Varid{safeDiv}\;\mathrm{5}){}\<[E]%
\\
\>[B]{}\mathrel{=}{}\<[BE]%
\>[4]{}\mathbf{do}\;{}\<[9]%
\>[9]{}\Varid{s}\leftarrow \Varid{get};\mathbf{if}\;\Varid{s}\equiv \mathrm{0}\;\mathbf{then}\;\Varid{return}\;\mathrm{42}\;\mathbf{else}\;\mathbf{do}\;\{\mskip1.5mu \Varid{put}\;(\Varid{n}\mathbin{/}\Varid{s});\Varid{return}\;(\Varid{n}\mathbin{/}\Varid{s})\mskip1.5mu\}{}\<[E]%
\\[\blanklineskip]%
\>[4]{}\Varid{handle}\;(\Varid{catchHdl'}\;(\Varid{return}\;\mathrm{42}))\;(\Varid{return}\hsdot{\cdot}{.\ }\Conid{Just})\;{}\<[52]%
\>[52]{}(\Varid{safeDiv}\;\mathrm{5}){}\<[E]%
\\
\>[B]{}\mathrel{=}{}\<[BE]%
\>[4]{}\mathbf{do}\;{}\<[9]%
\>[9]{}\Varid{s}\leftarrow \Varid{get};\mathbf{if}\;\Varid{s}\equiv \mathrm{0}\;{}\<[31]%
\>[31]{}\mathbf{then}\;\Varid{return}\;\Conid{Nothing}{}\<[E]%
\\
\>[31]{}\mathbf{else}\;\mathbf{do}\;\{\mskip1.5mu \Varid{put}\;(\Varid{n}\mathbin{/}\Varid{s});\Varid{return}\;(\Conid{Just}\;(\Varid{n}\mathbin{/}\Varid{s}))\mskip1.5mu\}{}\<[E]%
\ColumnHook
 \end{hscode}\resethooks
\end{nscenter}
\end{codepage}
\end{equation*}
Note that exception \emph{throwing} and \emph{catching} are modelled
differently in the approach of algebraic effects and handlers, one as an
operation in the signature \ensuremath{\Varid{ES}} and one as a handler, although it is
natural to expect both of them to be operations of \emph{the effect of exceptions}.
This asymmetry results from the fact that exception catching is \emph{not
algebraic}:
if \ensuremath{\Varid{catch}} was modelled as a binary operation in the signature, then the
monadic bind \ensuremath{\bind } of the free monad earlier, which intuitively means sequential
composition of programs, would imply that \ensuremath{(\Varid{catch}\;\Varid{r}\;\Varid{p})\bind \Varid{k}\mathrel{=}\Varid{catch}\;(\Varid{r}\bind \Varid{k})\;(\Varid{p}\bind \Varid{k})}, which is semantically undesirable.
Thus the perspective of Plotkin and Pretnar \cite{PlotPret13Hand} is that non-algebraic operations
like \ensuremath{\Varid{catch}} should be deemed different from algebraic operations, and they
can be modelled as handlers (of algebraic operations).

\subsection{Scoped Operations as Handlers Are Not Modular}\label{sec:problem}

However, this treatment of non-algebraic operations leads to a somewhat
subtle complication:
as observed by Wu et al.\ \cite{WuSH14}, when non-algebraic operations (such as \ensuremath{\Varid{catch}})
are modelled with handlers, these handlers play a dual role of (i)~modelling
the syntax of the operation (the scope for which exceptions are caught
by \ensuremath{\Varid{catch}}) and (ii) giving semantics to it (when an exception is caught, run
the recovery program).
To see the problem more concretely, ideally one would like to have a
syntactic operation \ensuremath{\Varid{catch}} of the following type that acts on computations
without giving specific semantics a priori,
\[
\ensuremath{\Varid{catch}\mathbin{::}\Conid{Free}\;\Varid{ES}\;\Varid{a}\to \Conid{Free}\;\Varid{ES}\;\Varid{a}\to \Conid{Free}\;\Varid{ES}\;\Varid{a}}
\]
allowing to write programs like
\begin{equation}\label{eq:prog}
\ensuremath{\Varid{prog}\mathrel{=}\mathbf{do}\;\{\mskip1.5mu \Varid{x}\leftarrow \Varid{catch}\;(\Varid{safeDiv}\;\mathrm{5})\;(\Varid{return}\;\mathrm{42});\Varid{put}\;(\Varid{x}\mathbin{+}\mathrm{1})\mskip1.5mu\}}
\end{equation}
and the semantics of (both algebraic and non-algebraic) operations in
\ensuremath{\Varid{prog}} can be given separately by handlers.
Unfortunately, when \ensuremath{\Varid{catch}} is modelled as handlers \ensuremath{\Varid{catchHdl}} or
\ensuremath{\Varid{catchHdl'}} as in the last subsection, the program \ensuremath{\Varid{prog}} must be written
differently depending on which handler is used:
\begin{equation*}
\setlength\mathindent{0pt}
\begin{codepage}[\linewidth]
\begin{nscenter}\begin{hscode}\SaveRestoreHook
\column{B}{@{}>{\hspre}l<{\hspost}@{}}%
\column{8}{@{}>{\hspre}l<{\hspost}@{}}%
\column{12}{@{}>{\hspre}l<{\hspost}@{}}%
\column{25}{@{}>{\hspre}l<{\hspost}@{}}%
\column{27}{@{}>{\hspre}l<{\hspost}@{}}%
\column{67}{@{}>{\hspre}l<{\hspost}@{}}%
\column{E}{@{}>{\hspre}l<{\hspost}@{}}%
\>[8]{}\mathbf{do}\;{}\<[12]%
\>[12]{}\Varid{x}\leftarrow \Varid{handle}\;{}\<[25]%
\>[25]{}(\Varid{catchHdl}\;(\Varid{return}\;\mathrm{42}))\;\Varid{return}\;(\Varid{safeDiv}\;\mathrm{5});\Varid{put}\;(\Varid{x}\mathbin{+}\mathrm{1}){}\<[E]%
\\[\blanklineskip]%
\>[B]{}\text{vs.}\quad\;{}\<[8]%
\>[8]{}\mathbf{do}\;{}\<[12]%
\>[12]{}\Varid{xMb}\leftarrow \Varid{handle}\;(\Varid{catchHdl'}\;(\Varid{return}\;\mathrm{42}))\;(\Varid{return}\hsdot{\cdot}{.\ }\Conid{Just})\;{}\<[67]%
\>[67]{}(\Varid{safeDiv}\;\mathrm{5}){}\<[E]%
\\
\>[12]{}\mathbf{case}\;\Varid{xMb}\;\mathbf{of}\;\{\mskip1.5mu {}\<[27]%
\>[27]{}\Conid{Nothing}\to \Varid{return}\;\Conid{Nothing}{}\<[E]%
\\
\>[27]{}(\Conid{Just}\;\Varid{x})\to \mathbf{do}\;\Varid{r}\leftarrow \Varid{put}\;(\Varid{x}\mathbin{+}\mathrm{1});\Varid{return}\;(\Conid{Just}\;\Varid{r})\mskip1.5mu\}{}\<[E]%
\ColumnHook
 \end{hscode}\resethooks
\end{nscenter}
\end{codepage}
\end{equation*}
The issue is that these handlers interpret the operation \ensuremath{\Varid{catch}} in different
semantic models,  \ensuremath{\Conid{Free}\;\Varid{ES}\;\Varid{a}} and \ensuremath{\Conid{Free}\;\Varid{ES}\;(\Conid{Maybe}\;\Varid{a})},
and this affects both the value \ensuremath{\Varid{x}} that is returned, and the way the
subsequent \ensuremath{\Varid{put}} is expressed.
Therefore, non-algebraic operation \ensuremath{\Varid{catch}} modelled as handlers is not as
modular as algebraic operations, weakening the advantage of programming with
algebraic effects.

\subsection{Scoped Effects and Functorial Algebras}\label{sec:overview}

Now we present an overview of a solution to the problem highlighted above by
modelling exception catching as \emph{scoped effects} \cite{PirogSWJ18} and
handle them using \emph{functorial algebras}, which will be more formally
developed in later sections.

\paragraph{Syntax of Scoped Operations}
To achieve modularity for (non-algebraic) operations delimiting scopes, such as
\ensuremath{\Varid{catch}}, which are called \emph{scoped operations}, Pir\'{o}g et al.\ \cite{PirogSWJ18}
generalise the free monad \ensuremath{\Conid{Free}\;\Sigma} to a monad \ensuremath{\Conid{Prog}\;\Sigma\;\Gamma}
accommodating both algebraic and scoped operations.
The monad is parameterised by two functors \ensuremath{\Sigma} and \ensuremath{\Gamma}, called the
\emph{algebraic signature} and the \emph{scoped signature} respectively.
The intention is that a constructor \ensuremath{\Conid{Op}\mathbin{::}(\Conid{R}\to \Varid{x})\to \Sigma\;\Varid{x}} of the algebraic
signature represents an algebraic operation \ensuremath{\Conid{Op}} producing an \ensuremath{\Conid{R}}-value
as usual, whereas a constructor \ensuremath{\Conid{Sc}\mathbin{::}(\Conid{N}\to \Varid{x})\to \Gamma\;\Varid{x}} of the scoped signature
represents a scoped operation \ensuremath{\Conid{Sc}} creating \ensuremath{\Conid{N}}-many scopes enclosing programs.

\begin{example}\label{ex:exc:sig}
As in the previous subsection,
the effect of \emph{exceptions} has an algebraic operation for \emph{throwing}
exceptions, which produces no values, and a scoped operation for
\emph{catching} exceptions, which creates two scopes, one enclosing the program
for which exceptions are caught, and the other enclosing the recovery
computation.
Thus the algebraic and scoped signatures are respectively
\begin{equation}
\ensuremath{\mathbf{data}\;\Conid{Throw}\;\Varid{x}\mathrel{=}\Conid{Throw}}  \hspace{2cm} \ensuremath{\mathbf{data}\;\Conid{Catch}\;\Varid{x}\mathrel{=}\Conid{Catch}\;\Varid{x}\;\Varid{x}}
\end{equation}
\end{example}

\begin{example}\label{ex:ndet:sig}
An effect of \emph{explicit nondeterminism} has two algebraic operations for
nondeterministic choice and a scoped operation \ensuremath{\Conid{Once}}:
\begin{equation}
\ensuremath{\mathbf{data}\;\Conid{Choice}\;\Varid{x}\mathrel{=}\Conid{Fail}\mid \Conid{Or}\;\Varid{x}\;\Varid{x}} \hspace{2cm} \ensuremath{\mathbf{data}\;\Conid{Once}\;\Varid{x}\mathrel{=}\Conid{Once}\;\Varid{x}}
\end{equation}
The intention is that this effect implements logic programming \cite{Hin98Pro}---solutions to 
a problem are exhaustively searched:
operation \ensuremath{\Conid{Or}\;\Varid{p}\;\Varid{q}} splits a search branch into two;
\ensuremath{\Conid{Fail}} marks a failed branch;
and the scoped operation \ensuremath{\Conid{Once}\;\Varid{p}} keeps only the first solution found by \ensuremath{\Varid{p}}, making
it \emph{semi-deterministic}, which is useful for speeding up the search with
heuristics from the programmer.
\end{example}

Similar to the free monad, the \ensuremath{\Conid{Prog}} monad models the syntax of computations
invoking operations from \ensuremath{\Sigma} and \ensuremath{\Gamma}:
\begin{equation}\label{eq:prog:monad}
\begin{codepage}
\begin{nscenter}\begin{hscode}\SaveRestoreHook
\column{B}{@{}>{\hspre}l<{\hspost}@{}}%
\column{28}{@{}>{\hspre}c<{\hspost}@{}}%
\column{28E}{@{}l@{}}%
\column{31}{@{}>{\hspre}l<{\hspost}@{}}%
\column{E}{@{}>{\hspre}l<{\hspost}@{}}%
\>[B]{}\mathbf{data}\;\Conid{Prog}\;\Sigma\;\Gamma\;\Varid{a}{}\<[28]%
\>[28]{}\mathrel{=}{}\<[28E]%
\>[31]{}\Conid{Return}\;\Varid{a}\mid \Conid{Call}\;(\Sigma\;(\Conid{Prog}\;\Sigma\;\Gamma\;\Varid{a})){}\<[E]%
\\
\>[28]{}\mid {}\<[28E]%
\>[31]{}\Conid{Enter}\;(\Gamma\;(\Conid{Prog}\;\Sigma\;\Gamma\;(\Conid{Prog}\;\Sigma\;\Gamma\;\Varid{a}))){}\<[E]%
\ColumnHook
 \end{hscode}\resethooks
\end{nscenter}
\end{codepage}
\end{equation}
Thus an element of \ensuremath{\Conid{Prog}\;\Sigma\;\Gamma\;\Varid{a}} can either (i) \emph{return} an \ensuremath{\Varid{a}}-value without
causing effects, or (ii) \emph{call} an algebraic operation in \ensuremath{\Sigma} with more
subterms of \ensuremath{\Conid{Prog}\;\Sigma\;\Gamma\;\Varid{a}} as the continuation after the operation,
or (iii) \emph{enter} the scope of a scoped operation.
The third case deserves more explanation:
the first \ensuremath{\Conid{Prog}} in \ensuremath{(\Gamma\;(\Conid{Prog}\;\Sigma\;\Gamma\;(\Conid{Prog}\;\Sigma\;\Gamma\;\Varid{a})))}
represents the programs enclosed by the scoped operation, and the second \ensuremath{\Conid{Prog}}
represents the continuation of the program after the scoped operation, and thus
the boundary between programs inside and outside the scope is kept in the
syntax tree, which is necessary because collapsing the boundary might change
the meaning of a program.
The distinction between algebraic and scoped operations can be seen more clearly
from the monadic bind of \ensuremath{\Conid{Prog}} (the monadic return of \ensuremath{\Conid{Prog}} is just
\ensuremath{\Conid{Return}}):
\begin{equation*}\label{eq:prog:monad:instance}
\setlength\mathindent{0pt}
\begin{codepage}
\begin{nscenter}\begin{hscode}\SaveRestoreHook
\column{B}{@{}>{\hspre}l<{\hspost}@{}}%
\column{13}{@{}>{\hspre}l<{\hspost}@{}}%
\column{E}{@{}>{\hspre}l<{\hspost}@{}}%
\>[B]{}(\bind )\mathbin{::}\Conid{Prog}\;\Sigma\;\Gamma\;\Varid{a}\to (\Varid{a}\to \Conid{Prog}\;\Sigma\;\Gamma\;\Varid{b})\to \Conid{Prog}\;\Sigma\;\Gamma\;\Varid{b}{}\<[E]%
\\
\>[B]{}(\Conid{Return}\;\Varid{a}){}\<[13]%
\>[13]{}\bind \Varid{k}\mathrel{=}\Varid{k}\;\Varid{a}{}\<[E]%
\\
\>[B]{}(\Conid{Call}\;\Varid{op}){}\<[13]%
\>[13]{}\bind \Varid{k}\mathrel{=}\Conid{Call}\;(\Varid{fmap}\;(\bind \Varid{k})\;\Varid{op}){}\<[E]%
\\
\>[B]{}(\Conid{Enter}\;\Varid{sc}){}\<[13]%
\>[13]{}\bind \Varid{k}\mathrel{=}\Conid{Enter}\;(\Varid{fmap}\;(\Varid{fmap}\;(\bind \Varid{k}))\;\Varid{sc}){}\<[E]%
\ColumnHook
 \end{hscode}\resethooks
\end{nscenter}
\end{codepage}
\end{equation*}
For algebraic operations, extending the continuation \ensuremath{(\bind \Varid{k})} directly acts on
the argument to the algebraic operation, whereas for scoped operation, \ensuremath{(\bind \Varid{k})}
acts on the second layer of \ensuremath{\Conid{Prog}}.
Thus for an algebraic operation \ensuremath{\Varid{o}}, \ensuremath{(\Varid{o}\;\Varid{p})\bind \Varid{k}} and \ensuremath{\Varid{o}\;(\Varid{p}\bind \Varid{k})} have the
same representation, whereas for a scoped operation \ensuremath{\Varid{s}}, \ensuremath{(\Varid{s}\;\Varid{p})\bind \Varid{k}} and \ensuremath{\Varid{s}\;(\Varid{p}\bind \Varid{k})} have different representations, which is precisely the distinction
between algebraic and scoped operations.

The constructors \ensuremath{\Conid{Call}} and \ensuremath{\Conid{Enter}} are clumsy to work with, and for writing
programs more naturally, we define \emph{smart constructors} for operations.
Generally, for algebraic operations \ensuremath{\Conid{Op}\mathbin{::}\Conid{F}\;\Varid{x}\to \Sigma\;\Varid{x}} and scoped operations
\ensuremath{\Conid{Sc}\mathbin{::}\Conid{G}\;\Varid{x}\to \Gamma\;\Varid{x}}, the smart constructors are
\begin{nscenter}\begin{hscode}\SaveRestoreHook
\column{B}{@{}>{\hspre}l<{\hspost}@{}}%
\column{58}{@{}>{\hspre}l<{\hspost}@{}}%
\column{83}{@{}>{\hspre}l<{\hspost}@{}}%
\column{E}{@{}>{\hspre}l<{\hspost}@{}}%
\>[B]{}\Varid{op}\mathbin{::}\Conid{F}\;(\Conid{Prog}\;\Sigma\;\Gamma\;\Varid{a})\to \Conid{Prog}\;\Sigma\;\Gamma\;\Varid{a}\;{}\<[58]%
\>[58]{}\hspace{1em}\;{}\<[83]%
\>[83]{}\Varid{sc}\mathbin{::}\Conid{G}\;(\Conid{Prog}\;\Sigma\;\Gamma\;\Varid{a})\to \Conid{Prog}\;\Sigma\;\Gamma\;\Varid{a}{}\<[E]%
\\
\>[B]{}\Varid{op}\mathrel{=}\Conid{Call}\hsdot{\cdot}{.\ }\Conid{Op}\;{}\<[83]%
\>[83]{}\Varid{sc}\mathrel{=}\Conid{Enter}\hsdot{\cdot}{.\ }\Varid{fmap}\;(\Varid{fmap}\;\Varid{return})\hsdot{\cdot}{.\ }\Conid{Sc}{}\<[E]%
\ColumnHook
 \end{hscode}\resethooks
\end{nscenter}
For example, the smart constructor for \ensuremath{\Conid{Catch}} (\autoref{ex:exc:sig}) is
\begin{nscenter}\begin{hscode}\SaveRestoreHook
\column{B}{@{}>{\hspre}l<{\hspost}@{}}%
\column{8}{@{}>{\hspre}l<{\hspost}@{}}%
\column{55}{@{}>{\hspre}l<{\hspost}@{}}%
\column{E}{@{}>{\hspre}l<{\hspost}@{}}%
\>[B]{}\Varid{catch}{}\<[8]%
\>[8]{}\mathbin{::}\Conid{Prog}\;\Sigma\;\Conid{Catch}\;\Varid{a}\to \Conid{Prog}\;\Sigma\;\Conid{Catch}\;\Varid{a}{}\<[55]%
\>[55]{}\to \Conid{Prog}\;\Sigma\;\Conid{Catch}\;\Varid{a}{}\<[E]%
\\
\>[B]{}\Varid{catch}\;\Varid{h}\;\Varid{r}\mathrel{=}\Conid{Enter}\;(\Conid{Catch}\;(\Varid{fmap}\;\Varid{return}\;\Varid{h})\;(\Varid{fmap}\;\Varid{return}\;\Varid{r})){}\<[E]%
\ColumnHook
 \end{hscode}\resethooks
\end{nscenter}
With all machinery in place, 
now we can define the program (\ref{eq:prog}) using \ensuremath{\Conid{Prog}} that we could not
write with \ensuremath{\Conid{Free}}:
\[\ensuremath{\Varid{prog}\mathrel{=}\mathbf{do}\;\{\mskip1.5mu \Varid{x}\leftarrow \Varid{catch}\;(\Varid{safeDiv}\;\mathrm{5})\;(\Varid{return}\;\mathrm{42});\Varid{put}\;(\Varid{x}\mathbin{+}\mathrm{1})\mskip1.5mu\}}\]

\paragraph{Handlers of Scoped Operations}
Similar to \ensuremath{\Conid{Free}}, the \ensuremath{\Conid{Prog}} monad merely models the syntax of effectful
computations, and more useful semantics need to be given by handlers.  Although
Pir\'{o}g et al.\ \cite{PirogSWJ18} developed a notion of \emph{indexed algebras} for this
purpose, indexed algebras turn out to be more complicated than necessary (we will discuss
them in \autoref{sec:comparing}), and the contribution of this paper is
a simpler kind of handlers for scoped operations, which we call \emph{functorial
algebras}.

\begin{figure}[t]
\small
\begin{minipage}{0.4\linewidth}
\begin{nscenter}\begin{hscode}\SaveRestoreHook
\column{B}{@{}>{\hspre}l<{\hspost}@{}}%
\column{3}{@{}>{\hspre}l<{\hspost}@{}}%
\column{13}{@{}>{\hspre}l<{\hspost}@{}}%
\column{E}{@{}>{\hspre}l<{\hspost}@{}}%
\>[B]{}\mathbf{data}\;\Conid{EndoAlg}\;\Varid{\Sigma}\;\Varid{\Gamma}\;\Varid{f}\mathrel{=}\Conid{EndoAlg}\;\{\mskip1.5mu {}\<[E]%
\\
\>[B]{}\hsindent{3}{}\<[3]%
\>[3]{}\Varid{returnE}{}\<[13]%
\>[13]{}\mathbin{::}\forall \Varid{x}\hsforall \hsdot{\cdot}{.\ }\Varid{x}\to \Varid{f}\;\Varid{x},{}\<[E]%
\\
\>[B]{}\hsindent{3}{}\<[3]%
\>[3]{}\Varid{callE}{}\<[13]%
\>[13]{}\mathbin{::}\forall \Varid{x}\hsforall \hsdot{\cdot}{.\ }\Varid{\Sigma}\;(\Varid{f}\;\Varid{x})\to \Varid{f}\;\Varid{x},{}\<[E]%
\\
\>[B]{}\hsindent{3}{}\<[3]%
\>[3]{}\Varid{enterE}{}\<[13]%
\>[13]{}\mathbin{::}\forall \Varid{x}\hsforall \hsdot{\cdot}{.\ }\Varid{\Gamma}\;(\Varid{f}\;(\Varid{f}\;\Varid{x}))\to \Varid{f}\;\Varid{x}\mskip1.5mu\}{}\<[E]%
\ColumnHook
 \end{hscode}\resethooks
\end{nscenter}
\end{minipage}%
\hspace{1.5cm}
\begin{minipage}{0.4\linewidth}
\begin{nscenter}\begin{hscode}\SaveRestoreHook
\column{B}{@{}>{\hspre}l<{\hspost}@{}}%
\column{3}{@{}>{\hspre}l<{\hspost}@{}}%
\column{12}{@{}>{\hspre}c<{\hspost}@{}}%
\column{12E}{@{}l@{}}%
\column{15}{@{}>{\hspre}l<{\hspost}@{}}%
\column{23}{@{}>{\hspre}l<{\hspost}@{}}%
\column{31}{@{}>{\hspre}c<{\hspost}@{}}%
\column{31E}{@{}l@{}}%
\column{E}{@{}>{\hspre}l<{\hspost}@{}}%
\>[B]{}\mathbf{data}\;\Conid{BaseAlg}\;\Varid{\Sigma}\;\Varid{\Gamma}\;\Varid{f}\;\Varid{a}{}\<[31]%
\>[31]{}\mathrel{=}{}\<[31E]%
\\
\>[B]{}\hsindent{3}{}\<[3]%
\>[3]{}\Conid{BaseAlg}\;{}\<[12]%
\>[12]{}\{\mskip1.5mu {}\<[12E]%
\>[15]{}\Varid{callB}{}\<[23]%
\>[23]{}\mathbin{::}\Varid{\Sigma}\;\Varid{a}\to \Varid{a}{}\<[E]%
\\
\>[12]{},{}\<[12E]%
\>[15]{}\Varid{enterB}{}\<[23]%
\>[23]{}\mathbin{::}\Varid{\Gamma}\;(\Varid{f}\;\Varid{a})\to \Varid{a}\mskip1.5mu\}{}\<[E]%
\ColumnHook
 \end{hscode}\resethooks
\end{nscenter}
\vspace{5pt}
\end{minipage}
\begin{nscenter}\begin{hscode}\SaveRestoreHook
\column{B}{@{}>{\hspre}l<{\hspost}@{}}%
\column{3}{@{}>{\hspre}l<{\hspost}@{}}%
\column{7}{@{}>{\hspre}l<{\hspost}@{}}%
\column{26}{@{}>{\hspre}l<{\hspost}@{}}%
\column{35}{@{}>{\hspre}l<{\hspost}@{}}%
\column{38}{@{}>{\hspre}l<{\hspost}@{}}%
\column{E}{@{}>{\hspre}l<{\hspost}@{}}%
\>[B]{}\Varid{hcata}\mathbin{::}(\Conid{Functor}\;\Varid{\Sigma},\Conid{Functor}\;\Varid{\Gamma})\Rightarrow (\Conid{EndoAlg}\;\Varid{\Sigma}\;\Varid{\Gamma}\;\Varid{f})\to \Conid{Prog}\;\Varid{\Sigma}\;\Varid{\Gamma}\;\Varid{a}\to \Varid{f}\;\Varid{a}{}\<[E]%
\\
\>[B]{}\Varid{hcata}\;\Varid{alg}\;(\Conid{Return}\;\Varid{x}){}\<[26]%
\>[26]{}\mathrel{=}\Varid{returnE}\;\Varid{alg}\;\Varid{x}{}\<[E]%
\\
\>[B]{}\Varid{hcata}\;\Varid{alg}\;(\Conid{Call}\;\Varid{op}){}\<[26]%
\>[26]{}\mathrel{=}(\Varid{callE}\;\Varid{alg}\hsdot{\cdot}{.\ }\Varid{fmap}\;(\Varid{hcata}\;\Varid{alg}))\;\Varid{op}{}\<[E]%
\\
\>[B]{}\Varid{hcata}\;\Varid{alg}\;(\Conid{Enter}\;\Varid{scope}){}\<[26]%
\>[26]{}\mathrel{=}(\Varid{enterE}\;\Varid{alg}\hsdot{\cdot}{.\ }\Varid{fmap}\;(\Varid{hcata}\;\Varid{alg}\hsdot{\cdot}{.\ }\Varid{fmap}\;(\Varid{hcata}\;\Varid{alg})))\;\Varid{scope}{}\<[E]%
\\[\blanklineskip]%
\>[B]{}\Varid{handle}{}\<[7]%
\>[7]{}\mathbin{::}(\Conid{Functor}\;\Varid{\Sigma},\Conid{Functor}\;\Varid{\Gamma}){}\<[E]%
\\
\>[7]{}\Rightarrow (\Conid{EndoAlg}\;\Varid{\Sigma}\;\Varid{\Gamma}\;\Varid{x})\to {}\<[38]%
\>[38]{}(\Conid{BaseAlg}\;\Varid{\Sigma}\;\Varid{\Gamma}\;\Varid{x}\;\Varid{b})\to (\Varid{a}\to \Varid{b})\to \Conid{Prog}\;\Varid{\Sigma}\;\Varid{\Gamma}\;\Varid{a}\to \Varid{b}{}\<[E]%
\\
\>[B]{}\Varid{handle}\;\Varid{ealg}\;\Varid{balg}\;\Varid{gen}\;(\Conid{Return}\;\Varid{x}){}\<[35]%
\>[35]{}\mathrel{=}\Varid{gen}\;\Varid{x}{}\<[E]%
\\
\>[B]{}\Varid{handle}\;\Varid{ealg}\;\Varid{balg}\;\Varid{gen}\;(\Conid{Call}\;\Varid{op}){}\<[35]%
\>[35]{}\mathrel{=}(\Varid{callB}\;\Varid{balg}\hsdot{\cdot}{.\ }\Varid{fmap}\;(\Varid{handle}\;\Varid{ealg}\;\Varid{balg}\;\Varid{gen}))\;\Varid{op}{}\<[E]%
\\
\>[B]{}\Varid{handle}\;\Varid{ealg}\;\Varid{balg}\;\Varid{gen}\;(\Conid{Enter}\;\Varid{sc}){}\<[E]%
\\
\>[B]{}\hsindent{3}{}\<[3]%
\>[3]{}\mathrel{=}(\Varid{enterB}\;\Varid{balg}\hsdot{\cdot}{.\ }\Varid{fmap}\;(\Varid{hcata}\;\Varid{ealg}\hsdot{\cdot}{.\ }\Varid{fmap}\;(\Varid{handle}\;\Varid{ealg}\;\Varid{balg}\;\Varid{gen})))\;\Varid{sc}{}\<[E]%
\ColumnHook
 \end{hscode}\resethooks
\end{nscenter}
\vspace{-\belowdisplayskip}\vspace{-0.7\baselineskip}
\caption{A Haskell implementation of handling with functorial algebras}\label{fig:functorial_impl}
\end{figure}

Given signatures \ensuremath{\Sigma} and \ensuremath{\Gamma}, a functorial algebra for them is a
quadruple $\tuple{\ensuremath{\Varid{f}}, \ensuremath{\Varid{b}}, \ensuremath{\Varid{ealg}}, \ensuremath{\Varid{balg}}}$ for some functor \ensuremath{\Varid{f}} called the
\emph{endofunctor carrier}, type \ensuremath{\Varid{b}} called the \emph{base carrier}.
The other two components \ensuremath{\Varid{ealg}\mathbin{::}\Conid{EndoAlg}\;\Sigma\;\Gamma\;\Varid{f}} and \ensuremath{\Varid{balg}\mathbin{::}\Conid{BaseAlg}\;\Sigma\;\Gamma\;\Varid{f}\;\Varid{b}} are called the \emph{endofunctor algebra} and the \emph{base
algebra}.
Their types are fully shown in \autoref{fig:functorial_impl}.
The intuition is that functor \ensuremath{\Varid{f}} and \ensuremath{\Varid{ealg}} interpret the part of a
program enclosed by scoped operations,
and the type \ensuremath{\Varid{b}} and \ensuremath{\Varid{balg}} interpret the part of a program not enclosed
by any scopes.

\begin{example}
The standard semantics of exception catching (cf.\ handler (\ref{eq:catchHdl}))
can be implemented by a functorial algebra with the conventional \ensuremath{\Conid{Maybe}}
functor as the endofunctor carrier with the following \ensuremath{\Conid{EndoAlg}}:
\begin{equation*}
\setlength\mathindent{0pt}
\begin{codepage}[\linewidth]
\begin{nscenter}\begin{hscode}\SaveRestoreHook
\column{B}{@{}>{\hspre}l<{\hspost}@{}}%
\column{4}{@{}>{\hspre}l<{\hspost}@{}}%
\column{42}{@{}>{\hspre}l<{\hspost}@{}}%
\column{52}{@{}>{\hspre}l<{\hspost}@{}}%
\column{61}{@{}>{\hspre}l<{\hspost}@{}}%
\column{75}{@{}>{\hspre}l<{\hspost}@{}}%
\column{83}{@{}>{\hspre}l<{\hspost}@{}}%
\column{101}{@{}>{\hspre}l<{\hspost}@{}}%
\column{E}{@{}>{\hspre}l<{\hspost}@{}}%
\>[B]{}\Varid{excE}\mathbin{::}\Conid{EndoAlg}\;\Conid{Throw}\;\Conid{Catch}\;\Conid{Maybe}{}\<[E]%
\\
\>[B]{}\Varid{excE}\mathrel{=}\Conid{EndoAlg}\;\{\mskip1.5mu \mathinner{\ldotp\ldotp}\mskip1.5mu\}\;\mathbf{where}\;{}\<[42]%
\>[42]{}\hspace{1em}\;{}\<[52]%
\>[52]{}\hspace{1em}\;{}\<[61]%
\>[61]{}\hspace{1em}\;{}\<[75]%
\>[75]{}\Varid{enterE}{}\<[83]%
\>[83]{}\mathbin{::}\Conid{Catch}\;(\Conid{Maybe}\;(\Conid{Maybe}\;\Varid{a})){}\<[E]%
\\
\>[B]{}\hsindent{4}{}\<[4]%
\>[4]{}\Varid{returnE}\mathrel{=}\Conid{Just}{}\<[83]%
\>[83]{}\to \Conid{Maybe}\;\Varid{a}{}\<[E]%
\\[\blanklineskip]%
\>[B]{}\hsindent{4}{}\<[4]%
\>[4]{}\Varid{callE}\;\Conid{Throw}\mathrel{=}\Conid{Nothing}\;{}\<[75]%
\>[75]{}\Varid{enterE}\;(\Conid{Catch}\;\Conid{Nothing}\;\Varid{r}){}\<[101]%
\>[101]{}\mathrel{=}\Varid{join}\;\Varid{r}{}\<[E]%
\\
\>[75]{}\Varid{enterE}\;(\Conid{Catch}\;(\Conid{Just}\;\Varid{k})\;\anonymous )\mathrel{=}\Varid{k}{}\<[E]%
\ColumnHook
 \end{hscode}\resethooks
\end{nscenter}
\end{codepage}
\end{equation*}
For the base carrier that interprets operations not enclosed by any \ensuremath{\Varid{catch}}, a
straightforward choice is just taking \ensuremath{\Conid{Maybe}\;\Varid{a}} as the base carrier for a type
\ensuremath{\Varid{a}}, and setting \ensuremath{\Varid{callB}\mathrel{=}\Varid{callE}} and \ensuremath{\Varid{enterB}\mathrel{=}\Varid{enterE}}, which means that
operations inside and outside scopes are interpreted in the same way.

In general, we can define a specialised version of \ensuremath{\Varid{handle}}
(\autoref{fig:functorial_impl}) that only takes an endofunctor algebra as
input for interpreting operations inside and outside scopes in the
same way:
\begin{equation*}
\begin{codepage}[\linewidth]
\setlength\mathindent{0pt}
\begin{nscenter}\begin{hscode}\SaveRestoreHook
\column{B}{@{}>{\hspre}l<{\hspost}@{}}%
\column{E}{@{}>{\hspre}l<{\hspost}@{}}%
\>[B]{}\Varid{handleE}\mathbin{::}(\Conid{EndoAlg}\;\Sigma\;\Gamma\;\Varid{f})\to \Conid{Prog}\;\Sigma\;\Gamma\;\Varid{a}\to \Varid{f}\;\Varid{a}{}\<[E]%
\\
\>[B]{}\Varid{handleE}\;\Varid{ealg}\mathord{@}(\Conid{EndoAlg}\;\{\mskip1.5mu \mathinner{\ldotp\ldotp}\mskip1.5mu\})\mathrel{=}\Varid{handle}\;\Varid{ealg}\;(\Conid{BaseAlg}\;\Varid{callE}\;\Varid{enterE})\;\Varid{returnE}{}\<[E]%
\ColumnHook
 \end{hscode}\resethooks
\end{nscenter}
\end{codepage}
\end{equation*}
Applying \ensuremath{\Varid{handleE}\;\Varid{excE}} to the following program produces \ensuremath{\Conid{Just}\;\mathrm{43}} as expected.
\begin{equation}\label{eq:excE:test}
\ensuremath{\mathbf{do}\;\{\mskip1.5mu \Varid{x}\leftarrow \Varid{catch}\;\Varid{throw}\;(\Varid{return}\;\mathrm{42});\Varid{return}\;(\Varid{x}\mathbin{+}\mathrm{1})\mskip1.5mu\}}
\end{equation}
For the non-standard semantics (cf.\ (\ref{eq:catchHdl2})) that disables
exception recovery, one can define another endofunctor algebra \ensuremath{\Varid{excE'}} by
replacing \ensuremath{\Varid{enterE}} in \ensuremath{\Varid{excE}} with 
\begin{nscenter}\begin{hscode}\SaveRestoreHook
\column{B}{@{}>{\hspre}l<{\hspost}@{}}%
\column{10}{@{}>{\hspre}l<{\hspost}@{}}%
\column{28}{@{}>{\hspre}l<{\hspost}@{}}%
\column{E}{@{}>{\hspre}l<{\hspost}@{}}%
\>[B]{}\Varid{enterE'}{}\<[10]%
\>[10]{}\mathbin{::}\Conid{Catch}\;(\Conid{Maybe}\;(\Conid{Maybe}\;\Varid{a}))\to \Conid{Maybe}\;\Varid{a}{}\<[E]%
\\
\>[B]{}\Varid{enterE'}\;(\Conid{Catch}\;\Conid{Nothing}\;\anonymous ){}\<[28]%
\>[28]{}\mathrel{=}\Conid{Nothing};\hspace{1em}\;\Varid{enterE'}\;(\Conid{Catch}\;(\Conid{Just}\;\Varid{k})\;\anonymous )\mathrel{=}\Varid{k}{}\<[E]%
\ColumnHook
 \end{hscode}\resethooks
\end{nscenter}
With \ensuremath{\Varid{excE'}}, handling the program in (\ref{eq:excE:test})
produces \ensuremath{\Conid{Nothing}} as expected.
\end{example}

Now we provide some intuition for how functorial algebras work.
First note that the three fields of \ensuremath{\Conid{EndoAlg}} in \autoref{fig:functorial_impl} precisely correspond
to the three cases of \ensuremath{\Conid{Prog}} (\ref{eq:prog:monad}).
Thus by replacing the constructors of \ensuremath{\Conid{Prog}} with the correspond fields of \ensuremath{\Conid{EndoAlg}},
we have a polymorphic function
\ensuremath{\Varid{hcata}\;\Varid{ealg}\mathbin{::}\forall \Varid{x}\hsforall \hsdot{\cdot}{.\ }\Conid{Prog}\;\Sigma\;\Gamma\;\Varid{x}\to \Varid{f}\;\Varid{x}} (\autoref{fig:functorial_impl})
turning a program into a value in \ensuremath{\Varid{f}}.

The function \ensuremath{\Varid{handle}} (\autoref{fig:functorial_impl}) takes a functorial algebra,
a function \ensuremath{\Varid{gen}\mathbin{::}\Varid{a}\to \Varid{b}} and a program \ensuremath{\Varid{p}} as arguments, and it handles all the 
effectful operations in \ensuremath{\Varid{p}} by using \ensuremath{\Varid{hcata}\;\Varid{ealg}} for interpreting the part of \ensuremath{\Varid{p}}
inside scoped operations and \ensuremath{\Varid{balg}} for interpreting the outermost layer of \ensuremath{\Varid{p}}
outside any scoped operations.
The function \ensuremath{\Varid{gen}} corresponds to the `value case' of handlers of algebraic
effects, which transforms the \ensuremath{\Varid{a}}-value returned by a program into the type
\ensuremath{\Varid{b}} for interpretation.

We close this section with some more examples of handling scoped effects
with functorial algebras.
The supplementary material of this paper also contains an OCaml implementation
of functorial algebras and the following examples.

\begin{example}\label{ex:ndet:hdl}
The standard way to handle explicit nondeterminism  with the semi-deterministic
operator \ensuremath{\Varid{once}} (\autoref{ex:ndet:sig}) is using a functorial algebra with the
list functor as the endofunctor carrier together with the following algebra:
\begin{equation*}
\begin{codepage}[\linewidth]
\begin{nscenter}\begin{hscode}\SaveRestoreHook
\column{B}{@{}>{\hspre}l<{\hspost}@{}}%
\column{5}{@{}>{\hspre}l<{\hspost}@{}}%
\column{21}{@{}>{\hspre}l<{\hspost}@{}}%
\column{34}{@{}>{\hspre}l<{\hspost}@{}}%
\column{61}{@{}>{\hspre}l<{\hspost}@{}}%
\column{69}{@{}>{\hspre}l<{\hspost}@{}}%
\column{104}{@{}>{\hspre}l<{\hspost}@{}}%
\column{E}{@{}>{\hspre}l<{\hspost}@{}}%
\>[B]{}\Varid{ndetE}\mathbin{::}\Conid{EndoAlg}\;\Conid{Choice}\;\Conid{Once}\;[\mskip1.5mu \mskip1.5mu]\;{}\<[34]%
\>[34]{}\hspace{1em}\;\hspace{1em}\;{}\<[61]%
\>[61]{}\Varid{enterE}\mathbin{::}\Conid{Once}\;[\mskip1.5mu [\mskip1.5mu \Varid{a}\mskip1.5mu]\mskip1.5mu]\to [\mskip1.5mu \Varid{a}\mskip1.5mu]{}\<[E]%
\\
\>[B]{}\Varid{ndetE}\mathrel{=}\Conid{EndoAlg}\;\{\mskip1.5mu \mathinner{\ldotp\ldotp}\mskip1.5mu\}\;\mathbf{where}\;{}\<[61]%
\>[61]{}\Varid{enterE}\;{}\<[69]%
\>[69]{}(\Conid{Once}\;\Varid{x})\mathrel{=}{}\<[E]%
\\
\>[B]{}\hsindent{5}{}\<[5]%
\>[5]{}\Varid{callE}\mathbin{::}\Conid{Choice}\;[\mskip1.5mu \Varid{a}\mskip1.5mu]\to [\mskip1.5mu \Varid{a}\mskip1.5mu]{}\<[61]%
\>[61]{}\hspace{2ex}\;\mathbf{if}\;\Varid{x}\equiv [\mskip1.5mu \mskip1.5mu]\;\mathbf{then}\;[\mskip1.5mu \mskip1.5mu]\;\mathbf{else}\;{}\<[104]%
\>[104]{}\Varid{head}\;\Varid{x}{}\<[E]%
\\
\>[B]{}\hsindent{5}{}\<[5]%
\>[5]{}\Varid{callE}\;\Conid{Fail}{}\<[21]%
\>[21]{}\mathrel{=}[\mskip1.5mu \mskip1.5mu]\;{}\<[61]%
\>[61]{}\Varid{returnE}\mathbin{::}\Varid{a}\to [\mskip1.5mu \Varid{a}\mskip1.5mu]{}\<[E]%
\\
\>[B]{}\hsindent{5}{}\<[5]%
\>[5]{}\Varid{callE}\;(\Conid{Or}\;\Varid{x}\;\Varid{y}){}\<[21]%
\>[21]{}\mathrel{=}\Varid{x}+\!\!\!+\Varid{y}\;{}\<[61]%
\>[61]{}\Varid{returnE}\;\Varid{x}\mathrel{=}[\mskip1.5mu \Varid{x}\mskip1.5mu]{}\<[E]%
\ColumnHook
 \end{hscode}\resethooks
\end{nscenter}
\end{codepage}
\end{equation*}
Then applying \ensuremath{\Varid{handleE}\;\Varid{ndetE}} to the following program produces \ensuremath{[\mskip1.5mu \mathrm{1},\mathrm{2}\mskip1.5mu]} as
expected.
In comparison, if \ensuremath{\Varid{once}} were algebraic, the result would be \ensuremath{[\mskip1.5mu \mathrm{1}\mskip1.5mu]}.
\[
\ensuremath{\mathbf{do}\;\{\mskip1.5mu \Varid{n}\leftarrow \Varid{once}\;(\Varid{or}\;(\Varid{return}\;\mathrm{1})\;(\Varid{return}\;\mathrm{3}));\Varid{or}\;(\Varid{return}\;\Varid{n})\;(\Varid{return}\;(\Varid{n}\mathbin{+}\mathrm{1}))\mskip1.5mu\}}
\] 
\end{example}

\begin{example}
In the last example we used the list functor to interpret explicit
nondeterminism, resulting in the \emph{depth-first search} (DFS) strategy for
searching.
Noted by Spivey \cite{Spivey2009}, other search strategies can be implemented by
other choices of functors.
For example, \emph{depth-bounded search} (DBS) can be implemented with 
the functor \ensuremath{\Conid{Int}\to [\mskip1.5mu \Varid{a}\mskip1.5mu]}, and \emph{breadth-first search} (BFS) can be implemented 
with the functor \ensuremath{[\mskip1.5mu [\mskip1.5mu \Varid{a}\mskip1.5mu]\mskip1.5mu]} (or Kidney and Wu \cite{Kidney21}'s more efficient \ensuremath{\Conid{LevelT}}
functor).

A powerful application of scoped effects is modelling search strategies:
\[
\ensuremath{\mathbf{data}\;\Conid{Strategy}\;\Varid{x}\mathrel{=}\Conid{DFS}\;\Varid{x}\mid \Conid{BFS}\;\Varid{x}\mid \Conid{DBS}\;\Conid{Int}\;\Varid{x}}
\]
so that the programmer can freely specify the search strategy of nondeterministic
choices in a scope.
The algebraic signature \ensuremath{\Conid{Choice}} and scoped signature \ensuremath{\Conid{Strategy}} can be
handled by a functorial algebra carried by the endofunctor \ensuremath{([\mskip1.5mu \Varid{a}\mskip1.5mu],[\mskip1.5mu [\mskip1.5mu \Varid{a}\mskip1.5mu]\mskip1.5mu],\Conid{Int}\to [\mskip1.5mu \Varid{a}\mskip1.5mu])} and a base type \ensuremath{[\mskip1.5mu \Varid{a}\mskip1.5mu]} (assuming that depth-first search is the default strategy).
The complete code is in the supplementary material. 
\end{example}

\begin{example}
A scoped operation for the effect of mutable state is the operation \ensuremath{\Varid{local}\;\Varid{s}\;\Varid{p}}
that executes the program \ensuremath{\Varid{p}} with a state \ensuremath{\Varid{s}} and restores to the original
state after \ensuremath{\Varid{p}} finishes.
Thus \ensuremath{(\Varid{local}\;\Varid{s}\;\Varid{p};\Varid{k})} is different from \ensuremath{\Varid{local}\;\Varid{s}\;(\Varid{p};\Varid{k})}, and \ensuremath{\Varid{local}}
should be modelled as a scoped operations of signature \ensuremath{\mathbf{data}\;\Conid{Local}\;\Varid{s}\;\Varid{a}\mathrel{=}\Conid{Local}\;\Varid{s}\;\Varid{a}}.
Together with the usual algebraic operations \ensuremath{\Varid{get}} and \ensuremath{\Varid{put}} of state,
\ensuremath{\Conid{Local}} can be interpreted with a functorial algebra carried
by the state monad \ensuremath{\mathbf{type}\;\Conid{State}\;\Varid{s}\;\Varid{a}\mathrel{=}\Varid{s}\to (\Varid{s},\Varid{a})}.
The essential part of the functorial algebra is the following \ensuremath{\Varid{enterE}}
for \ensuremath{\Conid{Local}} (complete code in the supplementary material):
\[
\begin{codepage}
\begin{nscenter}\begin{hscode}\SaveRestoreHook
\column{B}{@{}>{\hspre}l<{\hspost}@{}}%
\column{E}{@{}>{\hspre}l<{\hspost}@{}}%
\>[B]{}\Varid{enterE}\mathbin{::}\Conid{Local}\;(\Conid{State}\;\Varid{s}\;(\Conid{State}\;\Varid{s}\;\Varid{a}))\to \Conid{State}\;\Varid{s}\;\Varid{a}{}\<[E]%
\\
\>[B]{}\Varid{enterE}\;(\Conid{Local}\;\Varid{s'}\;\Varid{f})\;\Varid{s}\mathrel{=}\mathbf{let}\;(\anonymous ,\Varid{k})\mathrel{=}\Varid{f}\;\Varid{s}\;\mathbf{in}\;\Varid{k}\;\Varid{s}{}\<[E]%
\ColumnHook
 \end{hscode}\resethooks
\end{nscenter}
\end{codepage}
\]
\end{example}

\begin{example}
Parallel composition of processes is not an operation in the usual algebraic
presentations of process calculi \cite{Stark08Pi,Staton13} precisely because
it not algebraic: $\ensuremath{(\Varid{p}\mid \Varid{q})\bind \Varid{k}} \neq \ensuremath{(\Varid{p}\bind \Varid{k})\mid (\Varid{q}\bind \Varid{k})}$.
Again, we can model it as a scoped operation, and different scheduling
behaviours of processes can be given as different functorial algebras.
The supplementary material contains complete code of handling parallel composition
using the so-called resumption monad 
\cite{Claessen1999,Pirog2012}:
\[
\ensuremath{\mathbf{data}\;\Conid{Resumption}\;\Varid{m}\;\Varid{a}\mathrel{=}\Conid{More}\;(\Varid{m}\;(\Conid{Resumption}\;\Varid{m}\;\Varid{a}))\mid \Conid{Done}\;\Varid{a}}
\]
\end{example}

\section{Categorical Foundations for Scoped Operations}
\label{sec:foundation}

We now move on to a categorical foundation for scoped effects and functorial
algebras.
First, we recall some standard category theory underlying algebraic effects 
and handlers (\autoref{sec:algebraic:effects}) and also Pir\'{o}g et al.\ \cite{PirogSWJ18}'s
monad $P$ that models the syntax of scoped operations, which is exactly the
\ensuremath{\Conid{Prog}} monad in the Haskell implementation (\autoref{sec:syntax:scoped}).
Then, we define functorial algebras formally (\autoref{sec:fctAlgDef}) and show
that there is an adjunction between the category of functorial algebras and the
base category (\autoref{sec:adjFctAlg}) inducing the monad $P$, which provides
a means to interpret the syntax of scoped operations.

The rest of this paper assumes familiarity with basic category theory, such as
adjunctions, monads, and initial algebras, which are covered by standard texts
\cite{MacLane1978,Barr1990category,Riehl17category}.
The mathematical notation in this paper is summarised in \refapp{app:notation},
which may be consulted if the meaning of some symbols are unclear.

\subsection{Syntax and Semantics of Algebraic Operations}
\label{sec:algebraic:effects}

The relationships between \emph{equational theories}, \emph{Lawvere theories},
\emph{monads}, and \emph{computational effects} are well-studied for decades
from many perspectives
\cite{Robinson02,Kelly93,Power1999enriched,fossacs/PlotkinP02,Hyland07,Moggi95}.
Here we recap a simplified version of equational theories by Kelly and Power \cite{Kelly93}
that we follow to model algebraic and scoped effects on \emph{locally finitely
presentable} (lfp) categories \cite{Adamek_rosicky_1994}.
However, we only consider \emph{unenriched} categories in this paper.

\paragraph{Locally Finitely Presentable Categories}
The use of lfp categories in this paper is limited to some standard results
about the existence of many initial algebras in lfp categories, and thus a
reader not familiar with lfp categories may follow this paper with some simple
intuition:
a category $\CatC$ is lfp if it has all (small) colimits and a set of
\emph{finitely presentable objects} such that every object in $\CatC$ can be
obtained by `glueing' (formally, as \emph{filtered colimits} of) some finitely
presentable objects.
For example, $\Set$ is lfp with finite sets as its finitely
presentable objects, and indeed every set can be obtained by glueing, here
meaning taking the union of, all its finite subsets:
$
X = \bigcup \ \aset{N \subseteq X \mid N \text{ finite}}
$.
Other examples of lfp categories include the category of partially ordered sets,
the category of graphs, the category of small categories, and presheaf
categories (we refer the reader to the excellent exposition \cite{Robinson02}
for concrete examples), thus lfp categories are widespread to cover many
semantic settings of programming languages.

Moreover, an endofunctor $F : \CatC \to \CatC$ is said to be \emph{finitary} if it 
preserves `glueing' (filtered colimits), which implies that its values $F X$ 
are determined by its values at finitely presentable objects:
$
F X \cong F (\texttt{colim}_{i} N_i) \cong \texttt{colim}_{i} F N_i
$
where $N_i$ are the finitely presentable objects that generate $X$ when glued
together.
For example, polynomial functors $\coprod_{n \in \Nat}P n \times (-)^n$ on $\Set$
are finitary where $P n$ is a set for every $n$.

\paragraph{Algebraic Operations on LFP Categories}
Fixing an lfp category $\CatC$, we take finitary endofunctors
$\Sigma : \CatC \to \CatC$ as signatures of operations on $\CatC$.
Like in \autoref{sec:alg}, the intuition is that every natural transformation
$\coprod_{\CatC(R, -)} P \to \Sigma -$ for some object $P : \CatC$ and finitely
presentable object $R : \CatC$ standards for an operation taking a parameter of
type $P$ and $R$-many arguments.
The category $\Sigma\Alg$ of \emph{$\Sigma$-algebras} is defined as usual:
it has pairs $\tuple{X : \CatC, \alpha : \Sigma X \to X}$ as objects and 
morphisms $h : X \to X'$ such that $h \vcomp \alpha = \alpha' \vcomp \Sigma h$
as morphisms $\tuple{X, \alpha} \to \tuple{X', \alpha'}$.
The following classical results (see e.g.\ \cite{Adamek1974,Barr70Co}) give
sufficient conditions for constructing initial and free $\Sigma$-algebras:

\begin{lemma}\label{lem:free:alg:simple}
If category $\CatC$ has finite coproducts and colimits of all $\omega$-chains and
functor $\Sigma : \CatC \to \CatC$ preserves them, then the forgetful
functor $\Ul_\Sigma : \Sigma\Alg \to \CatC$ forgetting the structure maps has a
left adjoint $\Free_{\Sigma} : \CatC \to \Sigma\Alg$ mapping every $X : \CatC$
to a $\Sigma$-algebra $\abracket{\Sigma^* X, \op_X}$ where $\Sigma^* X$ denotes
the initial algebra $\lfix{Y}{X + \Sigma Y}$ and $\op_X : \Sigma \Sigma^* X \to
\Sigma^*X$.
\end{lemma}

\autoref{lem:free:alg:simple} is applicable to our setting since $\CatC$ being lfp
directly implies that it has all colimits, and finitary functors $\Sigma$
preserve colimits of $\omega$-chains because colimits of $\omega$-chains are
filtered. 
Hence we have an adjunction: $\Free_\Sigma \dashv \Ul_\Sigma : \Sigma\Alg \to \CatC$.
We denote the monad from the adjunction by $\Sigma^* = \Ul_\Sigma \Free_\Sigma$
(which is implemented as the \ensuremath{\Conid{Free}\;\Sigma} monad in \autoref{sec:alg}).
The idea is still that syntactic terms built from operations in $\Sigma$ are
modelled by the monad $\Sigma^*$, and semantics of operations are given by
$\Sigma$-algebras.
Given any $\Sigma$-algebra $\tuple{X, \alpha : \Sigma X \to X}$ and morphism 
$g : A \to X$ in $\CatC$, they induce an interpretation morphism $\ensuremath{\Varid{handle}}_{\tuple{X, \alpha}} g
: \Sigma^* A \to X$ s.t.
\begin{equation}\label{eq:handle:alg}
\ensuremath{\Varid{handle}}_{\tuple{X, \alpha}} g = \Ul_\Sigma (\epsilon_{\tuple{X, \alpha}} \vcomp
\Free_\Sigma g) : \Sigma^* A = \Ul_\Sigma \Free_\Sigma A \to X
\end{equation}
where $\epsilon_{\tuple{X, \alpha}} : \Free_\Sigma \Ul_\Sigma \tuple{X, \alpha}
\to \tuple{X, \alpha}$ is the counit of $\Free_\Sigma \dashv \Ul_\Sigma$.

\paragraph{Algebraic Effects and Handlers}
The perspective of Plotkin and Pretnar \cite{PlotPret13Hand} is that computational effects are
characterised by signatures $\Sigma$ of primitive effectful operations, and
they determine monads $\Sigma^*$ that model programs syntactically.
Then $\Sigma$-algebras are \emph{handlers} \cite{PlotPret13Hand} of
operations that can be applied to programs using (\ref{eq:handle:alg}) to give
specific semantics to operations.

The approach of algebraic effects has led to a significant body of research on
programming with effects and handlers, but it imposes an assumption on the
operations to be modelled: the construction of $\Sigma^*$ in
\autoref{lem:free:alg:simple} \cite{Adamek1974,Barr70Co} implies that the
multiplication $\mu$ of the monad $\Sigma^*$ satisfies the \emph{algebraicity}
property:
$\op \vcomp (\Sigma \hcomp \mu) = \mu \vcomp (\op \hcomp \Sigma^*) : \Sigma
\Sigma^* \Sigma^* \to \Sigma^*$ where $\op : \Sigma (\Sigma^*) \to \Sigma^*$.
This intuitively means that every operation in $\Sigma$ must be commutative with
sequential composition of computations.
Many, but not all, effectful operations satisfy this property, and they are
called \emph{algebraic operations}.

\paragraph{Adjoint Approach to Effects}
The crux of algebraic effects and handlers is the adjunction $\Free_\Sigma
\dashv \Ul_\Sigma$.
However, we have not relied on the adjunction being the free/forgetful one at all:
given any monad $P : \CatC \to \CatC$ that models the syntax of effectful
$P$rograms, if $L \dashv R : \CatD \to \CatC$ is an adjunction such that $R L
\cong P$ as monads, then objects $D$ in $\CatD$ provide a means to
interpret programs $P A$---for any $g : A \to R D$ in $\CatC$, we have the
following interpretation morphism
\begin{equation}
\label{rem:adjoint:theoretic}
\ensuremath{\Varid{handle}}_D g = R (\epsilon_D \vcomp L g) : P A \cong R (L A) \to R D
\end{equation}  
The intuition for $g$ is that it transforms the returned value $A$
of a computation into the carrier $R D$, so it corresponds to the `value
case' of effect handlers \cite{BauerP15}.
Pir\'{o}g et al.\ \cite{PirogSWJ18} call this approach the \emph{adjoint-theoretic approach to
syntax and semantics of effects}, and they construct an adjunction between
\emph{indexed algebras} and the base category for modelling scoped operations.
Earlier, Levy \cite{Levy2003248} and Kammar and Plotkin \cite{Kammar12} also adopt a similar
adjunction-based viewpoint in the treatment of call-by-push-value calculi:
\emph{value types} are interpreted in the base category $\CatC$, and
\emph{computation types} are interpreted in the algebra category $\CatD$.

\begin{remark}\label{rem:equations}
A notable missing part of our treatment is the \emph{equations} that specify
operations in a signature.
Following Kelly and Power \cite{Kelly93}, an equation for a signature $\Sigma : \CatC \to
\CatC$ can be formulated as a pair of monad morphisms $\sigma, \tau :
\Gamma^* \to \Sigma^*$ for some finitary functor $\Gamma$, and taking their
coequaliser 
\begin{tikzcd}[column sep=small]
	{\Gamma^*} & {\Sigma^*} & M
	\arrow["\tau", shift left=1, from=1-1, to=1-2]
	\arrow["\sigma"', shift right=1, from=1-1, to=1-2]
	\arrow[two heads, from=1-2, to=1-3]
\end{tikzcd}
in the category of finitary monads constructs a monad $M$ that represents terms
modulo the equation $l = r$.
Although it seems straightforward to extend this formulation of equational
theories work with scoped effects, we do not consider equations in this paper
for the sake of simplicity.
\end{remark}

\begin{remark}
Working with lfp categories precludes operations with infinite arguments, such
as the \ensuremath{\Varid{get}} operation (\ref{eq:sig:es}) of mutable state when the state has
infinite possible values, but this limitation is not inherent and can be
handled by moving to \emph{locally $\kappa$-presentable categories}
\cite{Adamek_rosicky_1994} for some larger cardinal $\kappa$.
\end{remark}

\subsection{Syntax of Scoped Operations}
\label{sec:syntax:scoped}

Not all operations in programming languages can be adequately modelled as
algebraic operations on $\Set$, for example,
$\lambda$-abstraction \cite{FiorePT99}, 
memory cell generation \cite{fossacs/PlotkinP02,Levy03Book},
more generally, effects with dynamically generated instances \cite{Staton13},
explicit substitution \cite{GhaniUH06}, channel restriction in $\pi$-calculus
\cite{Stark08Pi}, and their syntax are usually modelled in some functor
categories.
More recently, Pir\'{o}g et al.\ \cite{PirogSWJ18} extend Ghani and Uustalu \cite{GhaniUH06}'s work to model a
family of non-algebraic operations, which they call \emph{scoped operations}.
In this subsection, we review their development in the setting of lfp categories.
Throughout the rest of the paper, we fix an lfp category $\CatC$, and refer to
it as the \emph{base category}, and it is intended to be the category in which
types of a programming language are interpreted.
Furthermore, we fix two finitary endofunctors $\Sigma, \Gamma : \CatC \to \CatC$
and call them the \emph{algebraic signature} and \emph{scoped signature} respectively.

\paragraph{Syntax Endofunctor $P$}
Now our goal is to construct a monad $P : \CatC \to \CatC$ that models the
syntax of programs with algebraic operations in $\Sigma$ and non-algebraic
scoped operations in $\Gamma$.
First we construct its underlying endofunctor.
When $\CatC$ is $\Set$, the intuition for programs $P A$ is that they are terms
inductively built from the following inference rules:
\begin{gather*}
\mprset {sep=1em}
\inferrule{a \in A}{ \varr(a) \in P A }\quad\quad
\inferrule{o \in \Sigma n \\ k : n \to P A}{o(k) \in P A}\quad\quad
\inferrule{s \in \Gamma n \\ p : n \to P X \\ k : X \to P A} {\{s(p) ; k\} \in P A}
\end{gather*}
where $n$ ranges over finite sets and
$o \in \Sigma n$ represents an algebraic operation of $\vert n \vert $ arguments, and
similarly $s \in \Gamma n$ is a scoped operation that creates $\vert n \vert$ scopes.
The difference between algebraic and scoped operations is manifested by an
additional explicit continuation $k$ in the third rule, as it is \emph{not} the
case that sequentially composing $s(p)$ with $k$ equals $s(p ; k)$
like for algebraic operations,
so the continuation for scoped operations must be explicitly kept in the syntax.
When $\CatC$ is any lfp category, these rules translate to the following recursive
equation for the functor $P : \CatC \to \CatC$:
\begin{equation}\label{eq:e:recursive}
P A \cong A + \Sigma (P A) + \int^{X : \CatC} \coprod_{\CatC(X, P A)}  \Gamma (P X)
\end{equation}
where the existentially quantified $X$ in the third rule is translated to a
\emph{coend} $\int^{X : \CatC}$ in $\CatC$ \cite{MacLane1978}.
Moreover, the coend in (\ref{eq:e:recursive}) is isomorphic to $\Gamma (P (P
A))$ because by the coend formula of Kan extension, 
it exactly computes $\Lan{I}(\Gamma P)(P A)$, i.e.\ the left
Kan-extension of $\Gamma P$ along the identity functor $I : \CatC \to \CatC$,
and by definition $\Lan{I}(\Gamma P) = \Gamma P$.
Thus (\ref{eq:e:recursive}) is equivalent to
\begin{equation}\label{eq:e:recursive2}
P A \cong A + \Sigma (P A) + \Gamma (P (P A))
\end{equation}
which is exactly the \ensuremath{\Conid{Prog}\;\Sigma\;\Gamma} datatype that we saw in the Haskell
implementation (\ref{eq:prog:monad}).
To obtain a solution to (\ref{eq:e:recursive2}), we construct a (higher-order)
endofunctor $G : \CatEndof \to \CatEndof$ to represent the $G$rammar where
$\CatEndof$ is the category of finitary endofunctors on $\CatC$:
\begin{equation}\label{eq:g:equation}
G = \Id + \Sigma \hcomp {-} + \Gamma \hcomp {-} \hcomp {-}
\end{equation}
where $\Id : \CatC \to \CatC$ is the identity functor.
Then \autoref{lem:free:alg:simple} is applicable because $\CatEndof$ has all small colimits
since colimits in functor categories can be computed pointwise and $\CatC$ has all small
colimits.
Furthermore, $G$ preserves all filtered colimits, in particular colimits of
$\omega$-chains, because ${-} \hcomp {=} : \CatEndof \times \CatEndof \to \CatEndof$
is finitary following from direct verification.
Since initial algebras are precisely free algebras generated by the initial
object, by \autoref{lem:free:alg:simple}, there is an initial $G$-algebra
$\tuple{P : \CatEndof, \iniso : G P \to P}$ and $\iniso$ is an isomorphism. 
Thus $P$ obtained in this way is indeed a solution to (\ref{eq:e:recursive2})---%
the endofunctor modelling the syntax of programs with algebraic and scoped
operations.

\paragraph{Monadic Structure of $P$}

Next we equip the finitary endofunctor $P$ with a monad structure.
This can be done in several ways, either by the general result about
\emph{$\Sigma$-monoids} \cite{FiorePT99,Fiore09EqSys} in $\CatEndof$, or by
\cite[Theorem 4.3]{Matthes04Subs}, or by the following relatively
straightforward argument in \cite{PirogSWJ18}:
by the `diagonal rule' of computing initial algebras by Backhouse et al. \cite{Backhouse95},
$P = \mu G$ (\ref{eq:g:equation}) is isomorphic to $P' = \lfix{X}{\Id + \Sigma
\hcomp X + \Gamma \hcomp P \hcomp X}$.
Note that $P'$ is exactly $(\Sigma + \Gamma \hcomp P)^*$ as endofunctors by
\autoref{lem:free:alg:simple}, thus 
\begin{equation}\label{eq:e:from:free}
P \cong (\Sigma + \Gamma \hcomp P)^* : \CatEndof
\end{equation}
Then we equip $P$ with the same monad structure as the ordinary free monad
$(\Sigma + \Gamma \hcomp P)^*$.
The implementation in (\ref{eq:prog:monad}) is exactly this monad structure.

\subsection{Functorial Algebras of Scoped Operations}
\label{sec:fctAlgDef}

To interpret the monad $P$ (\ref{eq:e:recursive2}) modelling the syntax
of scoped operations, it is natural to expect that semantics is given by
$G$-algebras on $\CatEndof$ so that interpretation is then the catamorphisms
from $\mu G$ to $G$-algebras.
And following the adjoint-theoretic approach (\ref{rem:adjoint:theoretic}), 
we would like to have an adjunction 
\begin{tikzcd}[column sep=small]
	G\Alg & \CatC
	\arrow[""{name=0, anchor=center, inner sep=0}, shift right=2, from=1-2, to=1-1]
	\arrow[""{name=1, anchor=center, inner sep=0}, shift right=2, from=1-1, to=1-2]
	\arrow["{\small\dashv}"{anchor=center, rotate=-90}, draw=none, from=0, to=1]
\end{tikzcd}
such that the induced monad is isomorphic to $P$.
However, there seems no natural way to construct such an adjunction unless we
replace $G$-algebras with a slight extension of it, which we referred to
as \emph{functorial algebras}, as the notion for giving semantics to scoped
operations.
In the following, we first define functorial algebras formally
(\autoref{def:fctAlgs}) and then show the adjunction between the category of
functorial algebras and the base category (\autoref{thm:EisoT}), which
allows us to interpret $P$ with functorial algebras.

A functorial algebra is carried by an endofunctor $H : \CatC \to \CatC$
with additionally an object $X$ in $\CatC$.
The endofunctor $H$ also comes with a morphism $\alpha^G : G H \to H$ in
$\CatEndof$, and the object $X$ is
equipped with a morphism $\alpha^I : \Sigma X + \Gamma H X \to X$ in $\CatC$.
The intuition is that given a program of type $P X \cong X + \Sigma (P X) +
\Gamma (P (P X))$, the middle $P$ in $\Gamma P P$ corresponds to the part of
a program enclosed by some scoped operations (i.e.\ the $p$ in $\{s(p)
\ensuremath{\bind } k\}$), and this part of the program is interpreted by $H$ with $\alpha^G$.
After the enclosed part is interpreted, $\alpha^I$ interprets the outermost
layer of the program by $X$ with $\alpha^I$ in the same way as interpreting free
monads of algebraic operations.
More precisely, let $I : \CatEndof \times \CatC \to \CatC$ be a bi-functor such that
\footnote{The first argument $H$ to $I$ is written as
subscript so that we have a more compact notation $I_H^*$ when taking the free
monad of $I_H : \CatEndo$ with the first argument fixed.}
\begin{equation}\label{eq:I:def}
I_H X = \Sigma X + \Gamma (H X)
\hspace{4em}
I_\sigma f = \Sigma f + \Gamma (\sigma \hcomp f)
\end{equation}
for all $H : \CatEndof$ and $X : \CatC$ and all morphisms $\sigma : H \to H'$
and $f : X \to X'$.
Then we define an endofunctor $\Fn : \CatEC \to \CatEC$ such that
\begin{equation}\label{eq:fn:equation}
\Fn\abracket{H, X} = \abracket{G H, I_H X}
\end{equation}

\begin{definition}[Functorial Algebras]\label{def:fctAlgs}
A \emph{functorial algebra} is an object $\abracket{H,X}$ in $\CatEC$ paired 
with a structure map $\Fn\abracket{H,X} \to \abracket{H,X}$, or
equivalently it is a quadruple
\[
\big\langle{H : \CatEndof,\ \ \  X : \CatC},\ \ \ 
{\alpha^G : G H \to H,\ \ \ \alpha^I : \Sigma X + \Gamma (H X) \to X} \big\rangle
\]
where $G H = \Id + \Sigma \hcomp H + \Gamma \hcomp H \hcomp H$.
Morphisms between two functorial algebras $\tuple{H_1, X_1, \alpha^G_1, \alpha^I_1}$
and $\tuple{H_2, X_2, \alpha^G_2, \alpha^I_2}$ are pairs $\tuple{\sigma : H_1 \to
H_2,\,f : X_1 \to X_2}$ making the following diagrams commute:
\[
\begin{tikzcd}[row sep=12pt]
	{G H_1} & {H_1} \\
	{G H_2} & {H_2}
	\arrow["{\alpha^G_1}", from=1-1, to=1-2]
	\arrow["\alpha^G_2"', from=2-1, to=2-2]
	\arrow["{G \sigma}"', from=1-1, to=2-1]
	\arrow["\sigma", from=1-2, to=2-2]
\end{tikzcd}
\hspace{2cm}
\begin{tikzcd}[row sep=12pt]
	{\Sigma X_1 + \Gamma (H_1 X_1) } && {X_1} \\
	{\Sigma X_2 + \Gamma (H_2 X_2) } && {X_2}
	\arrow["{\alpha^I_1}", from=1-1, to=1-3]
	\arrow["f", from=1-3, to=2-3]
	\arrow["{\alpha^I_2}"', from=2-1, to=2-3]
	\arrow["{\Sigma f + \Gamma (\sigma \hcomp f)}"', from=1-1, to=2-1]
\end{tikzcd}
\]
Functorial algebras and their morphisms form a category $\Fn\Alg$.
\end{definition}

\begin{example}\label{ex:once:functorial}
We reformulate our programming example of nondeterministic choice with \ensuremath{\Varid{once}}
shown \autoref{ex:ndet:hdl} in the formal definition.
Let $\CatC = \Set$ in this example and $1 = \{\star\}$ be some singleton set.
We define signature endofunctors 
\[
\Sigma X = 1 + X \times X
\hspace{2cm}
\Gamma X = X
\]
so that $\Sigma$ represents nullary algebraic operation \ensuremath{\Varid{fail}} and binary
algebraic operation \ensuremath{\Varid{or}}, and $\Gamma$ represents the unary scoped operation
\ensuremath{\Varid{once}} that creates one scope.
Let $\ensuremath{\Conid{List}} : \Set \to \Set$ be the endofunctor mapping a set $X$ to the set
of finite lists with elements from $X$.
We define natural transformations $\alpha^\Sigma : \Sigma \circ \ensuremath{\Conid{List}} \to \ensuremath{\Conid{List}}$
and $\alpha^\Gamma : \Gamma \circ \ensuremath{\Conid{List}} \circ \ensuremath{\Conid{List}} \to \ensuremath{\Conid{List}}$ by
\[
\alpha^\Sigma_X(\iota_1\ \star) = \ensuremath{\Varid{nil}}
,\quad
\alpha^\Sigma_X(\iota_2\ \abracket{x, y}) = x \ensuremath{+\!\!\!+} y
,\quad
\alpha^\Gamma_X(\ensuremath{\Varid{nil}}) = \ensuremath{\Varid{nil}}
,\quad
\alpha^\Gamma_X(\ensuremath{\Varid{cons}\;\Varid{x}\;\Varid{xs}}) = x
\]
where \ensuremath{\Varid{nil}} is the empty list; $+\!\!+$ is list concatenation; and $\ensuremath{\Varid{cons}\;\Varid{x}\;\Varid{xs}}$ is the list with an element \ensuremath{\Varid{x}} in front of \ensuremath{\Varid{xs}}.
Then for any set $X$, $\abracket{\ensuremath{\Conid{List}}, \ensuremath{\Conid{List}\;\Conid{X}}}$ carries a functorial
algebra with structure maps
\begin{equation}
\label{eq:once:algebras}
\alpha^G = [\eta^\ensuremath{\Conid{List}}, \alpha^\Sigma, \alpha^\Gamma] : G \ensuremath{\Conid{List}} \to \ensuremath{\Conid{List}}
\qquad
\alpha^I = [\alpha^\Sigma_X, \alpha^\Gamma_X] : I_{\ensuremath{\Conid{List}}}X \to X
\end{equation}
where $\eta^\ensuremath{\Conid{List}} : \Id \to \ensuremath{\Conid{List}}$ wraps any element into a singleton list.
\end{example}

The last example exhibits that one can define a functorial algebra carried by
$\abracket{H, H X}$ from a $G$-algebra on $H : \CatEndof$ by simply
choosing the object component to be $H X$ for an arbitrary $X : \CatC$. 
In other words, there is a faithful functor $G\Alg \to \Fn\Alg$, which results
in functorial algebras that interpret the outermost layer of a program---the
part not enclosed by any scoped operation---in the same way as the inner
layers.
But in general, the object component of functorial algebras offers the
flexibility that the outermost layer can be interpreted differently from the
inner layers, as in the following example.
\begin{example}\label{ex:once:nat:alg}
Continuing \autoref{ex:once:functorial},
if one is only interested in the final number of possible outcomes,
then one can define a functorial algebra $\abracket{\ensuremath{\Conid{List}}, \Nat, \alpha^G,
\alpha^I}$ where $\alpha^G$ is (\ref{eq:once:algebras}) and
$\alpha^I(\iota_1\ (\iota_1 \star)) = 0$,
\[
\alpha^I(\iota_1\ (\iota_2 \abracket{x, y})) = x + y,\ \ 
\alpha^I(\iota_2\ \ensuremath{\Varid{nil}}) = 0,\ \ 
\alpha^I(\iota_2\ (\ensuremath{\Varid{cons}\;\Varid{n}\;\Varid{ns}})) = n
\]
\end{example}

\subsection{Interpreting with Functorial Algebras}\label{sec:adjFctAlg}
In the rest of this section we show how functorial algebras can be used to
interpret programs $P A$~(\ref{eq:e:recursive2}) with scoped operations.
We first construct a simple adjunction $\upC \dashv \downC$ 
between the base category $\CatC$ and $\CatEC$, which is then composed with the
free/forgetful adjunction $\Free_{\Fn} \dashv \Ul_{\Fn}$ between $\CatEC$ and
$\Fn\Alg$ for the functor $\Fn$ (\ref{eq:fn:equation}).
The resulting adjunction (\ref{eq:free:EndoC}) is proven to induce a monad $T$
isomorphic to $P$ (\autoref{thm:EisoT}), and by the adjoint-theoretic approach
to syntax and semantics (\ref{rem:adjoint:theoretic}), this adjunction
provides a means to interpret scoped operations modelled with the monad $P$
(\autoref{lem:fn:eval}).

First we define functor $\upC :  \CatC \to \CatEndof \times \CatC$ such that
$\upC X = \tuple{0, X}$ where $0 : \CatEndof$ is the initial endofunctor---%
the constant functor sending everything to the initial object in $\CatC$.
The functor $\upC$ is left adjoint to the projection functor $\downC : \CatEC
\to \CatC$ of the second component.

Then we would like to compose $\upC \dashv \downC$ with the free-forgetful
adjunction $\Free_\Fn \dashv \Ul_\Fn$ for the endofunctor $\Fn$
(\ref{eq:fn:equation}) on $\CatEC$, and the latter adjunction indeed exists.

\begin{lemma}\label{lem:fn:free:adj}
The endofunctor $\Fn$ (\ref{eq:fn:equation}) on $\CatEC$
has free algebras, i.e.\ there is a functor $\Free_\Fn : \CatEC \to \Fn\Alg$
left adjoint to the forgetful functor $\Ul_\Fn : \Fn\Alg \to \CatEC$.
\end{lemma}

\begin{proof}[Proof sketch]
It can be verified that $\Fn$ is finitary and then we apply
\autoref{lem:free:alg:simple}.
A detailed proof can be found in \refapp{app:proofs}.
\end{proof}

These two adjunctions are depicted in the following diagram:
\begin{equation}\label{eq:free:EndoC}
  \begin{tikzcd}[column sep=large]
    \Fn\Alg\rar[yshift=-1ex]{\bot}[swap]{\Ul_{\Fn}} 
    & \CatEC \lar[yshift=1ex][swap]{\Free_{\Fn}} \rar[yshift=-1ex]{\bot}[swap]{\downC} 
    & \lar[yshift=1ex][swap]{\upC}\CatC \arrow["T"', loop, distance=2em, in=35, out=325]
  \end{tikzcd}
\end{equation}
and we compose them to obtain an adjunction $\Free_\Fn \upC \dashv \downC \Ul_\Fn$
between $\Fn\Alg$ and $\CatC$, giving rise to a monad $T = \downC \Ul_\Fn
\Free_\Fn \upC$.
In the rest of this section, we prove that $T$ is isomorphic to $P$
(\ref{eq:e:recursive}) in the category of monads, which is crucial 
in this paper, since it allows us to interpret scoped operations modelled by the monad
$P$ with functorial algebras $\Fn\Alg$.

We first establish a technical lemma characterising the free $\Fn$-algebra on
the product category $\CatEC$ in terms of the free algebras in $\CatC$ and
$\CatEndof$.

\begin{lemma}\label{lem:freeJ}
There is a natural isomorphism between $\Free_{\Fn}$ and the following
\[
\FFn \abracket{H,X} = \left\langle
G^{*} H : \CatEndof,\ \ \ (I_{G^{*}H})^{*}X : \CatC,\ \ \ \op^{G^*}_H,\ \ \ \op^{(I_{G^*H})^*}_{X} 
\right\rangle
\]
where $\op^{G^*}_H : G (G^{*} H) \to G^{*} H$ and $ \op^{(I_{G^*H})^*}_{X} :
I_{G^*H} ((I_{G^{*}H})^{*}X) \to (I_{G^{*}H})^{*}X$ are the structure maps of
the free $G$-algebra and $I_{G^{*}H}$-algebra respectively.
\end{lemma}

\begin{proof}[Proof sketch]
It follows from the formula of computing free algebras from initial algebras
\autoref{lem:free:alg:simple}.
A detailed proof can be found in the appendix.
\end{proof}

\begin{theorem}\label{thm:EisoT}
Monads $P$ (\ref{eq:e:recursive2}) and $T$ (\ref{eq:free:EndoC}) are
isomorphic as monads.
\end{theorem}
\begin{proof}[Proof sketch]
By (\ref{eq:e:from:free}) and (\ref{eq:I:def}), $P \cong (\Sigma + \Gamma P)^* = (I_{P})^*$ as monads,
it is sufficient to show that $T \cong (I_{P})^*$ as monads, which follows
from \autoref{lem:freeJ} by careful calculation.
A detailed proof can be found in \refapp{app:proofs}.
\end{proof}

\begin{remark}\label{rem:notMonadic}
In general, the adjunction $\Free_\Fn \upC \dashv \downC \Ul_\Fn$ is \emph{not}
monadic since the right adjoint $\downC \Ul_\Fn$ does not reflect isomorphisms,
which is a necessary condition for it to be monadic by Beck's monadicity theorem
\cite{MacLane1978}.
This entails that the category $\Fn\Alg$ of functorial algebras is not
equivalent to the category of Eilenberg-Moore algebras.
Nonetheless, as we will see later in \autoref{sec:comparing}, functorial
algebras and Eilenberg-Moore algebras have the same expressive power
for interpreting scoped operations in the base category.
\end{remark}

The isomorphism established \autoref{thm:EisoT} enables us to interpret
programs modelled by the monad $P$ using functorial algebras following
(\ref{rem:adjoint:theoretic}):
for any functorial algebra $\abracket{H, X, \alpha^G, \alpha^I}$
(\autoref{def:fctAlgs}), and any morphism $g : A \to X$ in the base category
$\CatC$, there is a morphism 
\begin{equation}\label{eq:evalFctAlgDef}
\interp{\abracket{H, X, \alpha^G, \alpha^I}}{g} = \downC \Ul_\Fn
(\epsilon_{\abracket{H, X, \alpha^G, \alpha^I}} 
\vcomp \Free_\Fn \upC g) : TA \iso P A \to X
\end{equation}
which interprets programs $P A$ with the functorial algebra $\abracket{H, X,
\alpha^G, \alpha^I}$.
Furthermore, we can derive the following recursive formula
(\ref{eq:evalFctAlg}) for this interpretation morphism, which is exactly the
Haskell implementation in \autoref{fig:functorial_impl}.

\begin{theorem}[Interpreting with Functorial Algebras]
\label{lem:fn:eval}
For any functorial algebra $\alpha = \abracket{H, X, \alpha^G, \alpha^I}$
as in \autoref{def:fctAlgs}, and any morphism $g : A \to X$ for some $A$
in the base category $\CatC$, let $h = \cata{\alpha^G} : P \to H$ be
the catamorphism from the initial $G$-algebra $P$ to the $G$-algebra
$\alpha^G : G H \to H$.
The interpretation of $P A$ with this algebra $\alpha$ and $g$ satisfies
\begin{equation}\label{eq:evalFctAlg}
\interp{\alpha}{g} = [ g,\ \ 
  \alpha^I_\Sigma \vcomp \Sigma ({\interp{\alpha}{g}}),\ \ 
  \alpha^I_\Gamma \vcomp \Gamma h_X \vcomp \Gamma P (\interp{\alpha}{g})] \vcomp \inOp_A
\end{equation}
where $\inOp : P \to \Id + \Sigma \hcomp P + \Gamma \hcomp P \hcomp P$ is
the isomorphism between $P$ and $G P$; morphisms $\alpha^I_\Sigma = \alpha^I \vcomp
\iota_1 : \Sigma X \to X$ and $\alpha^I_\Gamma = \alpha^I \vcomp \iota_2 :
\Gamma H X \to X$ are the two components of $\alpha^I : \Sigma X + \Gamma H X
\to X$.
\end{theorem}
\begin{proof}[Proof sketch]
It can be calculated by plugging the formula of $\epsilon$ for the
adjunction $\downC \Ul_\Fn \dashv \Free_\Fn \upC$ (\autoref{lem:free:alg:simple})
into (\ref{eq:evalFctAlgDef}).
\end{proof}

To summarise, we have defined a notion of functorial algebras
that we use to handle scoped operations.
The heart of the development is the adjunction (\ref{eq:free:EndoC}) that
induces a monad isomorphic to the monad $P$ (\ref{eq:e:recursive2}) that models
the syntax of programs with scoped operations, following which 
we derive a recursive formula (\ref{eq:evalFctAlg}) that interprets programs
with functor algebras.
The formula is exactly the implementation in \autoref{fig:functorial_impl}:
the datatype \ensuremath{\Conid{EndoAlg}} represents the $\alpha^G$ in (\ref{eq:evalFctAlg});
datatype \ensuremath{\Conid{BaseAlg}} corresponds to $\alpha^I$;
function \ensuremath{\Varid{hcata}} implements $\cata{\alpha^G}$.

\section{Comparing the Models of Scoped Operations}
\label{sec:comparing}

Functorial algebras are not the only option for interpreting scoped operations.
In this section we compare functorial algebras with two other approaches, one being
Pir\'{o}g et al.\ \cite{PirogSWJ18}'s \emph{indexed algebras} and the other one being
\emph{Eilenberg-Moore} (EM) algebras of the monad $P$ (\ref{eq:e:recursive2}),
which simulate scoped operations with algebraic operations.
After a brief description of these two kinds of algebras, we compare them and
show that their expressive power is in fact equivalent.

\subsection{Interpreting Scoped Operations with Eilenberg-Moore Algebras}
\label{sec:EM-algebras}

In standard algebraic effects, handlers are just
$\Sigma$-algebras for some signature functor $\Sigma : \CatC \to \CatC$, and it
is well known that the category $\Sigma\Alg$ of $\Sigma$-algebras is
equivalent to the category $\CatC^{\Sigma^*}$ of EM algebras of the monad
$\Sigma^*$.
Thus handlers of algebraic operations are exactly EM algebras of the monad
$\Sigma^*$ modelling the syntax of algebraic operations.
This observation suggests that we may also use EM algebras of the monad $P$
(\ref{eq:e:recursive2}) as the notion of handlers for scoped operations.

\begin{lemma}
EM algebras of $P$ are equivalent to $(\Sigma + \Gamma \hcomp P)$-algebras.
In other words, an EM algebra of $P$ is equivalently a tuple
\begin{equation}\label{eq:e:em}
\tuple{X : \CatC,\ \alpha_\Sigma : \Sigma X \to X,\ \alpha_\Gamma : \Gamma (P X) \to X}  
\end{equation}
\end{lemma}

\begin{proof}
Recall that the monad structure of $P$ (\ref{eq:e:from:free}) is exactly the
monad structure of the free monad $(\Sigma + \Gamma \hcomp P)^*$, and therefore
they have the same EM algebras.
Moreover, EM algebras of $(\Sigma + \Gamma \hcomp P)^*$ are equivalent to
plain $(\Sigma + \Gamma \hcomp P)$-algebras by the monadicity of the free-forgetful
adjunction.
\end{proof}

Thus we obtain a way of interpreting scoped operations based on the
free-forgetful adjunction $\Free_{\Sigma + \Gamma \hcomp P} \dashv
\Ul_{\Sigma + \Gamma \hcomp P}$:
given an EM algebra $\alpha = \abracket{X, \alpha_\Sigma, \alpha_\Gamma}$ of
$P$ as in (\ref{eq:e:em}), then for any $A : \CatC$ and morphism $g : A \to X$,
the interpretation of $P A$ by $g$ and this EM algebra is
\begin{equation}\label{eq:evalEMAlgDef}
\interp{\alpha}{g}
= \Ul_{\Sigma + \Gamma \hcomp P}(\epsilon_{\alpha}
\vcomp \Free_{\Sigma + \Gamma \hcomp P}\ g)
: P A \iso (\Sigma + \Gamma \hcomp P)^* A \to X
\end{equation}
The formula (\ref{eq:evalEMAlgDef}) can also be turned into a recursive form:
\begin{equation}
\interp{\alpha}{g} = 
  [g,\ \ 
  \alpha_\Sigma \vcomp \Sigma (\interp{\alpha}{g}),\ \ 
  \alpha_\Gamma \vcomp \Gamma P (\interp{\alpha}{g})] \vcomp \inOp_A
\end{equation}
that suits implementation (see \refapp{app:em} for more details).

Interpreting scoped operation with EM algebras can be understood as simulating 
scoped operations with algebraic operations and general recursion:
a signature $(\Sigma, \Gamma)$ of algebraic-and-scoped operations is simulated
by a signature $(\Sigma + \Gamma \circ P)$ of algebraic operations where
$P$ is recursively given by $(\Sigma + \Gamma \circ P)^*$.
In this way, one can simulate scoped operation in languages implementing
algebraic effects that allow signatures of operation to be recursive, such as
\cite{Hillerstrom2018Shallow,BP14Eff,Koka17}, but not the original design by
Plotkin and Pretnar \cite{PlotPret13Hand}, which requires signatures of operations to mention only
some \emph{base types}.

The downside of this simulating approach is that the denotational semantics of the
language becomes more complex and usually involves solving some domain-theoretic
recursive equations, like in \cite{BP14Eff}.
Moreover, this approach typically requires handlers to be defined with general
recursion, which obscures the inherent structure of scoped operations, making
reasoning about handlers of scoped operations more difficult.

\subsection{Indexed Algebras of Scoped Effects}\label{sec:indexedAlgOfE}

Indexed algebras of scoped operations by Pir\'{o}g et al.\ \cite{PirogSWJ18} are yet another
way of interpreting scoped operations.
They are based on the following adjunction:
\begin{equation}\label{eq:IxAdj}
\begin{tikzcd}
	\Ix\Alg & \CatCN & {\mathbb{C}}
	\arrow[""{name=0, anchor=center, inner sep=0}, "\upharpoonright"', shift right=2, from=1-3, to=1-2]
	\arrow[""{name=1, anchor=center, inner sep=0}, "\downharpoonleft"', shift right=2, from=1-2, to=1-3]
	\arrow[""{name=2, anchor=center, inner sep=0}, "{\Ul_\Ix}"', shift right=2, from=1-1, to=1-2]
	\arrow[""{name=3, anchor=center, inner sep=0}, "{\Free_\Ix}"', shift right=2, from=1-2, to=1-1]
	\arrow["\dashv"{anchor=center, rotate=-90}, draw=none, from=0, to=1]
	\arrow["\dashv"{anchor=center, rotate=-90}, draw=none, from=3, to=2]
\end{tikzcd}
\end{equation}
where $\CatCN$ is the functor category from the discrete category $\CatN$ of 
natural numbers to the base category $\CatC$.
That is to say, an object in $\CatCN$ is a family of objects $A_i$ in $\CatC$ 
indexed by natural numbers $i \in \CatN$, and a morphism $\tau : A \to B$ in
$\CatCN$ is a family of morphisms $\tau_i : A_i \to B_i$ in $\CatC$ with 
no coherence conditions between the levels.
An endofunctor $\Ix : \CatCN \to \CatCN$ is defined to characterise indexed algebras:
\begin{equation*}
\Ix A = \Sigma \hcomp A + \Gamma \hcomp (\previous A) + (\later A)
\end{equation*}
where $\previous$ and $\later$ are functors $\CatCN \to \CatCN$ \emph{shifting
indices} such that $(\previous A)_i = A_{i+1}$ and $(\later A)_0 = 0$ and
$(\later A)_{i+1} = A_i$.
Then objects in $\Ix\Alg$ are called \emph{indexed algebras}.
Furthermore, since a morphism $(\later A) \to A$ is in bijection with $A \to
(\previous A)$, an indexed algebra can be given by the following tuple:
\begin{equation}\label{eq:indexed:alg}
\abracket{A : \CatCN,\ \ 
a : \Sigma \circ A \to A,\ \ 
d : \Gamma (\previous A) \to A,\ \ 
p : A \to \previous A
}
\end{equation}
The operational intuition for it is that the carrier $A_i$ at level $i$
interprets the part of syntax enclosed by $i$ layers of scopes, and when
interpreting a scoped operation $\Gamma (P (P X))$ at layer~$i$, the part of
syntax outside the scope is first interpreted, resulting in $\Gamma (P A_i)$,
and then the indexed algebra provides a way $p$ to \emph{p}romote the carrier
to the next level, resulting in $\Gamma (P A_{i+1})$.
After the inner layer is also interpreted as $\Gamma A_{i+1}$, the indexed
algebra provides a way $d$ to \emph{d}emote the carrier, producing $A_i$ again. 
Additionally the morphism $a$ interprets ordinary \emph{a}lgebraic operations.

\begin{example}\label{ex:onceIx}
\autoref{ex:once:functorial} for nondeterministic choice with \ensuremath{\Varid{once}} can be
re-expressed with an indexed algebra as follows.
For any set $X$, we define an indexed object $A : \CatCN$ by $A_0 = \ensuremath{\Conid{List}}\ X$ and
$A_{i+1} = \ensuremath{\Conid{List}}\ A_i$.
The object $A$ carries an indexed algebra with the following structure maps:
for all $i \in \Nat$, $a_i(\iota_1\ \star) = \ensuremath{\Varid{nil}}$ and
\[
a_i(\iota_2\ \abracket{x, y}) = x \ensuremath{+\!\!\!+} y,
\ \ 
d_i(\ensuremath{\Varid{nil}}) = \ensuremath{\Varid{nil}},
\ \ 
d_i(\ensuremath{\Varid{cons}\;\Varid{x}\;\Varid{xs}}) = x,
\ \ 
p_i(x) = \ensuremath{\Varid{cons}\;\Varid{x}\;\Varid{nil}}
\]
\end{example}

The adjunction $\Free_\Ix \dashv \Ul_\Ix$ in (\ref{eq:IxAdj}) is the free-forgetful
adjunction for $\Ix$ on $\CatCN$.
The other adjunction $\upharpoonright \dashv \downharpoonleft$ is given by
$\downharpoonleft\! A = A_0$, $(\upharpoonright\! X)_0 = X$, and $(\upharpoonright\!
X)_{i+1} = 0$ for all $i \in \Nat$.
Importantly, Pir\'{o}g et al.\ \cite{PirogSWJ18} show that the monad induced by the adjunction (\ref{eq:IxAdj})
is isomorphic to monad $P$ (\ref{eq:e:recursive2}), thus indexed algebras can also
be used to interpret scoped operations
\begin{equation}\label{eq:evalIxDef}
\interp{\abracket{A, a, d, p}}{g} = \,\downharpoonleft \Ul_\Ix (\epsilon_{\abracket{A, a, d, p}} \vcomp
\Free_\Ix \upharpoonright g) 
\end{equation}
in the same way as what we do for functorial algebras in
\autoref{sec:adjFctAlg}.
Interpreting with indexed algebras can also be implemented in Haskell with
GHC's \texttt{DataKinds} extension for type-level natural numbers (which can be
found in \refapp{fig:hfold}).

\subsection{Comparison of Resolutions}
\label{sec:resolutions}

Now we come back to the real subject of this section---comparing the expressivity
of the three ways for interpreting scoped operations.
Specifically, we construct \emph{comparison functors} between the respective
categories of the three kinds of algebras, which translate one kind of algebras
to another in a way \emph{preserving the induced interpretation} in the base category.
Categorically, the three kinds of algebras correspond to three \emph{resolutions} of the 
monad $P$, which form a category of resolutions (\autoref{def:resCmp}) with 
comparison functors as morphisms.
In this category, the Eilenberg-Moore resolution is the terminal object, and
thus it automatically gives us comparison functors translating other kinds of
algebras to EM algebras.  
To complete the circle of translations, we then construct comparison functors
$\KL : \CatC^P \to \Fn\Alg$ translating EM algebras to functorial ones
(\autoref{sec:EMtoFn}) and $\KK : \Fn\Alg \to \Ix\Alg$ translating functorial
algebras to indexed ones (\autoref{sec:FntoIx}). 
The situation is pictured in \autoref{fig:cmpFunctors}.

\begin{definition}[Resolutions and Comparison Functors \cite{LS86Intro}]\label{def:resCmp}
Given a monad $M$ on $\CatC$, the category $\CatRes[M]$
of \emph{resolutions} of $M$ has as objects adjunctions $\tuple{\CatD, L \dashv R :
\CatD \to \CatC, \eta, \epsilon}$ such that the induced monad $RL$ is $M$.
A morphism from a resolution $\abracket{\CatD, L \dashv R, \eta, \epsilon}$ to
$\abracket{\CatD', L' \dashv R', \eta', \epsilon'}$ is a functor $K : \CatD \to \CatD'$,
called a \emph{comparison functor},
such that it commutes with the left and right adjoints, i.e.\ 
$K L = L'$ and $R' K = R$.
\end{definition}

We have seen adjunctions for indexed algebras, EM algebras and functorial
algebras respectively, each inducing the monad $P$ up to isomorphism, so each
of them can be identified with an object in the category $\CatRes$.
For each resolution $\abracket{\CatD, L, R, \eta, \epsilon}$, 
we have been using the objects $D$ in $\CatD$ to interpret scoped operations modelled
by $P$:
for any morphism $g : A \to RD$ in $\CatC$, the interpretation of $P A$ by
$D$ and $g$ is
$\interp{D}{g} = R(\epsilon_{D} \vcomp Lg) : P A = RLA \to RD$.
Crucially, we show that interpretations are preserved by comparison functors.

\begin{lemma}[Preservation of Interpretation]\label{lem:CFPreservesI}
Let $K : \CatD \to \CatD'$ be any comparison functor between resolutions
$\abracket{\CatD, L, R, \eta, \epsilon}$ and $\abracket{\CatD', L', R', \eta',
\epsilon'}$ of some monad $M : \CatC \to \CatC$.
For any object $D$ in $\CatD$ and any $g : A \to RD$ in $\CatC$, 
\begin{equation}
\interp{D}{g} = \interp{KD}{g} : MA \to RD (= R'KD)
\end{equation}
where each side interprets $MA$ using $L \dashv R$ and $L' \dashv R'$
respectively.
\end{lemma}
\begin{proof}
Since $L \dashv R$ and $L' \dashv R'$ induce the same monad, their
unit must coincide $\eta = \eta'$.
Together with the commutativity properties $KL = L'$ and $R'K = R$,
it makes a comparison functor a special case of a \emph{map of adjunctions}.
Then by Proposition 1 in \cite[page~99]{MacLane1978},
it holds that $K \epsilon = \epsilon'K$, and we have
\begin{align*}
  \interp{KD}{g} &= R'(\epsilon'_{KD} \vcomp L'g) = R'(K\epsilon_D \vcomp L'g) \\
                 &= R \epsilon_D \vcomp R'L'g = R \epsilon_D \vcomp RLg = \interp{D}{g}
\end{align*}
\end{proof}

This lemma implies that if there is a comparison functor $K$ from some
resolution $L \dashv R : \CatD \to \CatC$ to $L' \dashv R' : \CatD' \to \CatC$
of the monad $P$, then $K$ can \emph{translate} a $\CatD$ object to a $\CatD'$
object that preserves the induced interpretation.
Thus the expressive power of $\CatD$ for interpreting $P$ is not greater than
$\CatD'$, in the sense that every $\interp{D}{g}$ that one can obtain
from $D$ in $\CatD$ can also be obtained by an algebra $K D$ in $\CatD'$.
Thus the three kinds of algebras for interpreting scoped operations have 
the same expressivity if we can construct a circle of comparison functors
between their categories, which is what we do in the following.

\paragraph{Translating to EM Algebras}
\label{sec:toEM}
As shown in \cite{MacLane1978}, an important property of the Eilenberg-Moore
adjunction is that it is the terminal object in the category $\CatRes[M]$ for
any monad $M$, which means that there \emph{uniquely exists} a comparison
functor from \emph{every} resolution to the Eilenberg-Moore resolution.
Specifically, given a resolution $\abracket{\CatD, L, R, \eta, \counit}$ of a
monad $M$, the unique comparison functor $K$ from $\CatD$ to the category
$\CatC^M$ of the Eilenberg-Moore algebras is
\[
KD = \big(M(RD) = RLRD \xrightarrow{R\epsilon_D} RD\big)
\qquad\text{ and }\qquad
K(D \xrightarrow{f} D') = Rf
\]

\begin{lemma}
There uniquely exist comparison functors $\Kb : \Ix\Alg \to \CatC^P$ 
and $\Kb[Fn] : \Fn\Alg \to \CatC^P$ from the resolutions of indexed algebras
and functorial algebras to the resolution of EM algebras.
\end{lemma}

\subsection{Translating EM Algebras to Functorial Algebras}
\label{sec:EMtoFn}

Now we construct a comparison functor $\KL : \CatC^P \to \Fn\Alg$ translating
EM algebras to functorial ones.
The idea is straightforward: 
given an EM algebra $X$, we map it to the functorial algebra with $X$ for
interpreting the outermost layer and the functor $P$ for interpreting the inner
layers, which essentially leaves the inner layers uninterpreted before they get
to the outermost layer.

Since $\CatC^P$ is isomorphic to $(\Sigma + \Gamma \hcomp P)\Alg$,
we can define $\KL$ on $(\Sigma + \Gamma \hcomp P)$-algebras instead.
Given any $\abracket{X : \CatC, \alpha : (\Sigma + \Gamma \hcomp P) X \to X}$,
it is mapped by $\KL$ to the functorial algebra
\[
\abracket{P,\ X,\ \iniso : G P \to P,\ \alpha : (\Sigma + \Gamma \hcomp P) X \to X}
\]
and for any morphism $f$ in $(\Sigma + \Gamma \hcomp P)\Alg$, it is mapped to
$\abracket{\identity_P, f}$.
To show $\KL$ is a comparison functor, we only need to show that it commutes
with the left and right adjoints of both resolutions.
Details can be found in \refapp{app:proofs}.

\begin{lemma}\label{lem:KLisCMP}
Functor $\KL$ is a comparison functor from the Eilenberg-Moore resolution of 
$P$ to the resolution $\Free_\Fn \upC \dashv \downC \Ul_\Fn$ of functorial
algebras.
\end{lemma}

\subsection{Translating Functorial Algebras to Indexed Algebras}
\label{sec:FntoIx}
At this point we have comparison functors $\Ix\Alg \xrightarrow{\Kb} \CatC^P
\xrightarrow{\KL} \Fn\Alg$.
To complete the circle of translations, we construct a comparison functor
$\KK : \Fn\Alg \to \Ix\Alg$ in this subsection.
The idea of this translation is that given a functorial algebra carried
by endofunctor $H : \CatEndo$ and object $X : \CatC$, we map it to an indexed
algebra by iterating the endofunctor $H$ on $X$.
More precisely, $\KK : \Fn\Alg \to \Ix\Alg$ maps a functorial algebra 
\[\abracket{H : \CatEndo,\ X : \CatC,\ \alpha^G : \Id + \Sigma \hcomp H + \Gamma \hcomp H \hcomp H \to H,
\ \alpha^I : \Sigma X + \Gamma H X \to X}\]
to an indexed algebra carried by $A : \CatCN$ such that $A_i = H^{i} X$, i.e.\ 
iterating $H$ $i$-times on $X$.
The structure maps of this indexed algebra $\abracket{a : \Sigma A \to A,\ 
d : \Gamma (\previous A) \to A,\ p : A \to (\previous A)}$ are given by
\begin{align*}
a_0 =\, & (\alpha^I \vcomp \iota_1) : \Sigma X \to X  
& a_{i+1} =\, & (\alpha^G_{H^{i}X} \vcomp \iota_2) : \Sigma H H^i X \to H^{i+1} X \\
d_0 =\, & (\alpha^I \vcomp \iota_2) : \Gamma H X \to X
& d_{i+1} =\, &(\alpha^G_{H^{i}X} \vcomp \iota_3) : \Gamma H H H^{i} X \to H^{i+1} X
\end{align*}
and $p_i = \alpha^G_{H^i X} \vcomp \iota_1 : H^i X \to HH^i X$.
On morphisms, $\KK$ maps a morphism $\abracket{\tau : H \to H', f : X \to X'}$ in $\Fn\Alg$
to $\sigma : H^{i} X \to H'^{i}X'$ in $\Ix\Alg$ such that
$\sigma_0 = f$ and $\sigma_{i+1} = \tau \hcomp \sigma_i $ where $\hcomp$ is
horizontal composition.

\begin{lemma}\label{lem:KK:is:comp}
$\KK$ is a comparison functor from the resolution 
$\Free_\Fn \upC \dashv \downC \Ul_\Fn$ of functorial algebras to the resolution
$\Free_\Ix \mathop{\upharpoonright} \dashv \mathop{\downharpoonleft} \Ul_\Ix$
of indexed algebras.
\end{lemma}
\begin{proof}[Proof sketch]
Again we only need to show the required commutativities for $\KK$ to be a comparison
functor:
$
\downC \Ul_\Fn \iso \mathop{\downharpoonleft} \Ul_\Ix \KK
$
and
$
\KK \Free_\Fn \upC \iso \Free_\Ix \mathop{\upharpoonright} 
$.
The first one is easy, and the second one follows from Pir\'{o}g et al.\ \cite{PirogSWJ18}'s
explicit characterisation of $\Free_\Ix \mathop{\upharpoonright} X$.
More details can be found in \refapp{app:proofs}.
\end{proof}

Since comparison functors preserve interpretation (\autoref{lem:CFPreservesI}),
the lemma above implies that the expressivity of functorial algebras is not
greater than indexed ones.
Together with the comparison functors defined earlier, we conclude that the
three kinds of algebras---indexed, functorial and Eilenberg-Moore algebras---have the
same expressivity for interpreting scoped operations.
\autoref{fig:cmpFunctors} summarises the comparison functors and resolutions that
we have studied.

\begin{remark}
Although the three kinds of algebras have the same expressivity in theory,
they structure the interpretation of scoped operations in different ways: 
EM algebras impose no constraint on how the part of syntax enclosed by scopes
is handled; 
indexed algebras demand them to be handled layer by layer but impose no coherent
conditions between the layers;
functorial algebras additionally force all inner layers must be handled in
a uniform way by an endofunctor.

On the whole, it is a trade-off \emph{simplicity} and \emph{structuredness}:
EM algebras are the simplest for implementation, whereas the structuredness of
functorial algebras make them easier to reason about.
This is another instance of the preference for structured programming over
unstructured language features, in the same way as structured loops being
favoured over \texttt{goto}, although they have the same expressivity in theory
\cite{Dijkstra68}.
\end{remark}

\section{Fusion Laws of Interpretation}
\label{sec:hybrid-adjoint}

A crucial advantage of the adjoint-theoretic approach to syntax and semantics is
that the naturality of an adjunction directly offers \emph{fusion laws} of
interpretation that fuse a morphism after an interpretation into a single
interpretation, which have proven to be a powerful tool for reasoning and
optimisation
\cite{Wad88Def,Tak95Fusion,Coutts07Stream,HHJ11The,WS15Fusion,YangWu21}.
In this section, we present the fusion law for functorial algebras.

\subsection{Fusion Laws of Interpretation}\label{sec:fusion}

Recall that given any resolution $L \dashv R$ with counit $\epsilon$ of some
monad $M : \CatC \to \CatC$ where $L : \CatC \to \CatD$, for any $g : A \to
RD$, we have an interpretation morphism
\[
\interp{D}{g} = R ( \epsilon_D \vcomp Lg) : M A \to RD
\]
Then whenever we have a morphism in the form of $(f \vcomp \interp{D}{g})$---an
interpretation followed by some morphism---the following \emph{fusion law}
allows one to fuse it into a single interpretation morphism.

\begin{lemma}[Interpretation Fusion]\label{lem:fusion}
Assume $L \dashv R$ is a resolution of monad $M : \CatC \to \CatC$ where $L :
\CatC \to \CatD$.
For every $D : \CatD$, $g : A \to R D$ and $f : R D \to X$, 
if there is some $D'$ and $h : D \to D'$ in $\CatD$ such that $R D' = X$ and
$R h = f$, then
\begin{equation}\label{eq:fusion}
f \vcomp \interp{D}{g} = \interp{D'}{(f \vcomp g)}
\end{equation}
\end{lemma}
\begin{proof}
We have $f \vcomp \interp{D}{g} = R h \vcomp R (\epsilon_D \vcomp L g) = R (h \vcomp \epsilon_D \vcomp L g)$.
Then by the naturality of the counit $\epsilon$, $R (h \vcomp \epsilon_D \vcomp
L g) = R (\epsilon_{D'} \vcomp L (R h \vcomp g)) = \interp{D'}{(f \vcomp g)}$.
\end{proof}

Applying the lemma to the three resolutions of $P$ gives us three fusion laws:
for any $D : \CatD$ where $\CatD \in \{\Ix\Alg, \Fn\Alg, \CatC^P\}$, one
can fuse $f \vcomp \interp{D}{g}$ into a single interpretation if one can
make $f$ a $\CatD$-homomorphism.
Particularly, the following is the fusion law for functorial algebras.

\begin{corollary*}[Fusion Law for Functorial Algebras]\label{lem:fn:fusion}
Let $\hat{\alpha_1} = \tuple{H, X_1, \alpha^G_1, \alpha^I_2}$ be a functorial
algebra (\autoref{def:fctAlgs}) and $g : A \to X_1$, $f : X_1 \to X_2$ be any
morphisms in $\CatC$.
If there is a functorial algebra $\hat{\alpha_2} = \tuple{H_2, X_2,
\alpha^G_2, \alpha^I_2}$ and a functorial algebra morphism $\tuple{\sigma : H_1
\to H_2, h : X_1 \to X_2}$, then
\[
f \vcomp \interp{\hat{\alpha_1}} g = \interp{\hat{\alpha_2}} (f \vcomp g)
\]
\end{corollary*}

\begin{example}
Let $\hat{\alpha}$ be the functorial algebra of nondeterminism with \ensuremath{\Varid{once}} in
\autoref{ex:once:functorial} and $\ensuremath{\Varid{len}} : \ensuremath{\Conid{List}}\,A \to \Nat$ be the function
mapping a list to its length.
Then using the fusion law, $\ensuremath{\Varid{len}} \vcomp \interp{\hat{\alpha}} g =
\interp{\hat{\beta}} (\ensuremath{\Varid{len}} \vcomp g)$ if we can find a suitable functorial
algebra $\hat{\beta} : \Fn\Alg $ and $h : \hat{\alpha} \to \hat{\beta}$ s.t.\
$\downarrow\Ul_\Fn h = \ensuremath{\Varid{len}}$.
In fact, a suitable $\hat{\beta}$ is just the functorial algebra in
\autoref{ex:once:nat:alg} and $h = \abracket{\identity, \ensuremath{\Varid{len}}}$.
\end{example}

\begin{example}\label{ex:hfold}
Although Pir\'{o}g et al.\ \cite{PirogSWJ18} propose the adjunction (\ref{eq:IxAdj}) to interpret
scoped operations with indexed algebras, their Haskell implementation is not a 
faithful implementation of the interpretation morphism (\ref{eq:evalIxDef}), but
rather a more efficient one skipping the step of transforming $P$ to the
isomorphic free indexed algebra $(\downharpoonleft \Ul_\Ix \Free_\Ix
\upharpoonright)$.
However, it is previously unclear whether this implementation indeed coincides
with the interpretation morphism (\ref{eq:evalIxDef}) due to the discrepancy
between the syntax monad $P$ and indexed algebras.

This issue is in fact one of the original motivations for us to develop functorial
algebras---a way to interpret $P$ that directly follows the syntactic structure.
Using the comparison functors to transform between indexed and functorial algebras,
we can reason about Pir\'{o}g et al.\ \cite{PirogSWJ18}'s implementation with functorial
algebras, and its correctness can be established using fusion laws.
Due to limited space, this extended case study is in \refapp{app:hfold}.
\end{example}

\section{Related Work}
\label{sec:related}

The most closely related work is that of Pir\'{o}g et al.\ \cite{PirogSWJ18} on
categorical models of scoped effects. That work in turn builds on Wu et al.\ \cite{WuSH14}
who introduced the notion of scoped effects after identifying modularity
problems with using algebraic effect handlers for catching
exceptions~\cite{PlotPret13Hand}.
Scoped effects have found their way into several Haskell implementations of
algebraic effects and handlers~\cite{fused-effects,polysemy,eff}.

\paragraph{Effect Handlers and Modularity} Spivey~\cite{Spivey1990}, Moggi~\cite{Moggi89a} and
Wadler~\cite{Wadler:1990} initiated monads for modeling and
programming with computational effects. Soon after, the desire arose to define
complex monads by combining modular definitions of individual
effects~\cite{steele,jones93}, and monad transformers were developed to meet
this need~\cite{liang95monad}.
Yet, several years later, algebraic effects were proposed as an alternative more
structured approach for defining and combining computational
effects~\cite{fossacs/PlotkinP02,GenericEffects,HPP06Combine}.  The addition of
handlers~\cite{PlotPret13Hand} has made them practical for implementation and
many languages and libraries have been developed since.  
Schrijvers et al.\ \cite{SchrijversPWJ19}
have characterized modular handlers by means of modular carriers, and shown that
they correspond to a subclass of monad transformers. 
Forster et al.\ \cite{ForsterKLP19} have
also shown that algebraic effects, monads and delimited control are
macro-expressible in terms of each other in an untyped language but not in a
simply typed language.

Scoped operations are generally not algebraic operations in the original design
of algebraic effects \cite{fossacs/PlotkinP02}, but as we have seen in
\autoref{sec:EM-algebras}, an alternative view on Eilenberg-Moore algebras of
scoped operations is regarding them as handlers of \emph{algebraic} operations
of signature $\Sigma + \Gamma P$.
However, the functor $\Sigma + \Gamma P$ involves the type $P$ modelling
computations, and thus it is not a valid signature of algebraic effects in the
original design of effect handlers \cite{PlotPret09Hand,PlotPret13Hand}, in
which the signature of algebraic effects can only be built from some \emph{base
types} to avoid the interdependence of the denotations of signature functors
and computations.
In spite of that, many later implementations of effect handlers such as
\textsc{Eff}~\cite{BP14Eff}, \textsc{Koka}~\cite{Koka17} and \textsc{Frank}~
\cite{Lin17DoBe} do not impose this restriction on signature functors (at the
cost that the denotational semantics involves solving recursive
domain-theoretic equations), and thus scoped operations can be implemented in
these languages with EM algebras as handlers.  

Other variations of scoped effects have been suggested.
Recently, Poulsen et al.\ \cite{PoulsenVS21} and {van den Berg} et al.\
\cite{berg2021latent} have proposed a notion of \emph{staged} or \emph{latent}
effect, which is a variant of scoped effects, for modelling the deferred
execution of computations inside lambda abstractions and similar constructs. 
Ahman and Pretnar \cite{AhmanP21} investigate \emph{asynchronous effects}, and they note that
interrupt handlers are in fact scoped operations.
We have not yet investigated this in our framework, but it will be an
interesting use case.

\paragraph{Abstract Syntax}
This work focusses on the problem of abstract
syntax and semantics of programs. 
The practical benefit of abstract syntax is that it allows for \emph{generic
programming} in languages like Haskell that have support for, e.g.\ type
classes, \GADTs~\cite{JohannG08} and so on. As an example, Swierstra \cite{Swierstra08}
showed that it is possible to modularly create compilers by formalising syntax
in Haskell.

The problem of formalising abstract syntax categorically for operations with
variable binding was first addressed by Fiore et al.\  \cite{FiorePT99,FioreT01}.
Subsequently,~Ghani and Uustalu \cite{GhaniUH06} model the abstract
syntax of explicit substitutions as an initial algebra in the endofunctor
category and show that it is a monad.
Pir\'{o}g et al.\ \cite{PirogSWJ18} and this paper use a monad $P$, which is a slight
generalisation of the monad of explicit substitutions, to model the syntax of
scoped operations.
The datatype underlying $P$ is an instance of \emph{nested datatypes} studied
by Bird and Paterson \cite{BirdP99} and Johann and Ghani \cite{JohannG07}.

In this paper we have not treated \emph{equations} on effectful
operations, which are both theoretically and practically important.
Plotkin and Power \cite{fossacs/PlotkinP02} show that theories of various effects with suitable
equations \emph{determine} their corresponding monads, and later
Hyland et al.\ \cite{HPP06Combine} show that certain combinations of effect theories are
equivalent to monad transformers.
Equations are also used for reasoning about programs with algebraic
effects and handlers \cite{PP08Logic,YangWu21,Kiselyov21}.
Possible ways to extend scoped effects with equations include the approach in
\cite{Kelly82} (\autoref{rem:equations}), the categorical framework of
\emph{equational systems} \cite{Fiore09EqSys}, second order Lawvere theories
\cite{Fiore10Sec}, and syntactic frameworks like \cite{Staton13}.

\section{Conclusion}
\label{sec:conclusion}

The motivation of this work is to develop a modular approach to the syntax and
semantics of scoped operations.
We believe our proposal, functorial algebras, is at a sweet spot in the
trade-off between structuredness and simplicity, allowing practical examples of
scoped operations to be programmed and reasoned about naturally, and
implementable in modern functional languages such as Haskell and OCaml.
We put our model and two other models for interpreting scoped effects in
the same framework of resolutions of the monad modelling syntax. 
By constructing interpretation-preserving functors between the three kinds of
algebras, we showed that they have equivalent expressivity for interpreting
scoped effects, although they form non-equivalent categories.
The uniform theoretical framework also induces fusion laws of interpretation
in a straightforward way.

There are two strains of work that should be pursued from here.
The first one would be investigating ways to compose algebras of scoped operations.
The second one would be the design of a language supporting handlers of scoped
operations natively and its type system and operational semantics.
\section*{Acknowledgements}
This work is supported by EPSRC grant number EP/S028129/1 on `Scoped
Contextual Operations and Effects', by FWO project G095917N, and KU Leuven
project C14/20/079.
The authors would like to thank the anonymous reviewers for their constructive
feedback.

\bibliography{../references}

\clearpage

\appendix
\section{Figures}\label{app:figures}

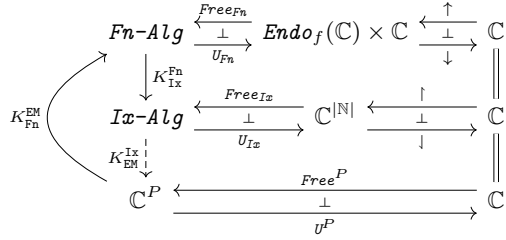
\begin{figure}[h]
\begin{center}
\begin{tikzcd}[row sep=15pt]
	\Fn\Alg & {\CatEC} & \CatC \\
	\Ix\Alg & {\CatCN} & \CatC \\
	{\CatC^P} && \CatC

	\arrow["\KK", from=1-1, to=2-1]
	\arrow["\Kb"', dashed, from=2-1, to=3-1]
	\arrow[""{name=0, anchor=center, inner sep=0}, "{\Free_\Fn}"', shift right=2, from=1-2, to=1-1]
	\arrow[""{name=1, anchor=center, inner sep=0}, "\Ul_\Fn"', shift right=2, from=1-1, to=1-2]
	\arrow[""{name=2, anchor=center, inner sep=0}, "{\Free_\Ix}"', shift right=2, from=2-2, to=2-1]
	\arrow[""{name=3, anchor=center, inner sep=0}, "\Ul_\Ix"', shift right=2, from=2-1, to=2-2]
	\arrow[""{name=4, anchor=center, inner sep=0}, "\upharpoonright"', shift right=2, from=2-3, to=2-2]
	\arrow[""{name=5, anchor=center, inner sep=0}, "\downharpoonleft"', shift right=2, from=2-2, to=2-3]
	\arrow[""{name=6, anchor=center, inner sep=0}, "\upC"', shift right=2, from=1-3, to=1-2]
	\arrow[""{name=7, anchor=center, inner sep=0}, "\downC"', shift right=2, from=1-2, to=1-3]
	\arrow[""{name=8, anchor=center, inner sep=0}, "{\Ul^P}"', shift right=2, from=3-1, to=3-3]
	\arrow[""{name=9, anchor=center, inner sep=0}, "{\Free^P}"', shift right=2, from=3-3, to=3-1]
	\arrow[Rightarrow, no head, from=1-3, to=2-3]
	\arrow[Rightarrow, no head, from=2-3, to=3-3]
	\arrow["\KL", shift left=2, curve={height=-30pt}, from=3-1, to=1-1]
	\arrow["\dashv"{anchor=center, rotate=-90}, draw=none, from=9, to=8]
	\arrow["\dashv"{anchor=center, rotate=-90}, draw=none, from=6, to=7]
	\arrow["\dashv"{anchor=center, rotate=-90}, draw=none, from=4, to=5]
	\arrow["\dashv"{anchor=center, rotate=-90}, draw=none, from=0, to=1]
	\arrow["\dashv"{anchor=center, rotate=-90}, draw=none, from=2, to=3]
\end{tikzcd}
\end{center}
\caption{The resolutions of functorial, indexed and Eilenberg-Moore algebras}\label{fig:cmpFunctors}
\end{figure}

\begin{figure}[h]
\begin{codepage}
\begin{nscenter}\begin{hscode}\SaveRestoreHook
\column{B}{@{}>{\hspre}l<{\hspost}@{}}%
\column{28}{@{}>{\hspre}l<{\hspost}@{}}%
\column{34}{@{}>{\hspre}l<{\hspost}@{}}%
\column{45}{@{}>{\hspre}l<{\hspost}@{}}%
\column{73}{@{}>{\hspre}l<{\hspost}@{}}%
\column{E}{@{}>{\hspre}l<{\hspost}@{}}%
\>[B]{}\mathbf{data}\;\Conid{EMAlg}\;\Sigma\;\Gamma\;\Varid{x}\mathrel{=}\Conid{EM}\;{}\<[34]%
\>[34]{}\{\mskip1.5mu \Varid{callEM}{}\<[45]%
\>[45]{}\mathbin{::}\Sigma\;\Varid{x}\to \Varid{x},\Varid{enterEM}{}\<[73]%
\>[73]{}\mathbin{::}\Gamma\;(\Conid{Prog}\;\Sigma\;\Gamma\;\Varid{x})\to \Varid{x}\mskip1.5mu\}{}\<[E]%
\\[\blanklineskip]%
\>[B]{}\Varid{handle}_{\Varid{EM}}\mathbin{::}(\Conid{Functor}\;\Sigma,\Conid{Functor}\;\Gamma)\Rightarrow (\Conid{EMAlg}\;\Sigma\;\Gamma\;\Varid{x})\to (\Varid{a}\to \Varid{x})\to \Conid{Prog}\;\Sigma\;\Gamma\;\Varid{a}\to \Varid{x}{}\<[E]%
\\
\>[B]{}\Varid{handle}_{\Varid{EM}}\;\Varid{alg}\;\Varid{gen}\;(\Conid{Return}\;\Varid{x}){}\<[28]%
\>[28]{}\mathrel{=}\Varid{gen}\;\Varid{x}{}\<[E]%
\\
\>[B]{}\Varid{handle}_{\Varid{EM}}\;\Varid{alg}\;\Varid{gen}\;(\Conid{Call}\;\Varid{op}){}\<[28]%
\>[28]{}\mathrel{=}(\Varid{callEM}\;\Varid{alg}\hsdot{\cdot}{.\ }\Varid{fmap}\;(\Varid{handle}_{\Varid{EM}}\;\Varid{alg}\;\Varid{gen}))\;\Varid{op}{}\<[E]%
\\
\>[B]{}\Varid{handle}_{\Varid{EM}}\;\Varid{alg}\;\Varid{gen}\;(\Conid{Enter}\;\Varid{op}){}\<[28]%
\>[28]{}\mathrel{=}(\Varid{enterEM}\;\Varid{alg}\hsdot{\cdot}{.\ }\Varid{fmap}\;(\Varid{fmap}\;(\Varid{handle}_{\Varid{EM}}\;\Varid{alg}\;\Varid{gen})))\;\Varid{op}{}\<[E]%
\ColumnHook
 \end{hscode}\resethooks
\end{nscenter}
\end{codepage}
\caption{Haskell implementation of EM algebras of $P$ based on \autoref{thm:em:rec}}\label{fig:em:impl}
\end{figure}

\begin{figure}[H]
\begin{codepage}
\begin{nscenter}\begin{hscode}\SaveRestoreHook
\column{B}{@{}>{\hspre}l<{\hspost}@{}}%
\column{3}{@{}>{\hspre}l<{\hspost}@{}}%
\column{10}{@{}>{\hspre}c<{\hspost}@{}}%
\column{10E}{@{}l@{}}%
\column{13}{@{}>{\hspre}l<{\hspost}@{}}%
\column{22}{@{}>{\hspre}l<{\hspost}@{}}%
\column{28}{@{}>{\hspre}l<{\hspost}@{}}%
\column{E}{@{}>{\hspre}l<{\hspost}@{}}%
\>[B]{}\mathbf{data}\;\Conid{Nat}\mathrel{=}\Conid{Zero}\mid \Conid{Nat}\mathbin{+}\mathrm{1}{}\<[E]%
\\[\blanklineskip]%
\>[B]{}\mathbf{data}\;\Conid{IxAlg}\;\Sigma\;\Gamma\;\Varid{a}\mathrel{=}{}\<[E]%
\\
\>[B]{}\hsindent{3}{}\<[3]%
\>[3]{}\Conid{IxAlg}\;{}\<[10]%
\>[10]{}\{\mskip1.5mu {}\<[10E]%
\>[13]{}\Varid{action}{}\<[22]%
\>[22]{}\mathbin{::}\forall \Varid{n}\hsforall \hsdot{\cdot}{.\ }\Sigma\;(\Varid{a}\;\Varid{n})\to \Varid{a}\;\Varid{n}{}\<[E]%
\\
\>[10]{},{}\<[10E]%
\>[13]{}\Varid{demote}{}\<[22]%
\>[22]{}\mathbin{::}\forall \Varid{n}\hsforall \hsdot{\cdot}{.\ }\Gamma\;(\Varid{a}\;(\Varid{n}\mathbin{+}\mathrm{1}))\to \Varid{a}\;\Varid{n}{}\<[E]%
\\
\>[10]{},{}\<[10E]%
\>[13]{}\Varid{promote}{}\<[22]%
\>[22]{}\mathbin{::}\forall \Varid{n}\hsforall \hsdot{\cdot}{.\ }\Varid{a}\;\Varid{n}\to \Varid{a}\;(\Varid{n}\mathbin{+}\mathrm{1})\mskip1.5mu\}{}\<[E]%
\\[\blanklineskip]%
\>[B]{}\Varid{hfold}\mathbin{::}(\Conid{Functor}\;\Varid{f},\Conid{Functor}\;\Varid{g})\Rightarrow \Conid{IxAlg}\;\Varid{f}\;\Varid{g}\;\Varid{a}\to \forall \Varid{n}\mathbin{::}\Conid{Nat}\hsforall \hsdot{\cdot}{.\ }\Conid{Prog}\;\Varid{f}\;\Varid{g}\;(\Varid{a}\;\Varid{n})\to \Varid{a}\;\Varid{n}{}\<[E]%
\\
\>[B]{}\Varid{hfold}\;\Varid{ixAlg}\;(\Conid{Return}\;\Varid{x}){}\<[28]%
\>[28]{}\mathrel{=}\Varid{x}{}\<[E]%
\\
\>[B]{}\Varid{hfold}\;\Varid{ixAlg}\;(\Conid{Call}\;\Varid{op}){}\<[28]%
\>[28]{}\mathrel{=}\Varid{action}\;\Varid{ixAlg}\;(\Varid{fmap}\;(\Varid{hfold}\;\Varid{ixAlg})\;\Varid{op}){}\<[E]%
\\
\>[B]{}\Varid{hfold}\;\Varid{ixAlg}\;(\Conid{Enter}\;\Varid{scope}){}\<[E]%
\\
\>[B]{}\hsindent{3}{}\<[3]%
\>[3]{}\mathrel{=}\Varid{demote}\;\Varid{ixAlg}\;(\Varid{fmap}\;(\Varid{hfold}\;\Varid{ixAlg}\hsdot{\cdot}{.\ }\Varid{fmap}\;(\Varid{promote}\;\Varid{ixAlg}\hsdot{\cdot}{.\ }\Varid{hfold}\;\Varid{ixAlg}))\;\Varid{scope}){}\<[E]%
\ColumnHook
 \end{hscode}\resethooks
\end{nscenter}
\end{codepage}
\caption{The hybrid fold for interpreting monad $P$ with indexed algebras \cite{PirogSWJ18}} \label{fig:hfold}
\end{figure}

\section{Notation}\label{app:notation}
This section summarises the mathematical notation used in this paper.

\paragraph{Typeface Conventions}
We use boldface variable such as $\CatC$ and $\CatD$ for abstract categories and
typewriter font for specific categories such as $\Set$.
Functors and objects are denoted by capitalised letters such as $F, K, G, X, A$
in general, but functors representing signatures of operations are always
denoted by Greek letters $\Gamma$ and $\Sigma$, and some concrete functors are
denoted by typewriter font such as $\Free$ and $\Fn$.
Morphisms and natural transformations are denoted by uncapitalised letters such
as $f$, $g$ and $h$ and Greek letters such as $\alpha$, $\beta$ and $\tau$.

\paragraph{Categories}
For any two categories $\CatC$ and $\CatD$, we write $\CatC \times \CatD$
for their \emph{product category}, whose objects are denoted by $\abracket{X, Y}$
where $X : \CatC$ and $Y : \CatD$ and morphisms are also denoted by $\abracket{f, g}$
where $f : X \to X'$ and $g : Y \to Y'$.
The \emph{functor category} from $\CatC$ to $\CatD$ is denoted by
$\CatD^\CatC$.  Specifically, the category of \emph{finitary endofunctors} on
$\CatC$ is denoted by $\CatEndof$.
Given a monad $P : \CatC \to \CatC$, the category of Eilenberg-Moore algebras is
denoted by $\CatC^P$.
Given an endofunctor $F$ on $\CatC$, the category of $F$-algebras is $F\Alg$.

\paragraph{Functors}

An adjoint pair of functors $L : \CatC \to \CatD$ and $R : \CatD \to \CatC$ is 
denoted by $L \dashv R$ where $L$ is the left adjoint.
The associated unit of the adjunction is $\eta : \Id \to R L$ and the counit is
$\epsilon : L R \to \Id$, where $\Id$ is the identity functor.
Sometimes we write $\eta^{R L}$ or $\epsilon^{R L}$ to make clear the monad for 
(co)units.

For monads, letters $P$, $T$, $M$ are used, and $\mu$ is used for their multiplication.
In particular, $P$ is used for the monad of programs with algebraic and scoped operations,
while $T$ is used for the monad from the adjunction $\Free_{\Fn}
\uparrow \dashv \downarrow \Ul_{\Fn}$.

Functor $G$ denotes the functor encoding the grammar (\ref{eq:g:equation}), and
functor $I$ (\ref{eq:I:def}) and $\Fn$ (\ref{eq:fn:equation}) are used for
defining functorial algebras.
Similarly, functors $\Ix$, $\previous$ and $\later$ are used for defining
indexed algebras (\autoref{sec:indexedAlgOfE}).

In some of the examples in the paper, we make use of the functor $\ensuremath{\Conid{List}} : \Set
\to \Set$ on $\Set$ that sends every set $X$ to the set of finite lists with elements
from $X$. 
We use \ensuremath{\Varid{xs}+\!\!\!+\Varid{ys}} for list concatenation and \ensuremath{\Varid{cons}\;\Varid{x}\;\Varid{xs}} for a list with \ensuremath{\Varid{x}} in
front of \ensuremath{\Varid{xs}}.

Comparison functors (\autoref{sec:comparing}) are denoted by $K$, and in particular
$\KL$ and $\KK$.

\paragraph{Objects and Morphisms}
We write $f \vcomp g$ for vertical composition of morphisms, and
$F \hcomp G$ for horizontal composition of natural transformations and composition
of functors.

Products in a category are denoted by $X_1 \times \dots \times X_n$ for finite products
or $\prod_{i \in I} X_i$ for some set $I$, and coproducts are denoted by $X_1 +
\dots + X_n$ or $\coprod_{i \in I} X_i$ for some set $I$.
The injection morphisms into coproducts are denoted by $\iota_i : X_i \to X_1 + \dots + X_n$
for each $i$, and the morphisms out of coproducts are $[f_1, \dots, f_n] : X_1 + \dots + X_n  \to Z$
for objects $Z$ with morphisms $f_i : X_i \to Z$ for all $i$.

Initial objects are generally written as $0$, and the unique morphism from the initial
object to an object $X$ is denoted by $! : 0 \to X$.
In particular, the initial object $0$ in a functor category is the constant functor
mapping to the initial object in the codomain.

The (carrier of the) \emph{initial algebra} of an endofunctor $G : \CatC
\to \CatC$ is denoted by $\mu G$ or $\mu Y.\ G Y$.
The structure map of the initial algebra is denoted by $\iniso : G (\mu G) \to G$
and its inverse is $\inOp$.
The unique morphism from the initial $G$-algebra to a $G$-algebra $\alpha : G X
\to X$, i.e.\ the catamorphism for $\alpha$, is denoted by $\cata{\alpha} : \mu
G \to X$.

Given an endofunctor $\Sigma : \CatC \to \CatC$, the free monad over it is
denoted by $\Sigma^*$.
Given a $\Sigma$-algebra $\abracket{X : \CatC, \alpha : \Sigma X \to X}$ and
a morphism $g : A \to X$, the interpretation morphism is denoted by
$\interp{\abracket{X, \alpha}}{g} : \Sigma^* A \to X$.
The notation $\interpName$ is also used for functorial and indexed algebras.
The structure map of the free algebra $\Sigma^* A$ is denoted by $\op : \Sigma
\Sigma^* A \to \Sigma^* A$.

Functorial algebras are usually denoted by $\abracket{H, X, \alpha^G,
\alpha^I}$ (\autoref{def:fctAlgs}), where $H$ is an endofunctor on the base
category $\CatC$ and $X$ is an object in $\CatC$.

\section{Proofs Omitted in the Main Paper}\label{app:proofs}
This section contains the proofs omitted in the paper.
The following two lemmas are a more detailed version of the standard results
\autoref{lem:free:alg:simple}.
We refer the reader to \cite{Adamek1974,Barr70Co} for their proofs.
\begin{lemma}\label{lem:initial:alg}
If category $\CatC$ has an initial object and colimits of all $\omega$-chains
and endofunctor $\Sigma : \CatC \to \CatC$ preserves them, then there exists an
initial $\Sigma$-algebra $\tuple{\mu \Sigma : \CatC,\ \iniso : \Sigma(\mu
\Sigma) \to \mu \Sigma}$. 
Moreover, the structure map $\iniso$ is an isomorphism.
\end{lemma}
\begin{lemma}\label{lem:free:alg}
If category $\CatC$ has finite coproducts and colimits of all $\omega$-chains and
endofunctor $\Sigma : \CatC \to \CatC$ preserves them, then the forgetful
functor $\Sigma\Alg \to \CatC$ has a left adjoint $\Free_{\Sigma} : \CatC \to
\Sigma\Alg$ mapping every object $X$ to $\lfix{Y}{X + \Sigma Y}$
with structure map 
\[
\op_X = \big(\Sigma (\lfix{Y}{X + \Sigma Y}) \xrightarrow{\iota_2} X +  \Sigma (\lfix{Y}{X + \Sigma Y}) 
\xrightarrow{\iniso} \lfix{Y}{X + \Sigma Y}\big)
\]
and the unit $\eta$ and counit $\epsilon$ of the adjunction $\Free_\Sigma \dashv \Ul_\Sigma$ are
\begin{equation}\label{eq:unitFromInitial}
  \eta_X = \iniso \vcomp \iota_1 : X \to U_\Sigma (\Free_\Sigma X)
  \quad\quad
  \epsilon_{\abracket{ X, \alpha}} = \cata{ [\identity, \alpha] }
     : \Free_\Sigma X  \to \abracket{ X, \alpha}
\end{equation}
where $\cata{ [\identity, \alpha] }$ denotes the catamorphism from the initial
algebra $\lfix{Y}{X + \Sigma Y}$ to $X$.
Moreover, this adjunction is strictly monadic.
\end{lemma}

\begin{lemma*}[\ref{lem:fn:free:adj}]
The endofunctor $\Fn$ (\ref{eq:fn:equation}) on $\CatEC$
has free algebras, i.e.\ there is a functor $\Free_\Fn : \CatEC \to \Fn\Alg$
left adjoint to the forgetful functor $\Ul_\Fn : \Fn\Alg \to \CatEC$.
\end{lemma*}

\begin{proof}
Since $\CatC$ is lfp, it is cocomplete, and it follows that $\CatEC$ is cocomplete
because colimits in functor categories and product categories can be computed
pointwise \cite[Thm V.3.2]{MacLane1978}.
It is also easy verification that $\Fn$ preserves all colimits of $\omega$-chains
following from the fact that $G$, $\Sigma$ and $\Gamma$, and all functors in
$\CatEndof$ preserve colimits of $\omega$-chains in their domains respectively.
Hence by \autoref{lem:free:alg:simple}, there is a functor $\Free_\Fn$ left adjoint to
$\Ul_\Fn$.
\end{proof}

\begin{lemma*}[\ref{lem:freeJ}]
There is a natural isomorphism between $\Free_{\Fn}$ and the following functor
\[
\FFn \abracket{H,X} = \left\langle
G^{*} H : \CatEndof,\ \ \ (I_{G^{*}H})^{*}X : \CatC,\ \ \ \op^{G^*}_H,\ \ \ \op^{(I_{G^*H})^*}_{X} 
\right\rangle
\]
where $\op^{G^*}_H : G (G^{*} H) \to G^{*} H$ and $ \op^{(I_{G^*H})^*}_{X} :
I_{G^*H} ((I_{G^{*}H})^{*}X) \to (I_{G^{*}H})^{*}X$ are the structure maps of
the free $G$-algebra and $I_{G^{*}H}$-algebra respectively.
\end{lemma*}

\begin{proof}
By \autoref{lem:free:alg}, the free $\Fn$-algebra generated by $\abracket{H,X}$
can be constructed from the initial algebra of a functor $\Fn_{\abracket{H,X}} :
\CatEC \to \CatEC$ such that
\[\Fn_{\abracket{H,X}}Y = \abracket{H,X} + \Fn Y\]
Then we show that an initial $\Fn_{\abracket{H,X}}$-algebra carried by $\abracket{
G^{*}H, (I_{G^{*}H})^* X }$ with structure map 
\begin{align*}
\abracket{ i_1, i_2 } :\ & \Fn_{\abracket{H,X}}\abracket{ G^{*}H, (I_{G^{*}H})^* X } =
\abracket{ H + GH, X + I_{G^{*}H}((I_{G^{*}H})^*X) } \\
  & \to \abracket{ G^{*}H, (I_{G^{*}H})^* X }
\end{align*}
where $i_1 = [\eta^{G^{*}}_H, \op^{G^{*}}_{H}] : H + G(G^{*}H) \to G^{*} H$
and 
\[
i_2 = [\eta^{(I_{G^{*}H})^*}_{X}, \op^{(I_{G^{*}H})^*}_{X}] : X +
I_{G^{*}H}((I_{G^{*}H})^{*}X) \to (I_{G^{*}H})^{*} X
.\]
To see that this $\Fn_{\abracket{H,X}}$-algebra is initial, consider any
$\abracket{C, D}$ in $\CatEndo \times \CatC$ with structure map $\abracket{
  j_1, j_2 } : \Fn_{\abracket{H,X}}\abracket{ C,D } \to \abracket{ C,D }$.
We have
\begin{align*}
&\abracket{ k_1, k_2 } : \Fn_{\abracket{H,X}}\Alg\left(
\big\langle \abracket{ G^{*}H, (I_{G^{*}H})^* X },\ \abracket{ i_1, i_2 }\big\rangle 
,\ 
\big\langle\abracket{ C, D }, \abracket{ j_1, j_2 }\big\rangle\right) \\
\Leftrightarrow\ & \Big(k_1 \in (H + G-)\Alg\big( \abracket{ G^{*}H, i_1 },\
\abracket{ C, j_1 } \big)\Big) \\
&\quad\quad \wedge\ \Big(k_2 \in (X + I_{G^*H}-)\Alg\big(\abracket{ I_{G^*H}X, i_2 },\
\abracket{ D, j_2 \vcomp (X + I_{k_1}\identity) }\big)\Big) \\
\Leftrightarrow\ & k_1 = \ladjunct{j_1 \vcomp \iota_1}_{\abracket{ C, j_1 \vcomp \iota_2 }}
\ \wedge\ k_2 = \ladjunct{ j_2 \vcomp (X + I_{k_1}) \vcomp
\iota_1}_{\abracket{ D,  j_2 \vcomp (X + I_{k_1}) \vcomp \iota_2 }}
\end{align*}
where we use subscripts of $\ladjunct{\cdot}$ to indicate the $B$ for some
$\ladjunct{f} : L A \to B$.
The calculation shows that the $\Fn_{\abracket{H,X}}$-algebra homomorphism $(k_1, k_2)$
uniquely exists, and thus $\abracket{ G^{*}H, (I_{G^{*}H})^* X }$ with
structure map $\abracket{ i_1, i_2 }$ is initial.
Then by \autoref{lem:free:alg}, this initial algebra gives the
free $\Fn$-algebra generated by $\abracket{H,X}$, and thus we have the isomorphism
between $\Free_\Fn$ and $\FFn$ in the lemma.
\end{proof}
This characterisation of free $\Fn$-algebras also allows us to express the
unit and counit of the adjunction $\Free_\Fn \dashv \Ul_\Fn$ in terms of those
of some simpler adjunctions.

\begin{lemma}\label{lem:unitFreeBarG}
Letting the $\phi$ be the isomorphism in \autoref{lem:freeJ},
the unit of adjunction $\Free_{\Fn} \dashv \Ul_{\Fn}$ is
\[
\eta_{\abracket{H,X}} = 
\begin{tikzcd}[column sep=large]
  \abracket{H,X} \arrow{r}{\abracket{ \eta^{G^{*}}_{H}, \eta^{(I_{G^{*}H})^*}_X
  }} & \abracket{ G^{*}H, (I_{G^{*}H})^* X } \arrow{r}{\phi^{-1}} &
  \Fn^*\abracket{H,X}
\end{tikzcd}
\]
and its counit $\epsilon$ at some $\Fn$-algebra $\abracket{ \abracket{H,X}, \abracket{
\beta_1, \beta_2 } }$ is
\[
\Free_{\Fn}\abracket{H,X} \xrightarrow{\phi}  \big\langle \abracket{ G^{*}H, (I_{G^{*}H})^{*}X }, 
\abracket{ \op^{G^*}_H, \op^{(I_{G^*H})^*}_{X} } \big\rangle \xrightarrow{\abracket{ e_1, e_2 }} 
\abracket{ \abracket{H,X}, \abracket{ \beta_1, \beta_2 } }
\]
where $e_1 = \Ul_{G}(\epsilon^{G^*}_{\abracket{ H, \beta_1 }})$ and $e_2 = \Ul_{I_{G^{*}
H}}(\epsilon^{(I_{G^{*} H})^*}_{\abracket{ X, \beta_2 \vcomp I_{e_1} }})$.
\end{lemma}
\begin{proof}[Proof sketch]
It can be calculated from (\ref{eq:unitFromInitial}) and \autoref{lem:freeJ}.
\end{proof}

\begin{theorem*}[\ref{thm:EisoT}]
The monad $P$ is isomorphic to $T$ in the category of monads.
\end{theorem*}

\begin{proof}
By (\ref{eq:e:from:free}), $P \cong (\Sigma + \Gamma P)^* = (I_{P})^*$ as monads,
it is sufficient to show that $T$ is isomorphic to $(I_{P})^*$ as monads.
Recall that $P = \mu G \cong G^* {0}$ as endofunctors.
Let $\psi : G^*{0} \to P$ be the isomorphism, and let $\phi$ be the isomorphism
between $\Fn^*$ and $\Ul_\Fn \FFn$ by \autoref{lem:freeJ}, then for all $X : \CatC$,
\begin{align}\label{eq:isoTandE}
T X &= \downC (\Ul_{\Fn} (\Free_{\Fn} (\upC X))) =
      \downC (\Fn^* \abracket{ 0, X }) \overset{\downC\phi}{\cong}
      \downC \tuple{ G^{*}0, \abracket{ I_{G^{*}0} }^*X } \\
    &= (I_{G^{*}0})^{*}X \overset{(I_{\psi})^*}{\cong}
      (I_{P})^*X  \nonumber
\end{align}
Thus $T$ is isomorphic to $(I_P)^*$ as endofunctors. 

What remains is to show that the isomorphism (\ref{eq:isoTandE}) preserves
their units and multiplications.
The unit of $T$ is precisely the unit of the adjunction
$(\Free_{\Fn}\upC) \dashv (\downC \Ul_{\Fn})$ composed from the
adjunctions $\upC \dashv \downC$ and $\Free_{\Fn} \dashv \Ul_{\Fn}$.
Therefore the unit of $T$ is
$\eta^{T}_X = \downC (\eta^{\Fn^*}_{\upC X}) \vcomp \eta^{\downC\upC}_X$
where $\eta^{\downC\upC}$ and $\eta^{\Fn^*}$ are the units of the adjunctions
$\upC \dashv \downC$ and $\Free_{\Fn} \dashv \Ul_{\Fn}$ respectively.
Hence by \autoref{lem:unitFreeBarG}, we have
\begin{align*}
\eta^{T}_X 
&= \downC (\eta^{\Fn^*}_{\upC X}) \vcomp \eta^{\downC\upC}_X
= \downC (\eta^{\Fn^*}_{\upC X}) \vcomp \identity
= \downC (\phi^{-1} \vcomp \abracket{ \eta^{G^{*}}_{0}, \eta^{(I_{G^{*}0})^*}_X }) \\
& = \downC (\phi^{-1}) \vcomp \eta^{(I_{G^{*}0})^*}_X
= \downC \phi \vcomp ({I_{\psi}}^*)^{-1} \vcomp \eta^{(I_{P})^*}_X
\end{align*}
which shows that the isomorphism~(\ref{eq:isoTandE}) preserves the units of
$T$ and $(I_{P})^*$.

Proving the preservation of the multiplications of the two monads is
also direct verification but slightly more involved.
By definition, $\mu^{T}_X = \downC \Ul_{\Fn}
\epsilon^T_{\Free_{\Fn} \upC X}$
where $\epsilon^T$ is the counit of the adjunction $(\Free_{\Fn}\upC)
\dashv (\downC \Ul_{\Fn})$ satisfying
\begin{align*}
\epsilon^T_{\Free_{\Fn} \upC X} 
=  \epsilon^{\Fn^*}_{\Free_{\Fn} \upC X} \vcomp
\Free_{\Fn} (\epsilon^{\downC\upC}_{\Fn^{*}\upC X})
=  \epsilon^{\Fn^*}_{\Free_{\Fn} \upC X} \vcomp
\Free_{\Fn} \abracket{ \mathop{!}, \identity }
\end{align*}
where $\mathop{!}$ is the unique $G$-algebra homomorphism from $G^*0$ to
$G^*G^*0$.
Then by \autoref{lem:unitFreeBarG}, we have
$\epsilon^{\Fn^*}_{\Free_{\Fn} \upC X} = \abracket{ e_1, e_2 } \vcomp
\phi$ where $e_1 = \Ul_{G}(\epsilon^{G^*}_{(G^*0,\,
 \op^{G^*}_{0})}) : {G}^* {G}^* {0} \to {G}^*
{0}$ and 
$e_2 = \Ul_{I_{G^{*} G^*0}}(\epsilon^{(I_{G^{*} G^*0})^*}_{b})$
where $b$ is the $I_{G^{*} G^*0}$-algebra
$\abracket{ (I_{G^*0})^*X, \op^{(I_{G^*0})^*}_{X}
\vcomp I_{e_1} }$.
Hence we have
\begin{align*}
\mu^{T}_X 
&= \downC \Ul_{\Fn} (\epsilon^T_{\Free_{\Fn} \upC X}) \\
&= \Ul_{I_{G^{*} G^*0}}(\epsilon^{(I_{G^{*}
G^*0})^*}_{b}) \vcomp \downC
\phi \vcomp \downC \Ul_{\Fn} \Free_{\Fn}(\mathop{!}, \identity) \\
&= \Ul_{I_{G^{*} G^*0}}(\epsilon^{(I_{G^{*}
G^*0})^*}_{b}) \vcomp \downC
\phi \vcomp \downC (\phi^{-1} \vcomp \abracket{ G^*{\mathop{!}},
(I_{G^*{\mathop{!}}})^* } \vcomp \phi)\\
&= \Ul_{I_{G^{*} G^*0}}(\epsilon^{(I_{G^{*}
G^*0})^*}_{b}) \vcomp (I_{G^*{\mathop{!}}})^* \vcomp \downC \phi 
\end{align*}
where $(I_{G^*{\mathop{!}}})^*$ is a natural transformation from
$(I_{G^*0})^*$ to $(I_{G^*G^*0})^*$.
Then by \emph{base functor fusion} \cite{Hinze13Adj} i.e.\ the naturality of the
free-forgetful adjunction in the base functor, we have
\begin{align*}
\mu^{T}_X 
= \Ul_{I_{G^{*} G^*0}}(\epsilon^{(I_{G^{*}
G^*0})^*}_{b}) \vcomp (I_{G^*{\mathop{!}}})^* \vcomp \downC \phi
= \Ul_{I_{G^*0}}(\epsilon^{(I_{G^*0})^*}_{b'}) \vcomp \downC \phi 
\end{align*}
where $b'$ is the $I_{G^*0}$-algebra with the same carrier 
as $b$ and with $b$'s structure map precomposed with
$I_{G^*{\mathop{!}}}$:
\[
\op^{(I_{G^*0})^*}_{X} \vcomp I_{e_1} \vcomp I_{G^*{\mathop{!}}} 
= \op^{(I_{G^*0})^*}_{X} \vcomp I_{e_1 \vcomp \mathop{!}}
= \op^{(I_{G^*0})^*}_{X} \vcomp I_{\identity}
= \op^{(I_{G^*0})^*}_{X} 
\]
where the second equality follows from that $e_1 \vcomp \mathop{!} :
G^*{0} \to G^*{0}$ is a $G$-algebra homomorphism, and
thus $e_1 \vcomp \mathop{!} = \identity$ since $G^*{0}$ is an initial
$G$-algebra.
Finally we have
\begin{align*}
\mu^{T}_X 
&= \Ul_{I_{G^*0}}(\epsilon^{(I_{G^*0})^*}_{\abracket{
  (I_{G^*0})^*X, \op^{(I_{G^*0})^*}_{X} }}) \vcomp
  \downC \phi \\
& = \Ul_{I_P}(\epsilon^{I^*_P}_{\abracket{ I^*_{P} X, \op^{(I_{P})^*}_{X} }})
\vcomp ((I_{\psi})^* \vcomp \downC \phi) \vcomp T((I_{\psi})^* \vcomp \downC
\phi)
\end{align*}
which means exactly that the isomorphism (\ref{eq:isoTandE}) preserves the
multiplications of $T$ and $(I_P)^*$.
\end{proof}

\begin{lemma*}[\ref{lem:KLisCMP}]
Functor $\KL$ is a comparison functor from the Eilenberg-Moore resolution of 
$P$ to the resolution $\Free_\Fn \upC \dashv \downC \Ul_\Fn$ of functorial
algebras.
\end{lemma*}

\begin{proof}
By \autoref{def:resCmp}, we need to show that $\KL$ commutes with left and right adjoints
of both resolutions:
for right adjoints, we have
$
\Ul_{\Sigma + \Gamma \hcomp P} \abracket{X, \alpha} = X = \downC \Ul_{\Fn} \KL \abracket{X, \alpha}
$; and for left adjoints, %
\begin{align*}
\KL(\Free_{\Sigma + \Gamma \hcomp P} X) =\,&
\KL \abracket{ (\Sigma + \Gamma \hcomp P)^* X,\,
  \op^{(\Sigma + \Gamma \hcomp P)^*}_X
  } \\
  =\,& \abracket{ P, (\Sigma + \Gamma \hcomp P)^* X, \iniso, \op^{\cdots}_X} 
    \tag*{\ensuremath{\{ P \iso G^*0
    \text{ and } (\Sigma + \Gamma \hcomp P) = I_P \text{ (\autoref{sec:fctAlgDef})}\}}} \\
 \iso\,&  \abracket{G^* 0,  (I_{G^* 0})^* X,
 \op^{G^*}_{0}, \op^{(I_{G^*{0}})^*}_{X}}  
 \tag*{\{By \autoref{lem:freeJ}\}}\\
 \iso\,& \Free_\Fn{(\upC X)}
\end{align*}
and similarly for the actions on morphisms.
Here we only have $\KL\Free_{\Sigma + \Gamma \hcomp P}$ being isomorphic to $\Free_\Fn \upC$
instead of a strict equality, since these two resolutions induce the monad
$P$ only up to isomorphism.
To remedy this, one can slightly generalise the definition of comparison functors to
take an isomorphism into account, but we leave it out here.
\end{proof}

\begin{lemma*}[\ref{lem:KK:is:comp}]
$\KK$ is a comparison functor from the resolution 
$\Free_\Fn \upC \dashv \downC \Ul_\Fn$ of functorial algebras to the resolution
$\Free_\Ix \mathop{\upharpoonright} \dashv \mathop{\downharpoonleft} \Ul_\Ix$
of indexed algebras.
\end{lemma*}
\begin{proof}
We need to show the required commutativities for $\KK$ to be a comparison
functor:
\[
\downC \Ul_\Fn \iso \mathop{\downharpoonleft} \Ul_\Ix \KK
\qquad\text{and}\qquad
\KK \Free_\Fn \upC \iso \Free_\Ix \mathop{\upharpoonright} 
\]
First it is easy to see that it commutes with the right adjoints:
\[
  \mathop{\downharpoonleft} \Ul_\Ix (\KK \abracket{H, X, \alpha^G, \alpha^I})
= \mathop{\downharpoonleft} \Ul_\Ix \abracket{ A, a, d, p }
= A_0
= X
= \mathop{\downharpoonleft} \Ul_\Ix (\KK \abracket{H, X, \alpha^G, \alpha^I})
\]
Its commutativity with the left adjoints is slightly more involved, and we show a 
sketch here.
Pir\'{o}g et al.\ \cite{PirogSWJ18} show that $\Free_\Ix \mathop{\upharpoonright} X$ is isomorphic
to the indexed algebras carried by $P^+_X : \CatCN$ such that $(P^+_X)_i = P^{i+1} X$
with structure maps
\begin{gather*}
k^{\Sigma}_n = \big(\Sigma (P^+_X)_n 
 = \Sigma P P^n X \xrightarrow{\iniso \vcomp \iota_2} PP^nX  \big)\\
k^{\Gamma \previous}_n = \big((\Gamma \previous P^+_X)_n = \Gamma P P P^n X  
\xrightarrow{\iniso \vcomp \iota_3} PP^n X \big) \\
k^{\later}_n = \big((P^+_X)_n = PP^n X \xrightarrow{\iniso \vcomp \iota_1} PPP^n X 
\big)
\end{gather*}
where $\iniso : \Id + \Sigma \hcomp P + \Gamma \hcomp P \hcomp P \to P$ is the
isomorphism between $P$ and $G P$.
Also by \autoref{lem:freeJ}, we know that $\Free_\Fn \upC X$ is isomorphic to 
the functorial algebra $\abracket{P, PX, \iniso, [\iniso \vcomp \iota_2,
\iniso \vcomp \iota_3]}$.
Clearly $\KK\Free_\Fn \upC$ and $\Free\Ix \mathop{\upharpoonright}$ agree on the
carrier $P^+_X$.
It can be checked that they agree on the structure maps and the action on morphisms
too.
\end{proof}

\section{Handling with EM Algebras in Haskell} \label{app:em}

This section provides some details about how Eilenberg-Moore algebras of the monad
$P$ (\ref{eq:e:recursive}) are used for interpret scoped operations.
First we have the follow theorem characterising $\interpName$ in a way that suits
recursive implementation.
A Haskell implementation is in (\autoref{fig:em:impl}).

\begin{theorem}[Interpreting with EM Algebras]\label{thm:em:rec}
Given an Eilenberg-Moore algebra as in (\ref{eq:e:em}), for any morphism $g : A
\to X$ in $\CatC$ for some $A$, the interpretation of $P A$ with this algebra
and $g$ satisfies
\begin{align}
 \interp{\abracket{X, \alpha_\Sigma, \alpha_\Gamma }}{g}  
= [\  & g,\ \ 
  \alpha_\Sigma \vcomp \Sigma (\interp{\abracket{X, \alpha_\Sigma, \alpha_\Gamma }}{g}),\ \ \label{eq:evalEMAlg} \\
   & \alpha_\Gamma \vcomp \Gamma P (\interp{\abracket{X, \alpha_\Sigma,
  \alpha_\Gamma }}{g})] \vcomp \inOp_A \nonumber
\end{align}
where $\inOp : P \to \Id + \Sigma \hcomp P + \Gamma \hcomp P \hcomp P$ is
the isomorphism between $P$ and $G P$.
\end{theorem}
\begin{proof}
It can be calculated by plugging the counit $\epsilon$ for the monad
$\Free_{\Sigma + \Gamma \hcomp P} \dashv \Ul_{\Sigma + \Gamma \hcomp P}$
(\autoref{lem:free:alg}) into (\ref{eq:evalEMAlgDef}).
\end{proof}

\begin{example}\label{ex:em:once}
With the implementation of EM algebras and their interpretation
(\autoref{fig:em:impl}), the scoped operation \ensuremath{\Varid{once}} in \autoref{ex:ndet:hdl}
can be interpreted by the following EM algebra
\begin{nscenter}\begin{hscode}\SaveRestoreHook
\column{B}{@{}>{\hspre}l<{\hspost}@{}}%
\column{3}{@{}>{\hspre}l<{\hspost}@{}}%
\column{5}{@{}>{\hspre}l<{\hspost}@{}}%
\column{7}{@{}>{\hspre}l<{\hspost}@{}}%
\column{14}{@{}>{\hspre}l<{\hspost}@{}}%
\column{20}{@{}>{\hspre}l<{\hspost}@{}}%
\column{E}{@{}>{\hspre}l<{\hspost}@{}}%
\>[B]{}\Varid{onceAlgEM}\mathbin{::}\Conid{EMAlg}\;\Conid{Choice}\;\Conid{Once}\;[\mskip1.5mu \Varid{a}\mskip1.5mu]{}\<[E]%
\\
\>[B]{}\Varid{onceAlgEM}\mathrel{=}\Conid{EM}\;\{\mskip1.5mu \mathinner{\ldotp\ldotp}\mskip1.5mu\}\;\mathbf{where}{}\<[E]%
\\
\>[B]{}\hsindent{3}{}\<[3]%
\>[3]{}\Varid{callEM}\mathbin{::}\Conid{Choice}\;[\mskip1.5mu \Varid{a}\mskip1.5mu]\to [\mskip1.5mu \Varid{a}\mskip1.5mu]{}\<[E]%
\\
\>[B]{}\hsindent{3}{}\<[3]%
\>[3]{}\Varid{callEM}\;\Conid{Fail}{}\<[20]%
\>[20]{}\mathrel{=}[\mskip1.5mu \mskip1.5mu]{}\<[E]%
\\
\>[B]{}\hsindent{3}{}\<[3]%
\>[3]{}\Varid{callEM}\;(\Conid{Or}\;\Varid{x}\;\Varid{y}){}\<[20]%
\>[20]{}\mathrel{=}\Varid{x}+\!\!\!+\Varid{y}{}\<[E]%
\\[\blanklineskip]%
\>[B]{}\hsindent{3}{}\<[3]%
\>[3]{}\Varid{enterEM}\mathbin{::}\Conid{Once}\;(\Conid{Prog}\;\Conid{Choice}\;\Conid{Once}\;[\mskip1.5mu \Varid{a}\mskip1.5mu])\to [\mskip1.5mu \Varid{a}\mskip1.5mu]{}\<[E]%
\\
\>[B]{}\hsindent{3}{}\<[3]%
\>[3]{}\Varid{enterEM}\;(\Conid{Once}\;\Varid{p})\mathrel{=}{}\<[E]%
\\
\>[3]{}\hsindent{2}{}\<[5]%
\>[5]{}\mathbf{case}\;\Varid{handle}_{\Varid{EM}}\;\Varid{onceAlgEM}\;(\lambda \Varid{x}\to [\mskip1.5mu \Varid{x}\mskip1.5mu])\;\Varid{p}\;\mathbf{of}{}\<[E]%
\\
\>[5]{}\hsindent{2}{}\<[7]%
\>[7]{}[\mskip1.5mu \mskip1.5mu]{}\<[14]%
\>[14]{}\to [\mskip1.5mu \mskip1.5mu]{}\<[E]%
\\
\>[5]{}\hsindent{2}{}\<[7]%
\>[7]{}(\Varid{x}\mathbin{:\char95 }){}\<[14]%
\>[14]{}\to \Varid{x}{}\<[E]%
\ColumnHook
 \end{hscode}\resethooks
\end{nscenter}
Note that this EM algebra is defined with recursion:
the part of syntax enclosed by the scope of \ensuremath{\Varid{once}} is interpreted by a
recursive call to \ensuremath{\Varid{handle}_{\Varid{EM}}} with the algebra itself.
\end{example}

\section{Handling with Indexed Algebras by Hybrid Fold}\label{app:hfold}

This section elaborates \autoref{ex:hfold} about using fusion laws and functorial 
algebra to prove the correctness of a recursive scheme used by Pir\'{o}g et al.\ \cite{PirogSWJ18}
to interpret scoped operations with indexed algebras.
Although Pir\'{o}g et al.\ \cite{PirogSWJ18} propose the adjunction $(\Free_\Ix \upharpoonright) \dashv
(\downharpoonleft \Ul_\Ix)$ for interpreting scoped operations with indexed algebras,
they use the recursive function \ensuremath{\Varid{hfold}} (\autoref{fig:hfold}) in their
implementation to interpret $P$ (\ref{eq:e:recursive2}) with indexed algebras.
Compared to a faithful implementation of $\interpName$, their \ensuremath{\Varid{hfold}} is more
efficient since it skips transforming $P$ to the free indexed algebra
$(\downharpoonleft \Ul_\Ix \Free_\Ix \upharpoonright)$ but directly works on $P$.
Thus we call it a \emph{hybrid fold} since it works on a syntactic structure
$P$ that is not freely generated by its type of algebras giving semantics.

While the definition of \ensuremath{\Varid{hfold}} is computationally intuitive,
Pir\'{o}g et al.\ \cite{PirogSWJ18} did not provide formal justification for this recursive
definition.
In this subsection, we fill the gap by showing that \ensuremath{\Varid{hfold}} coincides with \ensuremath{\Varid{handle}}
with indexed algebras.
We divide the proof into three parts for clarity:
after making the problem precise (\autoref{sec:hfoldProblem}),
we first show that the \ensuremath{\Varid{hfold}} for an indexed algebra \ensuremath{\Conid{A}} is equivalently 
a catamorphism from $P$ in $\CatEndof$ (\autoref{sec:hfoldProof1}), which is then
a special case of interpreting with functorial algebras
(\autoref{sec:hfoldProof2}), and finally we translate this functorial algebra
into the category $\Ix\Alg$ of indexed algebras using $\KK$, and show that it
induces the same interpretation as the one from the indexed algebra $A$ that we
start with (\autoref{sec:hfoldProof3}).

\subsection{Semantic Problem of Hybrid Fold}
\label{sec:hfoldProblem}

Fix an indexed algebra carried by $A : \CatCN$ in this section:
\begin{equation}\label{eq:fixed:a}
\abracket{A : \CatCN,\ a : \Sigma \hcomp A \to A,\ 
d : \Gamma \hcomp (\previous A) \to A, p : A \to (\previous A)}
\end{equation}
For notational convenience, we define a functor $S : \CatN \to \CatN$ such that
$Sn = n+1$, then $\previous A = A \hcomp S$ since $(\previous A)_n = A_{n+1} = A(Sn)$.
With functor $S$, we can view $p$ and $d$ as $p : A \to A \hcomp S$ and $d :
\Gamma \hcomp A \hcomp S \to A$.
Then the recursive definition of \ensuremath{\Varid{hfold}} in \autoref{fig:hfold} can be
understood as 
a morphism $h : P \hcomp A \to A$ in $\CatCN$ satisfying the equation 
\begin{equation}\label{eq:hfoldEq}
h = [h_1, h_2, h_3] \vcomp (\inOp \hcomp A)
\end{equation}
where $\inOp : P \to \Id + \Sigma \hcomp P + \Gamma \hcomp P \hcomp P$ is
the isomorphism between $P$ and $G P$, and $h_1$, $h_2$ and $h_3$ correspond to
the three cases of \ensuremath{\Varid{hfold}} respectively:
\begin{gather*}
h_1 = \big( \Id \hcomp A \xrightarrow{\identity} A \big)
\hspace{2em}
h_2 = \big( \Sigma \hcomp P \hcomp A \xrightarrow{\Sigma \hcomp h} \Sigma \hcomp A \xrightarrow{a} A\big) \\
h_3 = \big( \Gamma \hcomp P \hcomp P \hcomp A \xrightarrow{\Gamma \hcomp P \hcomp h}
\Gamma \hcomp P \hcomp A \xrightarrow{\Gamma \hcomp P \hcomp p} 
\Gamma \hcomp P \hcomp A \hcomp S \xrightarrow{\Gamma \hcomp h \hcomp S} 
\Gamma \hcomp A \hcomp S \xrightarrow{d} 
A \big)
\end{gather*}
In general, an equation in the form of (\ref{eq:hfoldEq}) does \emph{not}
necessarily have a (unique) solution.
Thus the semantics of the function \ensuremath{\Varid{hfold}} is not clear.
We settle this problem with the following result.

\begin{theorem}[Hybrid Folds Coincide with Interpretation]\label{thm:hfoldOK}
There exists a unique solution to (\ref{eq:hfoldEq}) and it coincides with
the interpretation with indexed algebra $A$ at level $0$ (\ref{eq:evalIxDef}):
\[h_0 = \interp{\abracket{A,a,d,p}} \identity : P A_0 \to A_0\]
\end{theorem}

\noindent
We prove the theorem in the rest of this section with the tools that we have
developed.

\subsection{Hybrid Fold Is an Adjoint Fold}
\label{sec:hfoldProof1}

The first step of our proof is to show the unique existence of the solution
to (\ref{eq:hfoldEq}) based on the observation that it is an \emph{adjoint
fold equation} \cite{Hinze13Adj} with the adjunction between \emph{right Kan
extension} \cite{MacLane1978} and composition with $A$.
Hinze \cite{Hinze13Adj} shows the following theorem stating that adjoint fold equations
have a unique solution.

\begin{theorem}[Mendler-style Adjoint Folds \cite{Hinze13Adj}]\label{thm:men}
Given any adjunction $L \dashv R : \CatD \to \CatC$, an endofunctor $G : \CatD
\to \CatD$ whose initial algebra $\abracket{\mu G, \iniso}$ exists, and a
natural transformation $\Phi : \CatC(L -, B) \to \CatC(LD-, B)$ for some $B :
\CatC$, then there exists a unique $x : L(\mu G) \to B$ satisfying
\begin{equation}
x = \Phi_{\mu G}(x) \vcomp L\inOp
\end{equation}
and the unique solution satisfies $\radjunct{x} = \cata{\, \radjunct{\Phi_{RB}
(\epsilon_B)}\,}$ where $\radjunct{\cdot} : \CatC(LD, C) \to \CatD(D, RC)$ is 
the isomorphism for the adjunction $L \dashv R$.
\end{theorem}

Since $h : P A \to A$ and $P = \mu G : \CatEndof$, to apply this theorem to
(\ref{eq:hfoldEq}), we only need to (i) make $(- \hcomp A) : \CatEndof \to \CatCN$ a left
adjoint and (ii) make $[h_1, h_2, h_3]$ in (\ref{eq:hfoldEq}) an instance of
$\Phi_P(h)$ for some natural transformation $\Phi$:
\begin{itemize}
\item For (i), the functor $- \hcomp A$ is left adjoint to the \emph{right
Kan extension} along $A$, that is a functor $\Ran{A} : \CatCN \to \CatEndof$,
which always exists when $\CatC$ is lfp.

\item For (ii), we define a natural transformation $\Phi :  \CatCN(- \hcomp A, A) \to
\CatCN((G -) \hcomp A, A)$ such that for all $H : \CatEndof$ and $f \in \CatCN(H \hcomp A, A)$,
\[\Phi_H(f) = [f_1, f_2, f_3] \vcomp (\inOp \hcomp A)\] where 
\begin{gather}
f_1 = \big( \Id \hcomp A \xrightarrow{\identity} A \big) \label{eq:phiDef}
\hspace{2em}
f_2 = \big( \Sigma \hcomp H \hcomp A \xrightarrow{\Sigma \hcomp f} \Sigma \hcomp A \xrightarrow{a} A\big) \\
f_3 = \big( \Gamma \hcomp H \hcomp H \hcomp A \xrightarrow{\Gamma \hcomp H \hcomp f}
\Gamma \hcomp H \hcomp  A \xrightarrow{\Gamma \hcomp H \hcomp p} 
\Gamma \hcomp H \hcomp A \hcomp S \xrightarrow{\Gamma f}
\Gamma \hcomp A \hcomp S \xrightarrow{d} A \big) \nonumber
\end{gather}
It is immediate that \ensuremath{\Varid{hfold}} (\ref{eq:hfoldEq}) is exactly $h = \Phi_{P}(h) \vcomp (\inOp \hcomp A)$.
\end{itemize}
Then by \autoref{thm:men} we have the following result.

\begin{lemma}[Unique Existence of Hybrid Fold]\label{lem:hfoldExists}
The recursive definition $h : P A \to A$ (\ref{eq:hfoldEq}) of hybrid folds
(\autoref{fig:hfold}) has a unique solution:
\begin{equation}
\label{eq:hfoldByCata}
h = \big( PA \xrightarrow{\cata{ \radjunct{\Phi_{\Ran{A} A}(\epsilon_A)}}A}
  (\Ran{A}A) A \xrightarrow{\epsilon_A}
  A \big)
\end{equation}
where $\epsilon_A : (\Ran{A}A)\hcomp A \to A$ is the counit, and
$\cata{\radjunct{\Phi_{\Ran{A} A}(\epsilon_A)}}$ is the catamorphism from the
initial $G$-algebra $P$ (\ref{eq:e:recursive2}) to the $G$-algebra carried by $\Ran{A}A$ with structure
map $\radjunct{\Phi_{\Ran{A} A}(\epsilon_A)}$.
\end{lemma}

\subsection{Catamorphism as Interpretation}
\label{sec:hfoldProof2}
We have shown that $\ensuremath{\Varid{hfold}}_A$ for all indexed algebras $A$ uniquely exists,
and now to show its coincidence with $\interpName_A$ (the second part of
\autoref{thm:hfoldOK}), we show that the hybrid fold coincides with
$\interpName$ with some \emph{functorial} algebra, and then use the comparison
functor $\KK$ (\autoref{sec:FntoIx}) to translate this functorial algebra into
an indexed algebra.

The following lemma relating $G\Alg$ and $\Fn\Alg$ is straightforward.
\begin{lemma}\label{prop:g:as:fn}
For the functor $G$ in (\ref{eq:g:equation}) and every $X : \CatC$, there is a
faithful functor $J_X : G\Alg \to \Fn\Alg$ that maps any $G$-algebra $\tuple{H,
\alpha : G H \to H}$ to the functorial algebra $\tuple{H, H X, \alpha,  \beta
}$ where
$\beta = [\alpha_X \vcomp \iota_2, \alpha_X \vcomp \iota_3] : I_H(H X) \to H X$.
\end{lemma}

Now consider the $G$-algebra in \autoref{lem:hfoldExists} carried by
$\Ran{A}{A}$ with structure map $\alpha^G = \radjunct{\Phi_{\Ran{A} A}(\epsilon_A)}$,
we can trivially make it a functorial algebra $\hat{\alpha}$ by \autoref{prop:g:as:fn},
hoping for applications of the fusion law of functorial algebras
(\ref{lem:fn:fusion}):
\begin{equation}\label{eq:fn:alpha}
\hat{\alpha} = J_{A_0} \tuple{\Ran{A}{A}, \alpha^G} = \abracket{\Ran{A} A,
(\Ran{A} A) A_0, \alpha^G, \alpha^G \vcomp [\iota_2, \iota_3]}
\end{equation}
The catamorphism to $\alpha^G$ and $\interpName_{\hat{\alpha}}$ are related by
$\cata{\alpha^G}_{A_0} = \interp{\hat{\alpha}}(\alpha^G_{A_0} \vcomp \iota_1)$,
which can be shown by checking that the formula (\ref{eq:evalFctAlg}) for
computing $\interpName_{\hat{\alpha}}$ is exactly the defining equation of the
catamorphism $\cata{\alpha^G}_{A_0}$.
Plugging the identity  into (\ref{eq:hfoldByCata}), we obtain
\begin{equation}\label{eq:hIsInterpAndEps}
h_0 = (\epsilon_A)_0 \vcomp \interp{\hat{\alpha}}{(\alpha^G_{A_0} \vcomp \iota_1)}
\end{equation}
Now we have made some progress because the right-hand side takes exactly the
form for applications of the fusion law---an interpretation followed by a
morphism $(\epsilon_A)_0 : (\Ran{A} A)A_0 \to A_0$.

In order to apply the fusion law, we need to make $(\epsilon_A)_0 :
(\Ran{A} A)A_0 \to A_0$ a functorial algebra homomorphism from $\hat{\alpha}$
(\ref{eq:fn:alpha}).
With some exploration, we can find a functorial algebra
$\hat{\alpha'} = \abracket{\Ran{A} A,\ A_0,\ \alpha^G,\ \alpha^I}$ differing
from $\hat{\alpha}$ by the second and fourth components, where
\[
\alpha^I = [a_0,  d_0 \vcomp \Gamma (\epsilon_A)_0 \vcomp \Gamma (\Ran{A}A) p]
: \Sigma A_0 + \Gamma ((\Ran{A}{A}) A_0) \to A_0
\]
and $a$, $d$ and $p$ are the structure maps of indexed algebra $A$ (\ref{eq:fixed:a}).
It can be checked that $\abracket{\identity, (\epsilon_A)_0}$ is
a functorial algebra homomorphism from $\hat{\alpha}$ to $\hat{\alpha'}$.
Thus by \autoref{lem:fn:fusion}, we obtain
\begin{equation}\label{eq:hfoldbyF}
h_0 = \interp{\hat{\alpha}}{ \big((\epsilon_A)_0 \vcomp (\alpha^G)_{A_0} \vcomp \iota_1 \big) } 
= \interp{\hat{\alpha'}}{\identity} : PA_0 \to A_0
\end{equation}
which means that the hybrid fold coincides with the interpretation with functorial
algebra $\hat{\alpha'}$.

\subsection{Translating Back to Indexed Algebras}
\label{sec:hfoldProof3}
The last step of our proof is translating the functorial algebra $\hat{\alpha'}$ back
to an indexed algebra using comparison functor $\KK$ (\autoref{sec:FntoIx}), and
showing that the resulting indexed algebra induces the same interpretation morphism
as the one induced by $A$.

Recall that the comparison functor $\KK$ maps a functorial algebra carried by
$\abracket{H,X}$ to an indexed algebra carried by $i \mapsto H^i X$.
Thus $\KK \hat{\alpha'}$ is an indexed algebra carried by $i \mapsto (\Ran{A} A)^i A_0$
and by \autoref{lem:CFPreservesI} and (\ref{eq:hfoldbyF}), $\KK$ preserves the
induced interpretation:
\begin{equation}
h_0 = \interp{\hat{\alpha'}} \identity = \interp{\KK \hat{\alpha'}} \identity : PA_0 \to A_0
\end{equation}
What remains is to prove that $\interp{\KK \hat{\alpha'}} \identity = \interp{A} \identity$,
and we show this by the fusion law (of indexed algebras):
define a natural transformation $\tau : (\Ran{A} A)^{i} A_0 \to A_{i}$ 
between $\CatCN$ functors by
$\tau_0 = \identity : (\Ran{A} A)^0 A_0 \to A_0$ and
\[
\tau_{i+1} = \big( (\Ran{A} A) (\Ran{A}A)^{i} A_0 \xrightarrow{(\Ran{A}A)\tau_i} (\Ran{A} A)A_i
\xrightarrow{(\Ran{A}A) p}  (\Ran{A} A)A_{i+1} \xrightarrow{\epsilon_A} A_{i+1}  \big)
\]
and it can be checked that $\tau$ is an indexed algebra homomorphism from
$\KK \hat{\alpha'}$ to $\abracket{A, a, d, p}$.
Note that we have $\downharpoonleft \Ul_\Ix \tau$ equals $\identity : A_0 \to A_0$, 
and by \autoref{lem:CFPreservesI}, we have that
\[
h_0 = \interp{\KK \hat{\alpha'}} \identity = \identity \vcomp \interp{\KK \hat{\alpha'}} \identity
= (\downharpoonleft \Ul_\Ix) \tau \vcomp \interp{\KK \hat{\alpha'}} \identity
= \interp{A}\identity
\]
This completes our proof of \autoref{thm:hfoldOK} saying that Pir\'{o}g et al.\ \cite{PirogSWJ18}'s
\ensuremath{\Varid{hfold}} indeed correctly implements $\interpName$ with indexed algebras.

To summarise, we first connected \ensuremath{\Varid{hfold}} to a functorial algebra using a Mendler-style 
adjoint fold, and then used the fusion law of functorial algebras to simplify it.
Then it is translated to the category of indexed algebras, and we used the
fusion law one more time there.

\end{document}